\documentclass[letterpaper,onecolumn,12pt,leqno,accepted=2026-09-07]{quantumarticle}
\pdfoutput=1
\DeclareTextAccentDefault{\k}{T1}

\usepackage{amsmath,amsthm,amsfonts,amssymb,latexsym,verbatim,bbm,hyperref}
\usepackage[shortlabels]{enumitem}

\usepackage{subcaption}
\usepackage[capitalize]{cleveref}
\usepackage{afterpage}
\usepackage{mathrsfs} 
\usepackage[style=alphabetic,sorting=nyt,sortcites=false,maxalphanames=5,maxbibnames=9,backend=biber, eprint=false]{biblatex}
\usepackage[english]{babel}
\usepackage{csquotes}

\DeclareRobustCommand{\SkipTocEntry}[5]{}
\usepackage{thmtools}
\usepackage{thm-restate}

\allowdisplaybreaks

\setlist{itemsep=.5\baselineskip,topsep=.5\baselineskip}

\numberwithin{equation}{section}
\theoremstyle{plain}
\newtheorem{theorem}{Theorem}[section]
\newtheorem{lemma}[theorem]{Lemma}

\newtheorem{corollary}[theorem]{Corollary}

\newtheorem{question}[theorem]{Question}

\newtheorem{claim}[theorem]{Claim}

\theoremstyle{definition}

\newtheorem{definition}[theorem]{Definition}
\newtheorem{remark}[theorem]{Remark}

\newcommand{\arr}{\rightarrow}

\newcommand{\wtd}{\tilde}

\newcommand{\R}{\mathbb{R}}
\newcommand{\C}{\mathbb{C}}
\newcommand{\Z}{\mathbb{Z}}
\newcommand{\N}{\mathbb{N}}

\newcommand{\Id}{\mathbbm{1}}

\newcommand{\mcA}{\mathcal{A}}
\newcommand{\mcB}{\mathcal{B}}
\newcommand{\mcC}{\mathcal{C}}

\newcommand{\mcO}{\mathcal{O}}

\newcommand{\mcS}{\mathcal{S}}

\newcommand{\mcX}{\mathcal{X}}
\newcommand{\mcY}{\mathcal{Y}}

\newcommand{\mbF}{\mathbb{F}}

\newcommand{\mbR}{\mathbb{R}}

\newcommand{\msA}{\mathscr{A}}

\newcommand{\abs}[1]{\lvert #1 \rvert}

\newcommand{\norm}[1]{\lVert #1 \rVert}

\usepackage{tikz-cd}
\usepackage{braket}
\usepackage{arydshln}
\usepackage{chngcntr}
\usepackage{apptools}
\AtAppendix{\counterwithin{theorem}{section}}

\DeclareMathOperator{\tr}{tr}

\newcommand{\Tr}{\mathrm{Tr}}

\newcommand{\dsync}{\delta_{\mathrm{sync}}}
\newcommand{\poly}{\mathrm{poly}}

\newcommand{\tA}{\tilde{A}}
\newcommand{\tB}{\tilde{B}}

\newcommand{\tH}{\tilde{H}}
\newcommand{\kp}{\ket{\psi}}

\newcommand{\kpz}{\ket{\psi_0}}
\newcommand{\kpo}{\ket{\psi_1}}
\newcommand{\kpt}{\ket{\tilde{\psi}}}

\newcommand{\expect}[1]{\underset{#1}{\mathbb{E}}}
\newcommand{\clmprvd}[1]{\\ \null\hfill \emph{This proves Claim #1.}}

\newcommand{\bfa}{\mathbf{a}}
\newcommand{\bfb}{\mathbf{b}}
\newcommand{\bfc}{\mathbf{c}}
\newcommand{\bfkappa}{\boldsymbol{\kappa}}
\newcommand{\bflambda}{\boldsymbol{\lambda}}
\newcommand{\bfmu}{\boldsymbol{\mu}}

\DeclareMathOperator{\latRe}{Re}

\newcommand{\id}{\mathrm{id}}
\newcommand{\qldt}{\mathrm{qldt}}

\makeatletter
\newcommand{\OriginalVerticalSpacing}{%
  \bigskipamount=.7\baselineskip plus.7\baselineskip
  \medskipamount=\bigskipamount \divide\medskipamount by 2
  \smallskipamount=\medskipamount \divide\smallskipamount by 2
  \abovedisplayskip=\medskipamount
  \belowdisplayskip=\abovedisplayskip
  \abovedisplayshortskip=\abovedisplayskip
  \advance\abovedisplayshortskip by -\abovedisplayskip
  \belowdisplayshortskip=\abovedisplayshortskip
  \advance\belowdisplayshortskip by \smallskipamount
  \jot=\baselineskip \divide\jot by 4}
\renewcommand{\normalsize}{\@setfontsize\normalsize{12}{14}%
  \OriginalVerticalSpacing\let\@listi\@listI}
\renewcommand{\small}{\@setfontsize\small{10.95}{13}\OriginalVerticalSpacing}
\renewcommand{\footnotesize}{\@setfontsize\footnotesize{10}{12}\OriginalVerticalSpacing}
\newcommand{\OriginalRunInHeading}[1]{#1.}
\renewcommand{\@seccntformat}[1]{\textup{\mdseries\csname the#1\endcsname.}\enspace}
\renewcommand{\section}{\@startsection{section}{1}{\z@}%
  {9.8pt plus14pt}{7pt}{\normalfont\normalsize\scshape\centering}}
\renewcommand{\subsection}{\@startsection{subsection}{2}{\z@}%
  {7pt plus9.8pt}{-.5em}{\normalfont\normalsize\bfseries\OriginalRunInHeading}}
\renewcommand{\subsubsection}{\@startsection{subsubsection}{3}{\z@}%
  {7pt plus9.8pt}{-.5em}{\normalfont\normalsize\itshape\OriginalRunInHeading}}
\patchcmd{\@maketitle}{\huge\hyphenpenalty=50000}{\huge\raggedright\hyphenpenalty=50000}{}{}
\makeatother
\AtBeginBibliography{\emergencystretch=2em}
\AtBeginDocument{%
  \setlength{\parindent}{12pt}%
  \setlength{\parskip}{0pt}%
  \normalsize
  \DeclareFieldFormat{doi}{\mkbibacro{DOI}\addcolon\space
    \href{https://doi.org/#1}{\nolinkurl{#1}}}%
}

\title{Lifting the maximally-entangledness assumption in robust self-testing for synchronous games}
\author{Matthijs Vernooij}
\email{m.n.a.vernooij@tudelft.nl}
\affiliation{Delft Institute of Applied Mathematics, TU Delft, The Netherlands}
\author{Yuming Zhao}
\email{yuming@math.ku.dk}
\affiliation{QMATH, Department of Mathematical Sciences, University of Copenhagen, Denmark}
\date{}

\begin{document}

\begin{abstract}
    Robust self-testing in non-local games allows a classical referee to certify that two untrustworthy players are able to perform a specific quantum strategy up to high precision. Proving robust self-testing results becomes significantly easier when one restricts the allowed strategies to symmetric projective maximally entangled (PME) strategies, which allow natural descriptions in terms of tracial von Neumann algebras. This has been exploited in the celebrated MIP*=RE paper and related articles to prove robust self-testing results for synchronous games when restricting to PME strategies. However, the PME assumptions are not physical, so these results need to be upgraded to make them physically relevant. In this work, we do just that: we prove that any perfect synchronous game which is a robust self-test when restricted to PME strategies, is in fact a robust self-test for all strategies. We then apply our result to the Quantum Low Degree Test to find an efficient $n$-qubit test.
\end{abstract}

\maketitle

\section{Introduction}
In a non-local game $G$, two cooperating but distant players respond to questions drawn from a known distribution to satisfy a known winning condition determined by a referee. A \textbf{quantum strategy} allows the players to share an entangled state and perform local measurements, often leading to a higher winning probability than classically achievable. Remarkably, certain non-local games exhibit an even stronger guarantee: they admit a unique optimal quantum strategy, making it possible to certify the underlying quantum state and measurements solely from the observed statistics. This is the essence of \textbf{self-testing}, a concept whose roots trace back to foundational work by Summers and Werner~\cite{SW87}, Popescu and Rohrlich~\cite{PR92}, and Tsirelson~\cite{Tsi87}, and was later introduced by Mayers and Yao~\cite{MY03} from a cryptographic perspective.  

Self-testing is arguably the strongest form of device-independent quantum certification, where one aims to classically verify the behaviour of a quantum device without making any assumptions about its internal workings. This idea has found impactful applications in device-independent 
cryptography~\cite{BSCA18a,BSCA18b}, verifiable quantum delegation~\cite{RUV13,CGJV19,BMZ24}, and quantum complexity theory, including the recent breakthrough $\text{MIP}^*=\text{RE}$~\cite{JNVWY22}.

Due to inevitable noise and imperfections in real-world implementations, the strategy executed in an experiment may only win the given non-local game near-optimally. As a result, all the aforementioned applications require self-tests to be \textbf{robust}: any near-optimal strategy must be close---under a suitable notion of distance---to the unique optimal strategy. We say that a game $G$ \textbf{$\kappa$-robustly self-tests} an ideal optimal strategy $\wtd{\mcS}$ for a class of employed strategies $\mcC$ if every $\epsilon$-optimal strategy in $\mcC$ is $\kappa(\epsilon)$-close (up to local isometries) to $\wtd{\mcS}$. Here, $\kappa:\R_{\geq 0}\arr\R_{\geq 0}$ satisfies $\kappa(\epsilon)\arr 0$ as $\epsilon\arr 0$ and is called the \textbf{robustness} of this self-test. 

To mathematically prove that a game is a robust self-test, one typically imposes additional assumptions on the class $\mcC$ of employed strategies. These assumptions simplify the analysis because the resulting strategies often admit nice algebraic forms. For instance,
the existing self-testing results typically assume that the employed strategies $\mcS=(\ket{\psi},\{A^x_a\},\{B^y_b\})$ are \textbf{projective}, meaning that the players' measurement operators $\{A^x_a\},\{B^y_b\}$ are projection-valued measures (PVMs). In this setting, strategies for a game with $n$ questions and $m$ answers correspond to representations of the group algebra $\C[\Z_{m}^{*n}\times\Z_{m}^{*n}]$. Another common assumption is that the employed strategies $\mcS=(\ket{\psi},\{A^x_a\},\{B^y_b\})$ are \textbf{full-rank}, in the sense that the shared entangled state $\ket{\psi}$ has full Schmidt rank.
Such strategies are centrally-supported~\cite{PSZZ23}, so for every measurement operator $A^x_a$ of Alice, there exists some operator $\hat{A}^x_a$ acting on Bob's registers such that $A^x_a\otimes \Id\ket{\psi}=\Id\otimes \hat{A}^x_a\ket{\psi}$.
However, from a device-independent perspective, such assumptions limit the scope of security and soundness guarantees provided by self-testing. For example, in quantum key distribution, an adversarial device might deviate from the assumed behaviour, potentially compromising the protocol; in quantum interactive proof systems, these assumptions may fail to capture the behaviour of general malicious provers, thereby weakening soundness. This gives rise to a trade-off: stronger assumptions make robust self-testing results easier to prove but less generally applicable. A central question is, therefore, how to lift such assumptions while preserving the robustness of self-tests. 

Recent progress has been made on this front. In particular, \cite{PSZZ23} establishes that for binary output games and synchronous games, self-testing for projective strategies implies self-testing for general POVM strategies. This was later extended in \cite{BCKLMNS23}, where the authors show that the projectivity assumption and the full-rankness assumption can both be lifted in robust self-testing, provided that the non-local game has an optimal strategy that is simultaneously projective and full-rank.

In this work, we focus on \textbf{synchronous games}~\cite{PSSWT16} and robust self-testing for their \textbf{perfect} quantum strategies. Such games exhibit a rich algebraic structure. To each synchronous game $G$ one can associate a $*$-algebra $\mcA(G)$---known as the synchronous algebra---whose tracial states correspond to perfect strategies within different mathematical models for entanglement~\cite{KPS18,HMPS19}. This algebraic framework plays a central role in the recent developments connecting quantum interactive proofs to operator algebras~\cite{JNVWY22,MNY22,NZ23}. 

One of the essential building blocks of synchronous games is the class of \textbf{PME} strategies, where the players share a \textbf{maximally entangled} state and perform \textbf{projective}, \textbf{symmetric} measurements. It is well-known that any perfect strategy for a synchronous game is a convex mixture of PME strategies~\cite{PSSWT16}, and recent results \cite{Vid22,Pad25} demonstrate that any near-perfect strategy is close to a convex mixture of PME strategies. This has been further extended to infinite-dimensional models~\cite{MS23,Lin24}. PME strategies are also particularly tractable in the context of self-testing: if we assume employed strategies are PME, then the robustness is closely related to the stability of the synchronous algebra in the normalised Hilbert-Schmidt norm. Therefore, it is often more straightforward to prove that a synchronous game is a robust self-test for PME strategies, using techniques from approximate representation theory (see e.g., \cite{CVY23}). However, perfectly maximally entangled states are not physically realisable in practice. Meanwhile, noise on maximally entangled states can significantly reduce the power of quantum interactive proof systems~\cite{QY21,QY23,NoisyEPR}. This motivates our central question about lifting the PME assumption in synchronous self-testing.

\begin{question}\label{question}
    If a synchronous game robustly self-tests a perfect strategy $\wtd{\mcS}$ for PME strategies, does it follow that it robustly self-tests $\wtd{\mcS}$ for general POVM strategies?
\end{question}

\subsection{Main results}
We answer the above question affirmatively and quantitatively.
\begin{theorem}\label{thm}
    Let $G$ be a synchronous game. If $G$ $\kappa$-robustly self-tests a perfect quantum strategy $\wtd{\mcS}$ for PME strategies, then $G$ $\kappa'$-robustly self-tests $\wtd{\mcS}$ for general POVM strategies, where $\kappa'$ is polynomially related to $\kappa$.
\end{theorem}

The precise relationship\footnote{We remark that the results in \cite{Zha24} and \cite{Kar25} show that for any perfect synchronous game $G$, both robust self-testing for PME strategies and robust self-testing for general strategies are equivalent to the uniqueness of amenable tracial state on the synchronous algebra $\mcA(G)$. However, their results do not establish a quantitative relationship between the robustness of the two cases.} between $\kappa$ and $\kappa'$ is established in \Cref{cor:PME-robust-self-testing-implies-robust-self-testing}. Crucially, this relationship is independent of the size (i.e., the number of questions and answers) of the game; it depends only on the \textbf{synchronicity} of $G$, which quantifies how frequently the referee checks the players' consistency by sending them the same question.

This notion of synchronicity is also closely tied to another fundamental aspect of robust self-testing: the probability distribution over question pairs. A non-local game is defined with respect to such a distribution $\nu$, determining how often each type of question is asked. Separately, when evaluating the closeness between an employed strategy and the ideal optimal strategy, a second distribution $\hat{\nu}$ can be used to weight the ``distance" (see \Cref{def:local_dilation}). Historically, most robust self-testing results assume that the game distribution $\nu$ and the distance distribution $\hat{\nu}$ coincide. However, in many natural scenarios---especially when certain questions only exist to enforce the desired structure of optimal strategies on the other questions---it is useful to consider the more general case where $\nu\neq \hat{\nu}$.

To capture such situations, we define a game to be a \textbf{$(\kappa,\hat{\nu})$-robust self-test} if every $\epsilon$-optimal strategy (with respect to the game distribution $\nu$) is $\kappa(\epsilon)$-close to the ideal strategy with respect to $\hat{\nu}$ (see \Cref{def:selftest}). Most of the results in this paper are stated and proved with respect to this generalised notion of robustness. This framework is particularly useful in the analysis 
of the Quantum Low Degree Test \cite{CVY23}, which is one of the key ingredients in the MIP$^*$=RE proof. Within this generalised notion of robust self-testing, we show that this test can be used to verify that the players approximately have access to maximally entangled qubits and Pauli operators acting on those qubits.

\begin{theorem}[Precise statement in \Cref{cor:QLDT}]\label{thm2}
    Performing the Quantum Low Degree Test and a synchronicity test with equal probability $\kappa'$-robustly self-tests the maximally entangled state on $n$ qubits together with a generating set of Pauli operators on those qubits. Here $\kappa'$ depends polynomially on $\kappa$, the robustness of the Quantum Low Degree Test restricted to maximally entangled states. 
\end{theorem}

\subsection{Comparison to other results}
A well-informed reader may wonder whether our main result, \Cref{thm}, follows directly from Vidick's result~\cite{Vid22}, which shows that an almost synchronous  strategy $\mcS$ can be approximated by a convex combination of PME strategies $\mcS_\lambda,\lambda\in\R_{+}$, each with their own state $\ket{\psi_{\lambda}}$. More precisely, after decomposing the reduced density $\rho$ used in $\mcS$ into spectral blocks $\rho_\lambda$, one obtains PME strategies $\mcS_{\lambda}$ on these blocks whose measurements (on average) approximate the original measurements with respect to the state $\ket{\psi_{\lambda}}$.

This result suggests a natural roadmap for proving our theorem. Suppose the synchronous game is a PME robust self-test with ideal perfect strategy $\tilde{S}$. Given a near-perfect strategy $\mcS$, one could try to apply PME robust self-testing to each PME block $\mcS_\lambda$, and then use the closeness of each $\mcS_{\lambda}$ to $\mcS$ to conclude general robust self-testing.

However, this approach faces several obstructions. For a robust self-test, closeness of measurement operators is measured with respect to the employed state $\ket{\psi}$. Vidick's result guarantees closeness of measurement operators with respect to $\ket{\psi_{\lambda}}$, but this is not sufficient to guarantee closeness with respect to $\ket{\psi}$ in a vacuum. This is a problem particular to robust self-testing; Vidick's result is well suited for transferring correlation estimates, consistency estimates, or algebraic relations from PME strategies to almost synchronous strategies, as illustrated in \cite[Section 4]{Vid22} as ``rigidity" properties. A second, more technical obstruction is that robust self-testing requires closeness up to a single pair of local isometries. PME robust self-testing gives you a good pair of local isometries for each PME strategy $\mcS_{\lambda}$, but combining them is arguably not straightforward. 

Nonetheless, Vidick's result is an important ingredient in our proof, and much of Section \ref{sec4} is devoted to overcoming the obstructions described above. Below we outline our proof approach and technical contributions.

\subsection{Proof approach and technical contributions}
Our main result---that robust self-testing for PME strategies implies robust self-testing for general strategies---is proved in three steps, corresponding to \Cref{sec3}, \Cref{sec4}, and \Cref{sec5}, respectively.

We begin by formulating PME strategies---more generally, any strategy employing a maximally entangled state---using tracial von Neumann algebras. We adapt the notion in \cite{CVY23} of distance between families of unitaries to define a `\textbf{von Neumann distance}' between PME strategies (see \Cref{def:vNA-local-dilation}). Our first main technical contribution is to show that this von-Neumann distance is equivalent, up to a constant-factor trade-off, to the standard distance defined in the Hilbert-space formalism (\Cref{lem:PME-robust-self-testing-vNA}). Therefore, we can reformulate robust self-testing for PME strategies in the von Neumann algebraic language. This algebraic formulation allows us to exploit the symmetry of PME strategies by working with the algebra generated by a single player's measurements and reduced state, and it enables us to connect our results to the main result from \cite{CVY23}.

As an intermediate step toward proving \Cref{thm}, our second main technical contribution is to show that $\kappa'$, the robustness for general strategies, is controlled by the robustness $\kappa$ for PME strategies and the \textbf{spectral gap} of the ideal perfect strategy $\wtd{\mcS}$ (\Cref{thm:lifting-PME-assumtion-if-spec-gap}). Here, $\wtd{\mcS}$ having spectral gap $\Delta$ means that any strategy $\mcS$ that uses the same measurements as $\wtd{\mcS}$ but employs a quantum state orthogonal to the maximally entangled state in $\wtd{\mcS}$ can achieve a winning probability of up to $1-\Delta$. Intuitively, such a strategy $\mcS$ is ``far from" $\wtd{\mcS}$, so a small spectral gap indicates poor robustness in self-testing. 

Our last step is to show that, perhaps surprisingly, if a synchronous game $G$ $\kappa$-robustly self-tests an ideal perfect strategy $\wtd{\mcS}$ for PME strategies, then the spectral gap of $\wtd{\mcS}$ admits a lower bound that is polynomially related to $\kappa$ (\Cref{thm:auto-spec-gap}). As a result, in \Cref{thm}, the robustness $\kappa'$ is controlled by---indeed, polynomially related to---the original robustness $\kappa$. In the case of Quantum Low Degree Test based on a linear code of relative distance $d$, we explicitly compute its spectral gap to be $d/2$ (\Cref{thm:QLDT-spec-gap}), which further yields our second main result \Cref{thm2}.

\subsection{Acknowledgement} MV is supported by the NWO Vidi grant VI.Vidi.192.018 `Non-commutative harmonic analysis and rigidity of operator algebras'. YZ is supported by VILLUM FONDEN via QMATH Centre of Excellence grant number 10059, Villum Young Investigator grant number 37532, and an ERC grant (QInteract, Grant No 101078107). The authors would like to thank Martijn Caspers for reading an earlier version of this work and identifying some inaccuracies. The authors would also like to thank Michael Brannan, William Slofstra and the Institute for Quantum Computing in Waterloo for allowing MV to do a research visit there, where this project was conceived.

\section{Preliminaries}

\subsection{Non-local games and strategies} \label{sec:non-local-games-and-strategies}

In this paper, given a Hilbert space $H$, we denote by $B(H)$ the algebra of bounded operators on $H$. We write $\Id_H$ for the identity operator on $H$, and simply $\Id$ if the underlying space is clear from the context. We also use $\Id$ for the identity element of a von Neumann algebra. A collection of operators $\{P_i\}_{i=1}^k\subset B(H)$ is a positive operator-valued measure (POVM) if every $P_i\geq 0$ and $\sum_{i=1}^k P_i=\Id$.  

A two-player (commonly called Alice and Bob) non-local game $G$ is specified by a tuple $G=(\mcX,\mcY,\nu,\mcA,\mcB,D)$, where $\mcX,\mcY,\mcA$, and $\mcB$ are finite sets, $\nu$ is a probability distribution on $\mcX\times\mcY$, and $D:\mcX\times\mcY\times\mcA\times\mcB\arr\{0,1\}$ is a predicate. Alice and Bob know all the data in $G$ and they can strategise together before the game begins, but they are not allowed to communicate once the game starts. During the game, Alice and Bob receive questions $x\in\mcX$ and $y\in\mcY$ respectively from a referee with probability $\nu(x,y)$, and they return answers $a\in\mcA$ and $b\in\mcB$ respectively. Based on the predicate $D$, the referee determines whether they win ($D(a,b|x,y)=1$) or lose ($D(a,b|x,y)=0$). In some cases, the sets of feasible answers are determined by questions. When this happens, we think of $\mcA$ and $\mcB$ as collections $\mcA=\{\mcA(x):x\in\mcX\}$ and $\mcB=\{\mcB(y):y\in \mcY  \}$), where for each question $x\in \mcX$ and $y\in \mcY$, Alice and Bob can only return some $a\in \mcA(x)$ and $b\in \mcB(y)$ respectively.

In quantum mechanics, the strategy Alice and Bob employ for a non-local game $G=(\mcX,\mcY,\nu,\mcA,\mcB,D)$ is described by a tuple
\begin{equation*}
    \mcS=\big(\ket{\psi}\in H_A\otimes H_B,A=\{A^x_a\},B=\{B^y_b\}\big)
\end{equation*}
where 
\begin{enumerate}[(i)]
    \item Alice and Bob share a quantum state (unit vector) $\ket{\psi}\in\ H_A\otimes H_B$, and 
    \item for each $x\in\mcX$ (resp. $y\in Y$) Alice measures her register using a POVM $\{A^x_a:a\in \mcA\}$ on $H_A$ (resp. Bob measures his register using a POVM $\{B^y_b:b\in \mcB\}$ on $H_B$), and we shorten $\{\{A_a^x\}_{a\in\mcA}|x\in\mcX\}$ to $\{A_a^x\}$ (resp. $\{\{B_b^y\}_{b\in\mcB}|y\in\mcY\}$ to $\{B_b^y\}$).
\end{enumerate}

 By Born's rule, if the players employ a strategy $\mcS=(\ket{\psi}, A,B)$, then the probability that they response $a\in \mcA$ and $b\in \mcB$ upon receiving $x\in \mcX$ and $y\in \mcY$ is given by
\begin{equation}
    C_{x,y,a,b}=\bra{\psi}A^x_a\otimes B^y_b\ket{\psi}.\label{eq:correlation}
\end{equation}
The collection $C=\{C_{x,y,a,b}\}\in \R^{\mcX\times\mcY\times\mcA\times\mcB}$ is called the correlation induced by $\mcS$. 
The winning probability of $C$ for $G=(\mcX,\mcY,\nu,\mcA,\mcB,D)$ is given by 
\begin{equation*}  \omega(G;C)=\expect{(x,y)\sim\nu}\sum_{a,b}D(a,b|x,y)C_{x,y,a,b},
\end{equation*}
where $\expect{(x,y)\sim\nu}\cdot=\sum_{x,y}\nu(x,y)\ \cdot$ is the expectation with respect to $\nu$.
The winning probability $\omega(G;\mcS)$ of a strategy $\mcS$ for a game $G$ refers to the winning probability of the correlation induced by $\mcS$. When the game $G$ is clear from the context, we just write $\omega(\mcS)$ for $\omega(G;\mcS)$. In this paper, we assume all strategies employ \textbf{finite-dimensional} systems. We denote by $C_q(\mcX,\mcY,\mcA,\mcB)$ the set of correlations induced by strategies through \Cref{eq:correlation}. The quantum value $\omega_q(G)$ of a non-local game $G$ is the supremum of $w(G;C)$ over all $C\in C_q$. A strategy $\mcS$ or a correlation $C$ is said to be optimal for $G$ if its winning probability achieves the quantum value $w_q(G)$. When $w_q(G)=1$, we replace ``optimal" with ``perfect". We call a game perfect if it admits a perfect (finite-dimensional) strategy.

Given a strategy $\mcS=(\ket{\psi},A,B)$, we refer to the element
\begin{equation*}
    T_{G,\mcS}:=\expect{(x,y)\sim\nu}\sum_{a,b}D(a,b|x,y)A^x_a\otimes B^y_b
\end{equation*}
in $B(H_A\otimes H_B)$ as the \textbf{game polynomial} of $\mcS$ for $G$. The winning probability of any strategy $\mcS'=(\ket{\psi'},A,B)$ for $G$ is equal to $\bra{\psi'}T_{G,\mcS}\ket{\psi'}$. We will be particularly interested in the \textbf{spectral gap} of $T_{G,\mcS}$, which is the difference between the largest and second largest eigenvalues of $T_{G,\mcS}$, \emph{accounting for multiplicities}. This means that any self-adjoint operator has spectral gap equal to zero if the largest eigenvalue has multiplicity larger than one.
 
If a POVM $\{P_i:1\leq i\leq k\}$ on $H$ consists of mutually orthogonal projections in the sense that $P_i^2=P_i$ and $P_iP_j=0$ for $i\neq j$, then it is called a PVM. Any $k$-outcome PVM $\{P_i:1\leq i\leq k\}$ corresponds a unitary $U$ of order $k$ via the Fourier transform 
\begin{equation*}
    U=\sum_{j=1}^k \exp\left(\frac{2\pi\sqrt{-1}}{k}j  \right)P_j
\end{equation*}
and vice versa via spectral decomposition. A strategy $\mcS=(\ket{\psi},A,B)$ is said to be \textbf{projective} if $\{A^x_a,a\in \mcA\}$ and $\{B^y_b,b\in\mcB\}$ are PVMs for all $x\in \mcX,y\in\mcY$. In this case, we often specify the measurement operators $\{A^x_a:a\in\mcA\}$ and $\{B^y_b:b\in\mcB\}$ using their corresponding unitaries $U(A^x)$ and $U(B^y)$.

Given two finite-dimensional Hilbert spaces $H_A$ and $H_B$, every vector $\ket{\psi}\in H_A\otimes H_B$ has a Schmidt decomposition 
\begin{equation*}
    \ket{\psi}=\sum_{i=1}^k \lambda_i\ket{\alpha_i}\otimes\ket{\beta_i}
\end{equation*}
where the Schmidt coefficients $\lambda_i$'s are positive real numbers, and $\{\ket{\alpha_i}:1\leq i\leq k\}$ and $\{\ket{\beta_i}:1\leq i\leq k\}$ are orthonormal subsets of $H_A$ and $H_B$ respectively. A unit vector $\ket{\psi}\in H_A\otimes H_B$ is \textbf{maximally entangled} if $\dim(H_A)=\dim(H_B)=:d$ and $\ket{\psi}$ has a Schmidt decomposition
\begin{equation*}
    \ket{\psi}=\sum_{i=1}^d\frac{1}{\sqrt{d}}\ket{\alpha_i}\otimes\ket{\beta_i}.
\end{equation*}
A strategy $\mcS=(\ket{\psi},A,B)$ is said to be \textbf{maximally entangled} if $\ket{\psi}\in H_A\otimes H_B$ is a maximally entangled state.

If $\mcX=\mcY$ and $\mcA=\mcB$ in a non-local game $G=(\mcX,\mcY,\nu,\mcA,\mcB,D)$, then we write $G=(\mcX,\nu,\mcA,D)$. It is often convenient to work with symmetric games and symmetric strategies. A non-local game $G=(\mcX,\nu,\mcA,D)$ is \textbf{symmetric} if 
\begin{enumerate}[(i)]
    \item $\nu(x,y)=\nu(y,x)$, and
    \item $D(a,b|x,y)=D(b,a|y,x)$ for all $a,b\in\mcA$ and $x,y\in\mcX$.
\end{enumerate}
A strategy $\mcS=(\ket{\psi},A,B)$ for a non-local game $G=(\mcX,\nu,\mcA,D)$ is \textbf{symmetric} if 
\begin{enumerate}[(i)]
    \item $H_A=H_B:=H$,
    \item $\ket{\psi}=\sum_i \sqrt{\lambda_i}\ket{u_i}\otimes\ket{u_i}$ where $\lambda_i\geq 0$ for and $\{\ket{u_i}\}_i$ is an orthonormal basis for $H$, and
    \item $A^x_a=(B^x_a)^T$ for all $a\in\mcA$ and $x\in\mcX$, where the transpose is taken with respect to the basis $\{\ket{u_i}\}$.
\end{enumerate}
Note that in this case, the reduced density matrix $\rho$ of $\ket{\psi}$ on Alice and Bob's sides are both $\sum_i\lambda_i\ket{u_i}\bra{u_i}$, and we have 
\begin{equation*}
    \bra{\psi}S\otimes T\kp=\Tr(S\rho^{1/2}T^{T}\rho^{1/2})
\end{equation*}
for $S,T\in B(H)$, which is called \emph{Ando's formula}. Since Bob's measurements are completely determined by Alice's measurements, we write any symmetric strategy as $\mcS=(\ket{\psi},A)$. Given any strategy $\mcS=(\ket{\psi},A,B)$ with Schmidt decomposition given by $\ket{\psi}=\sum_{i=1}^d\lambda_i\ket{\alpha_i}\otimes\ket{\beta_i}$, we define its \textbf{associated symmetric strategies} $\mcS_A=(\ket{\psi_A},A)$ and $\mcS=(\ket{\psi_B},B)$ where $\ket{\psi_A}=\sum_{i=1}^d\lambda_i\ket{\alpha_i}\otimes\ket
{\alpha_i}$ and $\ket{\psi_B}=\sum_{i=1}^d\lambda_i\ket{\beta_i}\otimes\ket{\beta_i}$.

Let $\mcS=(\ket{\psi},A)$ be a symmetric strategy. If $\ket{\psi}$ is maximally entangled, we call $\mcS$ an \textbf{ME} strategy. If an ME strategy is projective, we call it a \textbf{PME} strategy. Note that, given a symmetric strategy $\mcS=(\kp,A)$, Ando's formula tells us that choosing a different symmetric state $\ket{\psi'}$ with the same reduced density matrix gives rise to an equivalent strategy $\mcS'=(\ket{\psi'},A)$, in the sense that they are related through a unitary on Bob's side. This may seem strange at first, but this happens because the transpose depends on the state, so the operators for Bob will also change if one changes the state. In particular, this means that for any two maximally entangled state $\kp$ and $\ket{\psi'}$, the strategies $(\kp,A)$ and $(\ket{\psi'},A)$ are equivalent.

In this paper, we focus on \textbf{synchronous games} and \textbf{synchronous correlations}. A symmetric game $G=(\mcX,\nu,\mcA,D)$ is synchronous if $\nu(x,x)>0$ and $D(a,a'|x,x)=0$ for all $x$ and $a\neq a'$. Given $\beta\in (0,1)$, we say a game $G=(\mcX,\nu,\mcA,D)$ is $\beta$-synchronous if it is synchronous and
\begin{equation*}
    \nu(x,x)\geq \beta \sum_{y
    \in \mcX} \nu(x,y)
\end{equation*}
for all $x\in\mcX$. A correlation $C\in C_q(\mcX,\mcX,\mcA,\mcA)$ is said to be synchronous if $C_{x,x,a,a'}=0$ for all $x$ and $a\neq a'$. We use $C_q^s(\mcX,\mcA)$ to denote the set of synchronous correlations in $C_q(\mcX,\mcX,\mcA,\mcA)$. As shown in \cite{PSSWT16}, a correlation $C\in C_q(\mcX,\mcX,\mcA,\mcA)$ is synchronous if and only if there is a finite-dimensional unital tracial $C^*$-algebra $(\msA,\tau)$ and PVMs $\{E^x_{a}:a\in\mcA\},x\in\mcX$ in $\msA$ such that
\begin{equation*}
    C_{x,y,a,b}=\tau(E^x_aE^y_b)
\end{equation*}
for all $x,y\in\mcX$ and $a,b\in\mcA$.

Any ME strategy $\mcS=(\ket{\psi},A)$ induces a tracial state $\tau$ on the algebra generated by $A^x_a$'s via $\tau(\alpha)=\bra{\psi}\alpha\otimes\Id\ket{\psi}$. Furthermore, $\bra{\psi}A^x_a\otimes B^y_b\ket{\psi}=\tau(A^x_aA^y_b)$. This gives an alternative way to describe ME and PME strategies in terms of tracial von Neumann algebras. 

A tracial von Neumann algebra $(M,\tau)$ is a von Neumann algebra $M$ together with a normal faithful tracial state $\tau$ on $M$. The corresponding trace-norm (or so called $2$-norm) $\norm{\cdot}_\tau$ on $M$ is given by $\norm{\alpha}_\tau:=\sqrt{\tau(\alpha^*\alpha)}$. For example, $(M_n(\C),\tr)$ is the von Neumann algebra of $n\times n$ matrices with the normalised trace $\tr(\alpha)=\frac{1}{n}\Tr(\alpha)$. We also work with Schatten $p$-norm in $M_n(\C)$ for $p\in [1,\infty]$. In particular, $\norm{X}_1=\Tr(\abs{X})$, $\norm{X}_2=\sqrt{\Tr(X^*X)}$, and $\norm{X}_{\infty}$ is the largest singular value of $X$.

Any ME (resp. PME) strategy $\mcS$ for a game $G=(\mcX,\nu,\mcA,D)$ can be specified by a tuple $(M,\tau,A)$ where 
$M= B(H)$ for some finite-dimensional space $H$, $\tau$ is the normalized trace on $M$, and $A=\{\{A^x_a\}_{a\in\mcA}\subset M | x\in\mcX\}$ are POVMs (resp. PVMs) in $M$. The correlation induced by $\mcS=(M,\tau,A)$ is given by 
\begin{equation*}
    C_{x,y,a,b}=\tau(A^x_aA^y_b).
\end{equation*}
In fact, any such triple $\mcS=(M,\tau,A)$ gives rise to a strategy in this way, even if $M$ is a finite dimensional von Neumann algebra that is not of the form $B(H)$ for some finite-dimensional Hilbert space $H$. One can recover the usual formulation of a strategy $\mcS=(\kp,A',B')$ on $H_A\otimes H_B$ using the GNS construction. Let $M\subset B(H)$. Then $H_A=H_B=H$, $A'=A$, $B'=A^T$ and $\kp$ is the cyclic vector corresponding to the GNS construction by identifying $B(H)$ (equipped with the Hilbert-Schmidt inner product) with $H\otimes H$. Since $M\subset B(H)$, the cyclic vector becomes a vector in $H\otimes H$. Note that while $H\otimes \overline{H}$ and $B(H)$ are naturally isomorphic, $H$ and $\overline{H}$ are not. One needs to specify a basis to identify $B(H)$ and $H\otimes H$ and to take the transpose in $B'=A^T$. Consequently, a triple $\mcS=(M,\tau, A)$ defines $\mcS=(\kp,A',B')$ up to unitary equivalence on Bob's side. We call the state $\kp$ obtained in this way a GNS state for $(M,\tau)$, and it satisfies
\begin{equation*}
    \tau(X)=\bra{\psi}(X\otimes \Id)\kp=\bra{\psi}(\Id\otimes X^T)\kp
\end{equation*}
for all $X\in M$.

Note that not all strategies $\mcS=(\kp,A,B)$ can be written in von Neumann algebra terms. A strategy is of the form $(M,\tau,A)$ if and only if it is a classical convex combination of ME strategies, intuitively meaning that each round an ME strategy is selected based on classical shared randomness.

\begin{definition}[Local dilation]\label{def:local_dilation}
    For two strategies $\mcS=(\ket{\psi},A,B)$ and $\wtd{\mcS}=(\ket{\wtd{\psi}},\wtd{A},\wtd{B})$, we say that $\wtd{\mcS}$ is a local $(\epsilon,\nu)$-dilation of $\mcS$ for some $\epsilon\geq 0$ and distribution $\nu$ on $\mcX\times \mcY$ if there are isometries $V_A:H_A\arr \wtd{H}_A\otimes K_A$ and $V_B:H_B\arr\wtd{H}_B\otimes K_B$ and a unit vector $\ket{aux}\in K_A\otimes K_B$ such that
    \begin{align*}
    \norm{(V_A\otimes V_B)\ket
    \psi-\ket{\wtd{\psi}}\otimes\ket{aux}}&\leq \epsilon,\\
        \left(\expect{x\sim \nu_A}\sum_{a}  \norm{(V_A\otimes V_B)(A^x_a\otimes \Id)\ket{\psi} - \big((\wtd{A}^x_a\otimes\Id)\ket
        {\wtd{\psi}}\big)\otimes\ket{aux}}^2    \right)^{1/2}&\leq \epsilon,\\
        \left(\expect{y\sim \nu_B}\sum_b  \norm{(V_A\otimes V_B)(\Id\otimes B^y_b)\ket{\psi} - \big((\Id\otimes\wtd{B}^y_b)\ket
        {\wtd{\psi}}\big)\otimes\ket{aux}}^2    \right)^{1/2}&\leq \epsilon,
    \end{align*}
    where $\nu_A$ and $\nu_B$ are marginal distributions of $\nu$ on $\mcX$ and $\mcY$ respectively.
\end{definition}
\begin{remark}
    The above definition is different from \cite[Definition 2.4]{Vid22}, where the last two inequalities are replaced by the single inequality
    \begin{equation*}
             \left(\expect{(x,y)\sim \nu}\sum_{a,b}  \norm{(V_A\otimes V_B)(A^x_a\otimes B^y_b)\ket{\psi} - \big((\wtd{A}^x_a\otimes\wtd{B}^y_b)\ket
        {\wtd{\psi}}\big)\otimes\ket{aux}}^2    \right)^{1/2}\leq \epsilon.
    \end{equation*}
    The reason we choose this different definition is that we need to relate local dilations in this framework to local dilations in the von Neumann algebra framework, as will be defined in Definition \ref{def:vNA-local-dilation}. This is only possible using our definition of a local dilation. We will now show that a local $(\epsilon,\nu)$-dilation in the sense of Definition \ref{def:local_dilation} is a local $(3\epsilon,\nu)$-dilation in the sense of \cite[Definition 2.4]{Vid22}, but that the reverse implication does not hold for any constant trade-off for POVM strategies. We do not know if such a separation exists for projective strategies as well. 

    Let $\tilde{\mcS}=(\kpt,\tA,\tB)$ be a local $(\epsilon,\nu)$-dilation of $\mcS=(\kp,A,B)$ in the sense of Definition \ref{def:local_dilation}. Using the fact that 
    \begin{equation*}
        \sum_b(V_BB_b^yV_B^*)^2\leq 1, \sum_a(\tA_a^x)^2\leq 1\text{ and } \sum_{a,b}(\tA_a^x\otimes V_BB_b^y)^*(\tA_a^x\otimes V_BB_b^y)\leq 1,
    \end{equation*}
   it follows from the inequalities of Definition \ref{def:local_dilation} that
    \begin{align*}
        \left(\expect{(x,y)\sim \nu}\sum_{a,b}  \norm{(V_A\otimes V_B)(A^x_a\otimes B^y_b)\ket{\psi} - (\tA_a^x\otimes \Id_{K_A}\otimes V_BB_b^yV_B^*)\ket
        {\wtd{\psi}}\otimes\ket{aux}}^2    \right)^{1/2}&\leq \epsilon,\\
        \left(\expect{(x,y)\sim \nu}\sum_{a,b}  \norm{(\tA_a^xV_A\otimes V_BB_b^y)\ket{\psi} - \big((\tA_a^x\otimes \tB_b^y)\ket
        {\wtd{\psi}}\big)\otimes\ket{aux}}^2    \right)^{1/2}&\leq \epsilon
    \end{align*}
    and
    \begin{align*}
        \left(\expect{(x,y)\sim \nu}\sum_{a,b}  \norm{(\tA_a^x\otimes V_BB_b^y)\big((V_A\otimes \Id_H)\ket{\psi} - (\Id_{\tH\otimes K_A}\otimes V_B^*)\ket
        {\wtd{\psi}}\otimes\ket{aux}\big)}^2\right)^{1/2}&\leq \epsilon.
    \end{align*}
    Using the triangle inequality now yields 
    \begin{equation*}
             \left(\expect{(x,y)\sim \nu}\sum_{a,b}  \norm{(V_A\otimes V_B)(A^x_a\otimes B^y_b)\ket{\psi} - \big((\wtd{A}^x_a\otimes\wtd{B}^y_b)\ket
        {\wtd{\psi}}\big)\otimes\ket{aux}}^2    \right)^{1/2}\leq 3\epsilon.
    \end{equation*}
    
    To disprove an implication with constant trade-off in the converse direction, let $\tilde{\mcS}=(\kpt,\tA,\tB)$ be a local $(\epsilon,\nu)$-dilation of $\mcS=(\kp,A,B)$ in the sense of \cite[Definition 2.4]{Vid22}. Consider now the families of strategies $\{\tilde{\mcS}_n\}_{n\geq 1}$ and $\{\mcS_n\}_{n\geq 1}$, given by
    \begin{align*}
        \mcS_n&=(\kp,A^{(n)},B^{(n)}),\ (A^{(n)})_{a,i}^{x}=\frac{1}{n}A_a^x,\ (B^{(n)})_{b,i}^{y}=\frac{1}{n}B_b^y \text{ and }\\
        \tilde{\mcS}_n&=(\kp,\tA^{(n)},\tB^{(n)}),\ (\tA^{(n)})_{a,i}^{x}=\frac{1}{n}\tA_a^x,\ (\tB^{(n)})_{b,i}^{y}=\frac{1}{n}\tB_b^y
    \end{align*}
    for answer sets $\mcA_n=\mcA\times\{1,\dots,n\}$ and $\mcB_n=\mcB\times\{1,\dots,n\}$. We see that $\tilde{\mcS}_n$ is a local $(\epsilon/n,\nu)$-dilation of $\mcS_n$ in the sense of \cite[Definition 2.4]{Vid22}. On the other hand, if $\tilde{\mcS}$ is a local $(\delta,\nu)$-dilation of $\mcS$ in the sense of Definition \ref{def:local_dilation} for some optimal $\delta>0$, then $\tilde{\mcS}_n$ is only a local $(\delta/\sqrt{n},\nu)$-dilation of $\mcS_n$ in the sense of Definition \ref{def:local_dilation}, proving that it is impossible to universally bound $\delta$ from above by $C\epsilon$ for some constant $C$. 

\end{remark}

\begin{definition}[Self-testing]\label{def:selftest}
Given a non-local game $G=(\mcX,\mcY,\nu,\mcA,\mcB,D)$, a class of strategies $\mcC$, an $\wtd{\mcS}\in\mcC$ that is optimal for $G$, a probability distribution $\hat{\nu}$ on $\mcX\times\mcY$ and a function $\kappa:\R_{\geq 0}\arr\R_{\geq 0}$ such that $\kappa(\epsilon)\arr 0$ as $\epsilon\arr 0$, we say that $G$ $(\kappa,\hat{\nu})$-robustly self-tests $\wtd{\mcS}$ for the class $\mcC$ if for any $\epsilon\geq 0$ and $\mcS\in\mcC$ with $\omega(G;\mcS)\geq w_q(G)-\epsilon$, $\wtd{\mcS}$ is a local $(\kappa(\epsilon),\hat{\nu})$-dilation of $\mcS$. We say that $G$ $\kappa$-robustly self-tests $\tilde{\mcS}$ for class $\mcC$ if $\hat{\nu}$ is equal to $\nu$, the distribution used in the non-local game.
\end{definition}

The optimal strategy $\wtd{\mcS}$ in the above definition is usually referred to as the ideal optimal strategy for $G$. In this work, we are primarily interested in the class of all PME strategies $\mcC_{PME}$ and the class of all strategies $\mcC_{all}$. We simply say that $G$ $(\kappa,\hat{\nu})$-robustly self-tests $\wtd{\mcS}$ (or $G$ is a $(\kappa,\hat{\nu})$-robust self-test) if $G$ $(\kappa,\hat{\nu})$-robustly self-tests $\wtd{\mcS}$ for $\mcC_{all}$. We say that $G$ $(\kappa,\hat{\nu})$-PME-robustly self-tests $\wtd{\mcS}$ (or $G$ is a $(\kappa,\hat{\nu})$-PME-robust self-test) if $G$ $(\kappa,\hat{\nu})$-robustly self-tests $\wtd{\mcS}$ for $\mcC_{PME}$. 

In Section \ref{sec:non-local-games-and-strategies} the perhaps somewhat uncommon notion of a $\beta$-synchronous non-local game was introduced. If one views robust self-testing from the viewpoint of certification protocols, the $\beta$-synchronous condition is not significantly stronger than being synchronous. This is captured in the following definition and lemma.

\begin{definition}
    Let $G=(\mcX,\nu,\mcA,D)$ be a synchronous game, let $\nu_A$ be the marginal distribution of $\nu$ on $\mcX$ and let $\beta\in (0,1)$. Let $\nu'$ be the probability distribution on $\mcX\times\mcX$ defined by $\nu'(x,y)=\beta\nu_A(x)\delta_{x,y}+(1-\beta)\nu(x,y)$. Then we call $(\mcX,\nu',\mcA,D)$ the $\beta$-synchronised version of $G$.
\end{definition}

\begin{lemma}\label{lem:beta-synchronised}
    Let $G=(\mcX,\nu,\mcA,D)$ be a synchronous game and let $\beta\in (0,1)$. Let $\hat{\nu}$ be a probability distribution on $\mcX\times\mcX$ and let $G'$ be the $\beta$-synchronised version of $G$. If $G$ $(\kappa,\hat{\nu})$-robustly self-tests an optimal synchronous strategy $\mcS$ for class $\mcC$, then $G'$ $(\kappa',\hat{\nu})$-robustly self-tests $\mcS$ for class $\mcC$ with $\kappa'(\epsilon)=\kappa\big(\frac{\epsilon}{1-\beta}\big)$.
\end{lemma}
\begin{proof}
    This is immediate after realising that for any strategy $\hat{\mcS}$ we have the implication 
    \begin{equation*}
        |\omega(G';\mcS)-\omega(G';\hat{\mcS})|\leq \epsilon \implies |\omega(G;\mcS)-\omega(G;\hat{\mcS})|\leq \frac{\epsilon}{1-\beta}.\qedhere
    \end{equation*}
\end{proof}

Historically, robust self-testing has been studied in the case where the probability distribution of the game equals the probability distribution used in the local dilation. However, one can conceive situations where one wants to verify that some parts of a strategy match the ideal one, and that the other questions are merely present to ensure this behaviour. We will encounter an example of this behaviour when we come to the Quantum Low Degree Test in Section \ref{sec:intro-QLDT}.

\Cref{def:local_dilation} is about measuring the ``distance" between two strategies. As we will discuss in \Cref{sec3}, for ME strategies, their distance can be described in the framework of von Neumann algebras. So we provide more background on von Neumann algebras for subsequent use. We use $\ell^2(\N)$ to denote the Hilbert space of sequences in $\C^\N$ that are convergent in the Euclidean norm and we use $\{\ket{e_i}:i\in \N\}$ to denote the standard basis. The trace $\Tr$ on the von Neumann algebra $B(\ell^2(\N))$ of bounded operators on $\ell^2(\N)$ is given by $\Tr(x)=\sum_{i\in \N}\bra{e_i}x\ket{e_i}$. In general, the tensor product of two tracial von Neumann algebras is viewed as a tracial von Neumann algebra by taking the spacial tensor product and equipping it with the tensor product of the traces. For any tracial von Neumann algebra $(M,\tau^M)$, we denote by $M^{\infty}$ the von Neumann algebra $M\overline{\otimes} B(\ell^2(\N))$, i.e. the $\sigma$-weak closure of $M\otimes B(\ell^2(\N))$ equipped with the trace $\tau^{\infty}=\tau^M\otimes\Tr$. Let $I_M$ be the projection onto the $1^{\text{st}}$ coordinate in $\C^\N$. We usually identify $M$ with $M\otimes I_M$ in $M^{\infty}$ and write $I_M$ for $1\otimes I_M\in M^{\infty}$. For any projection $P\in M^{\infty}$, an operator $V\in M^{\infty}P$ is an isometry if $V^*V=P$.

\subsection{Almost synchronous correlations}

To demonstrate that a game is a robust self-test, we need to study strategies that are nearly optimal. In particular, we consider correlations and strategies that are \textbf{almost synchronous}. Given a distribution $\nu$ on $\mcX$, the asynchronicity of a correlation $C\in C_q(\mcX,\mcX,\mcA,\mcA)$ with respect to $\nu$ is 
\begin{equation*}
    \dsync(C;\nu):=\expect{x\sim \nu}\sum_{a\neq b}C_{x,x,a,b}=1-\expect{x\sim \nu}\sum_{a}C_{x,x,a,a}.
\end{equation*}
The asynchronicity $\dsync(\mcS;\nu)$ of a strategy $\mcS=(\ket{\psi},A,B)$ refers to the asynchronicity of the correlation induced by $\mcS$. If $\mcS=(M,\tau,A)$ is an ME strategy, then
\begin{equation}
    \dsync(\mcS;\nu)=\expect{x\sim \nu}\sum_{a\neq b}\tau(A^x_aA^x_b)=1-\expect{x\sim \nu} \sum_a\tau\left((A^x_a)^2\right). \label{eq:dsync-vna}
\end{equation}
From \Cref{eq:dsync-vna}, it is easy to see that the asynchronicity of any PME strategy is $0$.

\begin{lemma}\label{lem:value-dsync}
    Let $G=(\mcX,\nu,\mcA,D)$ be a $\beta$-synchronous game and let $\nu_A$ be the marginal distribution of $\nu$ on $\mcX$. Then
    \begin{equation*}
    \omega(G;C)  \leq 1-\beta \dsync(C;\nu_A)
    \end{equation*}
    for any correlation $C$.
\end{lemma}
\begin{proof}
    Since $G$ is $\beta$-synchronous, we have
    \begin{equation*}
        \nu(x,x)\geq \beta\sum_a\nu(x,y)=\beta\nu_A(x)
    \end{equation*}
    for all $x\in\mcX$. So 
    \begin{equation*}
        \delta:=\dsync(C;\nu_A)=1-\sum_a\nu_A(x)\sum_aC_{x,x,a,a}\geq 1-\frac{1}{\beta}\sum_a\nu(x,x)\sum_aC_{x,x,a,a}.
    \end{equation*}
It follows that
\begin{align*}
    1-\omega(G;C)&=\sum_{x,y}\nu(x,y)\sum_{D(a,b|x,y)=0}C_{x,y,a,b}\\
    &\geq \sum_{x}\nu(x,x)\sum_{D(a,b|x,x)=0}C_{x,x,a,b}\\
    &=\sum_{x}\nu(x,x)\sum_{a\neq b}C_{x,x,a,b}\\
    &\geq \sum_{x}\beta\nu_A(x)\sum_{a\neq b}C_{x,x,a,b}\\
    &=\beta \delta.
\end{align*}
This proves the inequality.
\end{proof}

\begin{lemma}\label{lem:symm-dsync-est}
Let $\mcS=(\ket{\psi},A,B)$ be a strategy and let $\mcS_A$ and $\mcS_B$ be the associated symmetric strategies. Then for any distribution $\nu$ on $\mcX$,
\begin{align}
    1-\dsync(\mcS;\nu)&\leq \sqrt{1-\dsync(\mcS_A;\nu)}\sqrt{1-\dsync(\mcS_B,\nu)},\label{eq:symm-dsync-est1} \\
    \dsync(\mcS_A;\nu)&\leq 2 \dsync(\mcS;\nu),\label{eq:symm-dsync-est2}\text{ and}\\
    \dsync(\mcS_B;\nu)&\leq 2 \dsync(\mcS;\nu)\label{eq:symm-dsync-est3}.
\end{align}
\end{lemma}
\begin{proof}
    \Cref{eq:symm-dsync-est1} was proved in \cite[Corollary 3.2]{Vid22}. Since $\sqrt{1-\dsync(\mcS_B,\nu)}\leq 1$,
    \begin{align*}
        1-\dsync(\mcS;\nu)&\leq \sqrt{1-\dsync(\mcS_A;\nu)}\sqrt{1-\dsync(\mcS_B,\nu)}\\ 
        &\leq \sqrt{1-\dsync(\mcS_A;\nu)}\leq 1-\frac{1}{2}\dsync(\mcS_A;\nu).   \end{align*}
        This proves \Cref{eq:symm-dsync-est2}. \Cref{eq:symm-dsync-est3} holds similarly.
\end{proof}

The following ``replacement" lemma is similar to \cite[Lemma 2.10]{Vid22}. For the subsequent use, we provide specific constants instead of big-$\mcO$ in the asymptotic analysis.
\begin{lemma}\label{lem:subs-meas-ops-corr-est}
   Let $\mcS=(\ket{\psi},A,B)$ be a strategy for a symmetric game $G=(\mcX,\nu,\mcA,D)$. Let $\nu_A$ be the marginal distribution of $\nu$ on $\mcX$, and let $\rho_A$ and $\rho_B$ be the densities of $\ket{\psi}$ on $H_A$ and $H_B$ respectively. If $\hat{A}=\{\hat{A}^x_a\}$ is a family of POVMs on $H_A$ and $\hat{B}=\{\hat{B}^y_b\}$ is a family of POVMs on $H_B$ with
\begin{align*}
    \gamma_A:&=\expect{x\sim \nu_A}\sum_a \Tr\left( 
(A^x_a-\hat{A}^x_a)^2\rho_A \right),\\
 \gamma_B:&=\expect{y\sim \nu_A}\sum_b \Tr\left( 
(B^y_b-\hat{B}^y_b)^2\rho_B \right), 
\end{align*}
then
\begin{equation*}
    \expect{x,y\sim \nu}\sum_{a,b}\abs{C_{x,y,a,b}-\hat{C}_{x,y,a,b}}\leq 12\dsync(\mcS;\nu_A)+4\sqrt{\gamma_A}+4\sqrt{\gamma_B},
\end{equation*}
where $C$ and $\hat{C}$ are the correlations induced by $\mcS$ and $\hat{\mcS}=(\ket{\psi},\hat{A},\hat{B})$, respectively.
\end{lemma}
\begin{proof}
    Let $\mcS_A$ and $\mcS_B$ be the associated symmetric strategies of $\mcS$,
    and let $C'$ be the correlations induced by the strategy $(\kp,\hat{A},B)$. Then by \cite[Lemma 2.10]{Vid22}, we have
    \begin{align*}
        \expect{x,y\sim \nu}\sum_{a,b}\abs{C_{x,y,a,b}-C'_{x,y,a,b}}&\leq 3\dsync(\mcS_A;\nu_A)+4\sqrt{\gamma_A},\\
        \expect{x,y\sim \nu}\sum_{a,b}\abs{C'_{x,y,a,b}-\hat{C}_{x,y,a,b}}&\leq 3\dsync(\mcS_B;\nu_A)+4\sqrt{\gamma_B}.
    \end{align*}
    The rest follows from the triangle inequality and \Cref{lem:symm-dsync-est}. 
\end{proof}

Let $\mcS$ be an almost synchronous strategy. \Cref{lem:subs-meas-ops-corr-est} implies that small perturbations on the measurement operators will result in small perturbations on the correlation. This is particularly useful when we want to orthogonalize the measurements (i.e., find a projective strategy nearby).

\begin{lemma}\label{lem:nearby-proj-strat}
    Let $\mcS=(\ket{\psi},A, B)$ be a strategy for a symmetric game $G=(x,\nu, \mathcal{A}, D)$ with reduced densities $\rho_A$ and $\rho_B$ on $H_A$ and $H_B$, respectively. Let $\nu_A$ be the marginal of $\nu$ on $\mcX$ and let $\delta=\dsync(\mcS,\nu_A)$. Then there exists a projective strategy $\hat{\mcS}=(\ket{\psi},\hat{A},\hat{B})$ such that $\hat{\delta}:=\dsync(\hat{\mcS},\nu_A)=\mathcal{O}(\delta^{\frac{1}{8}})$,
\begin{align}
	\expect{x\sim \nu_A}\sum_a \Tr\left((A_a^x-\hat{A}_a^x)^2\rho_A\right)&=\mathcal{O}(\delta^{\frac{1}{4}})\label{eq:lem-proj-sym-strat-op-diff-A},\\
	\expect{x\sim \nu_A}\sum_a \Tr\left((B_a^x-\hat{B}_a^x)^2\rho_B\right)&=\mathcal{O}(\delta^{\frac{1}{4}}),\label{eq:lem-proj-sym-strat-op-diff-B}
\end{align}
and 
\begin{equation}\label{eq:lem-proj-sym-strat-corr-diff-A}
	\expect{x,y\sim \nu}\sum_{a,b}|C_{x,y,a,b}-\hat{C}_{x,y,a,b}|\leq \mathcal{O}(\delta^{\frac{1}{8}}),
\end{equation}
where $C$ and $\hat{C}$ are the correlations induced by $\mcS$ and $\hat{\mcS}$, respectively.
\end{lemma}
\begin{proof}
    Let $\mcS_A$ and $\mcS_B$ be the associated symmetric strategies of $\mcS$. \Cref{lem:symm-dsync-est} implies that $\delta_A:=\dsync(\mcS_A,\nu_A)= \mcO(\delta)$ and $\delta_B:=\dsync(\mcS_B,\nu_A)= \mcO(\delta)$. Then by \cite[Lemma 3.4]{Vid22}, there are PVMs $\hat{A}=\{\hat{A}^x_a\}$ on $H_A$ such that 
    \begin{equation*}
        \expect{x\sim \nu_A}\sum_a \Tr\left((A_a^x-\hat{A}_a^x)^2\rho_A\right)=\mathcal{O}(\delta_A^{\frac{1}{4}})= \mathcal{O}(\delta^{\frac{1}{4}}).
    \end{equation*}
    This proves \Cref{eq:lem-proj-sym-strat-op-diff-A}. The existence of PVMs $\hat{B}$ and \Cref{eq:lem-proj-sym-strat-op-diff-B} follows similarly. Since every $B^x_a$ and $\hat{A}^x_a$  are measurement operators,
    \begin{align*}
        \abs{\hat{\delta}-\delta}&=\left\vert \expect{x\in\nu_A}\sum_{a}\left(\bra{\psi}A^x_a\otimes B^x_a\ket{\psi}-\bra{\psi}\hat{A}^x_a\otimes \hat{B}^x_a\ket{\psi} \right)  \right\vert\\
        &= \left\vert \expect{x\in\nu_A}\sum_{a}\left(\bra{\psi}(A^x_a-\hat{A}^x_a)\otimes B^x_a\ket{\psi}+\bra{\psi}\hat{A}^x_a\otimes(B^x_a- \hat{B}^x_a)\ket{\psi} \right)  \right\vert\\
        &\leq \left(\expect{x\sim \nu_A}\sum_a \Tr\left((A_a^x-\hat{A}_a^x)^2\rho_A\right) \right)^{1/2} + \left(\expect{x\sim \nu_A}\sum_a \Tr\left((B_a^x-\hat{B}_a^x)^2\rho_B\right)  \right)^{1/2}\\
        &= \mcO(\delta^{1/8}).
    \end{align*}
    Hence $\hat{\delta}=\mcO(\delta^{1/8})$. Then \Cref{eq:lem-proj-sym-strat-corr-diff-A} follows by \Cref{lem:subs-meas-ops-corr-est}.
\end{proof}

\Cref{eq:lem-proj-sym-strat-corr-diff-A} in the above lemma says that the correlations induced by $\mcS$ and $\hat{\mcS}$ are close. The following lemma further demonstrates that for any non-local game $G$, the winning probabilities of $\mcS$ and $\hat{\mcS}$ are close.

\begin{lemma}\label{lem:win-prob-est}
   Suppose $C$ and $C'$ are correlations that are induced by two strategies $\mcS$ and $\mcS'$, respectively, for a non-local game $G=(\mcX,\mcY,\nu,\mcA,\mcB,D)$ such that
   \begin{equation*}
       \expect{x,y\sim \nu} \sum_{a,b}\abs{C_{x,y,a,b}-C'_{x,y,a,b}}\leq \epsilon.
   \end{equation*}
   Then $\abs{\omega(G;\mcS)-\omega(G;\mcS')}\leq \epsilon$.
\end{lemma}

\begin{proof}
    Since every $D(a,b|x,y)$ is either 1 or 0, this lemma is an immediate consequence of the H\"older inequality.
\end{proof}

We call the collection of positive operators $\{A_i:1\leq i\leq m\}$ an \textbf{incomplete POVM} if $\sum_{i=1}^m A_i\leq \Id$. It can be completed to a POVM $\{A_i:1\leq i\leq m+1\}$ by adding an outcome $m+1$ where $A_{m+1}=\Id-\big(\sum_{i=1}^m A_i\big)$. The following lemma gives an upper bound of the ``distance" between any two incomplete POVMs.

\begin{lemma}\label{lem:op-norm-bound-game-poly}
    Let $\{A_i:1\leq i\leq m\}$ and $\{\wtd{A}_i:1\leq i\leq m\}$ be two collections of positive operators on some Hilbert space $H$ such that $\sum_{i=1}^m A_i\leq \Id$ and $\sum_{i=1}^m \wtd{A}_i\leq \Id$. Then
 \begin{equation*}
    \left\Vert \sum_{i=1}^m(A_i-\wtd{A}_i)^2  \right\Vert_{\infty}\leq 4.
\end{equation*}
\end{lemma}

\begin{proof}
    We assume without loss of generality that $\{A_i\}$ and $\{\tA_i\}$ are POVMs, as completing them can only increase the norm. We start by considering just the sum and expanding the square to find that
    \begin{equation*}
        0\leq \sum_{i=1}^m(A_i-\tA_i)^2=\sum_{i=1}^mA_i^2+\tA_i^2-A_i\tA_i-\tA_iA_i\leq 2-\left(\sum_{i=1}^mA_i\tA_i+\tA_iA_i\right),
    \end{equation*}
    having used that $0\leq A_i,\tA_i\leq 1$, so $A_i^2\leq A_i$ and $\tA_i^2\leq \tA_i$. Consequently,
    \begin{equation}\label{eq:est-lem-op-norm-bound-game-poly}
        \left\Vert\sum_{i=1}^m(A_i-\tA_i)^2\right\Vert_{\infty}\leq 2+2\left\Vert\sum_{i=1}^mA_i\tA_i\right\Vert_{\infty}.
    \end{equation}
    By the Cauchy-Schwarz inequality for inner product C*-modules \cite[Lemma 2.5]{RW98}, we know that
    \begin{equation*}
        \left(\sum_{i=1}^mA_i\tA_i\right)^*\left(\sum_{i=1}^mA_i\tA_i\right)\leq \left\Vert\sum_{i=1}^mA_i^2\right\Vert_{\infty}\sum_{i=1}^m\tA_i^2\leq 1,
    \end{equation*}
    so by the C*-identity we have
    \begin{equation*}
        \left\Vert\sum_{i=1}^mA_i\tA_i\right\Vert_{\infty}\leq 1.
    \end{equation*}
    Plugging this into \Cref{eq:est-lem-op-norm-bound-game-poly} proves the lemma.
\end{proof}

We now describe the results in \cite[Section 3]{Vid22}. These are essential for our proof of Theorem \ref{thm:constr-isom-gen-strat}. However, we need to use the explicit constructions and results in the proofs of Vidick's main theorems, instead of the final results. For this reason, we combine the results we need in the appropriate form in the following theorem. In this theorem, and in the rest of this paper, we say that $x\leq \poly(\epsilon)$ if there exist universal constants $C,c>0$ such that $x\leq C\epsilon^c$. 
\begin{theorem}[Vidick]\label{thm:Vidick}
	Let $\mcX$ and $\mcA$ be finite question and answer sets and $\nu$ be a measure on $\mcX\times \mcX$ with marginal $\nu_A$ on the first entry. Let $\mcS=(\kp,A,B)$ be a projective strategy on $H_A\otimes H_B$ and $\delta=\dsync(\mcS,\nu_A)$. Let $\rho_A$ be the reduced density of $\kp$ on $H_A$, let $P_{\lambda}=\chi_{\geq \lambda}(\rho_A)$ be its spectral projections for $\lambda\geq 0$ and let $H_{\lambda}=P_{\lambda}H_A$. Then the following are true:
	\begin{enumerate}[(i)]
		\item The strategies $\mcS_{\lambda}=(\ket{\psi_{\lambda}},P_{\lambda}AP_{\lambda})$ are ME strategies, where $\ket{\psi_{\lambda}}$ is the maximally entangled state on $H_{\lambda}\otimes H_{\lambda}$. 
		\item Let $\rho_{\lambda}$ be the maximally mixed state on $H_{\lambda}$. Then
		\begin{equation*}
			\rho_A=\int_0^{\infty}\rho_{\lambda}d\mu(\lambda),
		\end{equation*}
		where $\mu$ is the probability measure defined by $d\mu(\lambda)=\Tr(P_{\lambda})d\lambda$.
		\item The strategies $\mcS_{\lambda}$ provide an approximate decomposition of $\mcS$ as a convex sum of maximally entangled strategies in the following sense:
		\begin{equation*}
			\expect{x\sim \nu_A}\sum_a\int_0^{\infty}\Tr\left(\left(A_a^x-P_{\lambda}A_a^xP_{\lambda}\right)^2\rho_{\lambda}\right)d\mu(\lambda)\leq \sqrt{2\delta}.
		\end{equation*}
		\item Let $C$ be the correlation of $\mcS$. The correlations $C^{\lambda}$ of $\mcS_{\lambda}$ satisfy
		\begin{equation*}
			\expect{(x,y)\sim \nu}\sum_{a,b}\left|C_{x,y,a,b}-\int_0^{\infty}C^{\lambda}_{x,y,a,b}d\mu(\lambda)\right|\leq \poly(\delta).
		\end{equation*}
		\item There exist PME strategies $\mcS_{\lambda}'=(\ket{\psi_{\lambda}},A^{\lambda})$ such that (iii) and (iv) hold with $\mcS_{\lambda}$, $P_{\lambda}A_a^xP_{\lambda}$ and $\sqrt{2\delta}$ replaced by $\mcS_{\lambda}'$, $A^{\lambda,x}_a$ and $\poly(\delta)$, respectively.
	\end{enumerate}
\end{theorem}
\begin{remark}
	The main omission from the results in \cite[Section 3]{Vid22} is that similar results with worse dependence on $\delta$ hold for non-projective strategies $\mcS$. Furthermore, one should note that this theorem can also be applied to the B-side of the strategy $\mcS$. 
\end{remark}
\begin{remark}\label{rmk:comm-meas-proj}
	Note that 
	\begin{equation*}
		\frac{1}{\Tr(P_{\lambda})}\norm{A_a^xP_{\lambda}-P_{\lambda}A_a^x}_2^2=2\Tr\left(\left(A_a^x-P_{\lambda}A_a^xP_{\lambda}\right)^2\rho_{\lambda}\right),
	\end{equation*}
	so the above theorem also shows that
	\begin{equation*}
		\expect{x\sim \nu_A}\sum_a\int_0^{\infty}\frac{1}{\Tr(P_{\lambda})}\norm{A_a^xP_{\lambda}-P_{\lambda}A_a^x}_2^2d\mu(\lambda)\leq 2\sqrt{2\delta},
	\end{equation*}
	which is the form we will use in our proofs.
\end{remark}

\subsection{Quantum low degree test} \label{sec:intro-QLDT}
A kind of strategy that we would like to self-test in particular is strategies that use the spectral projections of elements of the Pauli group and the maximally entangled state on the corresponding number of qubits. After making this statement more precise, we present the recently discovered Quantum Low Degree Test, which is a PME-robust self-test.

Let $k\in \N$ and let $X,Z\in M_2(\C)$ be the corresponding Pauli matrices. For a prime power $q$ we let $\mbF_q$ denote the finite field with $q$ elements. For an $\bfa\in \mbF_2^k$ we denote the operator $\bigotimes_{j=1}^kX^{a_j}$ by $\sigma^X(\bfa)$ and $\bigotimes_{j=1}^kZ^{a_j}$ by $\sigma^Z(\bfa)$. The Pauli group on $k$ qubits, $P_k$, is then given by
\begin{equation*}
    P_k=\{(-1)^a\sigma^X(\bfb)\sigma^Z(\bfc)|a\in \mbF_2,\bfb,\bfc\in\mbF_2^k\}.
\end{equation*}
Informally, having access to $P_k$ means that you have access to $k$ qubits, since it is a basis for $B((\C^2)^{\otimes k})$ \cite{CRSV17}, so a qubit test should in some sense verify that you have access to the Pauli group. In the non-local game setting, a qubit test verifies that both Alice and Bob have access to $k$ qubits, and that those qubits are maximally entangled. The definition was introduced in \cite{CVY23}, where they only consider PME strategies. Here, we will include the class of strategies in the definition.

\begin{definition} \label{def:qubit-test}
    Let $k\in \N, \kappa: [0,1]\rightarrow \R_+$ and $\mcC$ be a class of strategies. A $(k,\kappa)-\mcC-$qubit test is a synchronous game $G=(\mcX,\nu,\mcA,D)$ such that there exist
    \begin{itemize}
        \item two sets $S_X,S_Z\subset \mbF_2^k$ that each span $\mbF_2^k$,
        \item an injection $\phi:(\{X\}\times S_X)\cup(\{Z\}\times S_Z)\rightarrow \mcX$ such that $\mcA(\phi(X,\bfa))=\mcA(\phi(Z,\bfb))=\mbF_2$ for all $\bfa\in S_X,\ \bfb\in S_Z$,
    \end{itemize}
and the following hold:
    \begin{itemize}
        \item (Completeness:) There exists a strategy $\tilde{\mcS}=(\ket{\tilde{\psi}},\tilde{A},\tilde{B})$ for $G$ on $(\C^{2^k}\otimes H_A)\otimes (\C^{2^k}\otimes H_B)$ for some Hilbert spaces $H_A$ and $H_B$ such that $\omega(\tilde{\mcS})=1$, $U(\tA^{\phi(W,\bfa)})=\sigma^W(\bfa)$ for every $W\in\{X,Z\}$ and $\bfa\in S_W$ and $\Tr_{H_A,H_B}(\ket{\tilde{\psi}}\bra{\tilde{\psi}})$ is maximally entangled on $\C^{2^k}\otimes \C^{2^k}$. Here $U$ denotes the Fourier transform defined in \Cref{sec:non-local-games-and-strategies}, which maps a PVM to the corresponding observable.
        \item (Soundness:) Let $\nu'$ be the renormalised restriction of $\nu$ on the image of $\phi$. For any strategy $\mcS=(\ket{\psi},A,B)$ in the class $\mcC$ with $\omega(\mcS)= 1-\epsilon$ for some $\epsilon \geq 0$, we have that $\tilde{\mcS}$ is a local $(\epsilon,\nu')$-dilation of $\mcS$. In other words, $G$ $(\kappa,\nu')$-robustly self-tests $\tilde{\mcS}$ for the class $\mcC$.
    \end{itemize}
\end{definition}

We now turn towards the Quantum Low Degree Test, which was introduced in \cite{CVY23}. This test is based on linear codes, so we will first briefly introduce this. More details can be found in \cite{CVY23}. 

Let $n,k,d\in \N$ and $q$ a prime power. An $[n,k,d]_q$-linear code $C\subset \mbF_q^n$ is a $k$-dimensional subspace such that for all $x\neq 0$ the number of non-zero elements of $x$ (called the Hamming weight) is at least $d$, the distance of the code. We call $d/n$ the relative distance. A parity check matrix for a code $C$ is a matrix $h\in \mbF_q^{m\times n}$ such that $\ker(h)=C$. It is an $r$-local tester if the Hamming weight of each row is at most $r$. A generating matrix for a code $C$ is a matrix $E\in \mbF_q^{n\times k}$ such that the rows of $E$ form a basis for $C$. One example of a code is the binary version of the Reed-Muller code $C_{\mathrm{RM2}}$, which is an $[2^{t(m+1)},t(d+1)^m,D]_2$-linear code with $D\geq \frac{1}{2}(1-\frac{md}{2^t})2^{t(m+1)}$ and an $(d+2)$-local tester with $2^{t(m+2)}(1+m)$ rows \cite{CVY23}. 

From the code $C_{\mathrm{RM2}}$ and a generating matrix $E$ for this code one can construct the Quantum Low Degree Test $G_{\mathrm{qldt}}=(\mcX_{\mathrm{qldt}}, \nu_{\mathrm{qldt}}, \mcA_{\mathrm{qldt}}, D_{\mathrm{qldt}})$, for which we refer to \cite{CVY23}. This non-local game is a $(k,\kappa)$-$\mcC_{PME}$-qubit test with $\kappa(\epsilon)=\poly(m,d,t)\cdot \poly(\epsilon,2^{-t})$ and $k=t(d+1)^m$. One can choose $m,\ d$ and $t$ such that the non-local game has $2^{\poly(\log(k))}$ questions and $\kappa(\epsilon)=\poly(\log(k))\cdot \poly(\epsilon)$. As $G_{\mathrm{qldt}}$ is a PME-qubit test, there exists a renormalised restriction of $\nu_{\mathrm{qldt}}$ on the image of $\phi$ as in Defintion \ref{def:qubit-test}, which we denote by $\nu'_{\mathrm{qldt}}$. The main properties we need is that this game is symmetric, that $S_X=S_Z=\{(E_{ij})_{i=1}^k|1\leq j\leq n\}$ where $n=2^{t(m+1)}$, the length of the code $C_{\mathrm{RM2}}$ and that the ideal strategy does not use auxiliary Hilbert spaces.

\section{Self-testing in the von Neumann algebra picture}\label{sec3}
For a large part of our analysis it is necessary to work with the measurement operators of a single party and with the reduced density matrix of the state for that party. In this case, Definition \ref{def:local_dilation} is not a convenient way to view local dilations. Instead, we would like a formulation in terms of the measurement operators on Alice's side and the reduced state for Alice. This is not possible for all strategies, but it is possible for strategies that are classical convex combinations of ME strategies, i.e. the strategies that can be described as a triple $(M,\tau,A)$ for a tracial von Neumann algebra $(M,\tau)$ and a family of POVMs $A$ in $M$.

\begin{definition}\label{def:vNA-local-dilation}
    Given two strategies $\wtd{\mcS}=(M,\tau^M,\wtd{A})$ and $\mcS=(N,\tau^N,A)$, a distribution $\nu$ on $\mcX$, and an $\epsilon\geq 0$, we say that $\wtd{\mcS}$ is a local $(\epsilon,\nu)$-vNA-dilation of $\mcS$ if the following statements hold. There exist a finite dimensional von Neumann algebra $M_0$ with tracial state $\tau^{M_0}$, a projection $P\in (M\otimes M_0)^{\infty}$ of finite trace such that $N\cong P (M\otimes M_0)^{\infty}P$ and $\tau^N=\tau^{\infty}/\tau^{\infty}(P)$, and a partial isometry $W\in P(M\otimes M_0)^{\infty}I_{M\otimes M_0}$ such that
\begin{equation}
    \expect{x\sim \nu}\sum_a\norm{\wtd{A}^x_a\otimes I_{M_0}-W^* A^x_a W}^2_{\tau^{M\otimes M_0}}\leq \epsilon \label{eq:tau-infinity-diff}
\end{equation}
and
\begin{equation}
    \tau^N(P-WW^*)\leq \epsilon, \tau^{M\otimes M_0}(I_{M\otimes M_0}-W^*W)\leq \epsilon. \label{eq:iso-proj}
\end{equation}
\end{definition}

The above definition is modelled on the notions of closeness of strategies and soundness of a qubit test in \cite[Definition 5.3 and 5.6]{CVY23} in the sense that soundness in the qubit test implies that the corresponding strategies are approximate local vNA-dilations. The following lemma shows that for ME strategies, there is a square dependence between the notions of local dilation and local vNA-dilation, with constants that are independent of the question and answer sizes.

\begin{lemma}\label{lem:PME-robust-self-testing-vNA}
    Let $G=(\mcX,\nu,\mcA,D)$ be a symmetric game, $\hat{\nu}$ a symmetric probability distribution on $\mcX\times\mcX$ with marginal distribution $\hat{\nu}_A$ on $\mcX$, and let $\wtd{\mcS}=(M,\tau^M,\wtd{A})$ and $\mcS=(N,\tau^N,A)$ be two ME strategies for $G$, where $M= B(\tilde{H})$ and $N= B(H)$ for some finite-dimensional spaces $H$ and $\wtd{H}$. 
\begin{enumerate}[(a)]
    \item If $\wtd{\mcS}$ is a local $(\epsilon,\hat{\nu})$-dilation of $\mcS$, then $\wtd{\mcS}$ is a local $(1700\epsilon^2,\hat{\nu}_A)$-vNA-dilation of $\mcS$.
    \item If $\wtd{\mcS}$ is an local $(\epsilon,\hat{\nu}_A)$-vNA-dilation of $\mcS$, then $\wtd{\mcS}$ is a local $((4+\sqrt{2})\sqrt{\epsilon},\hat{\nu})$-dilation of $\mcS$. 
\end{enumerate}
\end{lemma}
\begin{remark}
    We prove this lemma in the appendix because it takes quite some effort. This is mainly due to the fact that an auxiliary state in a local dilation does not have to be maximally entangled, even if both $\mcS$ and $\tilde{\mcS}$ are ME strategies. As we want to obtain a tracial state on $M\otimes M_0$, this is a problem. We get around this by proving that the auxiliary state can always be taken to be maximally entangled (see Lemmas \ref{lem:me-aux} and \ref{lem:partial-iso}), but this comes at a price in terms of the constant for the approximate local dilation, which is the reason why the constant 1700 appears in the theorem. It is good to note that the auxiliary state is not required to be fully maximally entangled if the auxiliary von Neumann algebra $M_0$ is not a factor, i.e. $B(H_0)$ for some Hilbert space $H_0$. Unfortunately, it is unclear how this freedom can be exploited.
\end{remark}

We now formulate the first lifting result for robust self-testing, where we show that we can use robust self-testing for the class of PME strategies to say something about almost synchronous ME strategies as well. Note that the most natural case is when $\hat{\nu}=\nu$, but to apply this result to the Quantum Low Degree Test, we need this more general version.

\begin{lemma}\label{lem:ME-robust-self-testing-vNA}
Let $G=(\mcX,\nu,\mcA,\lambda)$ be a symmetric game with the marginal distribution $\nu_A$ on $\mcX$ and let $\hat{\nu}$ be a symmetric probability distribution on $\mcX\times\mcX$ and $c\geq 1$ such that $\hat{\nu}_A\leq c\nu_A$, where $\hat{\nu}_A$ is the marginal distribution of $\hat{\nu}$ on $\mcX$. Suppose $\wtd{\mcS}=(M,\tau^M,\wtd{A})$ is a PME strategy that is optimal for $G$ and $G$ $(\kappa,\hat{\nu})$-PME-robustly self-tests $\wtd{\mcS}$. Then for any ME strategy $\mcS=(N,\tau^N,A)$ with $\delta=\dsync(\mcS,\nu_A)$ and $\omega(G;\mcS)\geq \omega(G,\tilde{S})-\epsilon$, $\wtd{\mcS}$ is a $(\kappa'(\epsilon,\delta),\hat{\nu})$-local dilation of $\mcS$,
where $\kappa'(\epsilon,\delta)=\kappa\left(24\sqrt{\delta}+\epsilon\right)+9c\delta$.
\end{lemma}

\begin{proof}
    Let $\delta_x:=1-\sum_a\tau^N\big((A^x_a)^2\big)$ for all $x\in\mcX$. Then the asynchronicity of $\mcS$ is $\delta:=\dsync(\mcS;\nu_A)=\expect{x\sim \nu_A}\delta_x$.
    By \cite[Theorem 1.2]{dlS22a}, there is a family of PVMs $P=\{P^x_a:a\in\mcA\}_{x\in\mcX}$ in $N$ such that
    \begin{equation*}
    \sum_a\norm{A^x_a-P^x_a}_{\tau^N}^2\leq 9\delta_x
\end{equation*}
for all $x\in\mcX$. So
\begin{align}
    \gamma:=\expect{x\sim\nu_A}\sum_{a}\norm{A^x_a-P^x_a}_{\tau^N}^2\leq9\delta.\label{eq:dls22a}
\end{align}
Consider the PME strategy $\mcS':=(N,\tau^N,P)$. Since $\dsync(\mcS',\nu_A )=0$ and both $\mcS$ and $\mcS'$ are symmetric, by \Cref{lem:subs-meas-ops-corr-est,lem:win-prob-est},
\begin{align*}
    \abs{\omega(G;\mcS)-\omega(G;\mcS')}\leq 8\sqrt{\gamma}=24\sqrt{\delta}.
\end{align*}
Hence, $\omega(G;\mcS')\geq \omega(G;\tilde{\mcS})-\epsilon-24\sqrt{\delta}$. Let $\lambda=\epsilon+24\sqrt{\delta}$. Let $\kp\in H\otimes H$ and $\kpt\in \tH\otimes \tH$ be GNS states for $(N,\tau^N)$ and $(M,\tau^M)$, respectively. Since $G$ $(\kappa,\hat{\nu})$-PME robustly self-tests $\wtd{\mcS}$, there are isometries $V_A:H\arr \wtd{H}\otimes K_A$ and $V_B:H\arr\wtd{H}\otimes K_B$ and a unit vector $\ket{aux}\in K_A\otimes K_B$ such that
    \begin{align*}
    \norm{(V_A\otimes V_B)\ket
    \psi-\ket{\wtd{\psi}}\otimes\ket{aux}}&\leq \kappa\left(\lambda\right),\\
        \left(\expect{x\sim \hat{\nu}_A}\sum_{a}  \norm{(V_A\otimes V_B)(P^x_a\otimes \Id)\ket{\psi} - \big((\wtd{A}^x_a\otimes\Id)\ket
        {\wtd{\psi}}\big)\otimes\ket{aux}}^2    \right)^{1/2}&\leq \kappa\left(\lambda\right),\\
        \left(\expect{y\sim \hat{\nu}_A}\sum_b  \norm{(V_A\otimes V_B)(\Id\otimes (P^y_b)^T)\ket{\psi} - \big((\Id\otimes(\wtd{A}^y_b)^T)\ket
        {\wtd{\psi}}\big)\otimes\ket{aux}}^2    \right)^{1/2}&\leq \kappa\left(\lambda\right).
    \end{align*}
Since $V_A\otimes V_B$ is an isometry, together with \Cref{eq:dls22a} and the fact that $\hat{\nu}_A\leq c\nu_A $, we have
\begin{align*}
    \left(\expect{x\sim \hat{\nu}_A}\sum_{a}  \norm{(V_A\otimes V_B)(A^x_a\otimes \Id)\ket{\psi} - \big((\wtd{A}^x_a\otimes\Id)\ket
        {\wtd{\psi}}\big)\otimes\ket{aux}}^2    \right)^{1/2}&\leq \kappa\left(\lambda\right)+9c\delta,
\end{align*}
and by symmetry
\begin{equation*}
    \left(\expect{y\sim \hat{\nu}_A}\sum_b  \norm{(V_A\otimes V_B)(\Id\otimes (A^y_b)^T)\ket{\psi} - \big((\Id\otimes(\wtd{A}^y_b)^T)\ket
        {\wtd{\psi}}\big)\otimes\ket{aux}}^2    \right)^{1/2}\leq \kappa\left(\lambda\right)+9c\delta.
\end{equation*}
As $\lambda=24\sqrt{\delta}+\epsilon$, we conclude that $\wtd{\mcS}$ is a $\left( \kappa\left(24\sqrt{\delta}+\epsilon\right)+9c\delta, \hat{\nu} \right)$-local dilation of $\mcS$.
\end{proof}

The next lemma shows that local $(\epsilon,\nu)$-vNA-dilations can equivalently be defined in terms of the distance between the measurement operators in $N$ instead of in $M\otimes M_0$. 
\begin{lemma}\label{lem:vNA-close-strat-N-measure}
    Let $\tilde{\mcS}=(M,\tau^M,\tA)$ and $\mcS=(N,\tau^N,A)$ be strategies and $\epsilon\geq 0$. Suppose that there exist a finite dimensional von Neumann algebra $M_0$ with tracial state $\tau^{M_0}$, a projection $P\in (M\otimes M_0)^{\infty}$ of finite trace such that $N\cong P(M\otimes M_0)^{\infty}P$ and $\tau^N=\tau^{\infty}/\tau^{\infty}(P)$, and a partial isometry $W\in P(M\otimes M_0)^{\infty}I_{M\otimes M_0}$ such that
    \begin{equation*}
        \tau^N(P-WW^*)\leq \epsilon, \tau^{M\otimes M_0}(I_{M\otimes M_0}-W^*W)\leq \epsilon.
    \end{equation*}
    Then for all $x\in\mcX$ we have
    \begin{align*}
        \sum_a\norm{\wtd{A}^x_a\otimes I_{M_0}-W^* A^x_a W}^2_{\tau^{M\otimes M_0}}&\leq2\epsilon+\frac{1}{1-\epsilon} \sum_a\norm{W(\wtd{A}^x_a\otimes I_{M_0})W^*-A^x_a}^2_{\tau^{N}}\text{ and}\\
        \sum_a\norm{W(\wtd{A}^x_a\otimes I_{M_0})W^*-A^x_a}^2_{\tau^{N}}&\leq \frac{2\epsilon}{1-\epsilon}+ \frac{1}{1-\epsilon} \sum_a\norm{\wtd{A}^x_a\otimes I_{M_0}-W^* A^x_a W}^2_{\tau^{M\otimes M_0}}.
    \end{align*}
\end{lemma}
\begin{proof}
    Let $x\in\mcX$ be arbitrary. Let us identify $\tA_a^x\otimes \Id$ and $\tA_a^x$. Note that
    \begin{align*}
        \left|\sum_a\tau^{\infty}\left(W^*W\tA^x_aW^*W\tA_a^x-(\tA_a^x)^2\right)\right|\leq& \left|\sum_a\tau^{\infty}\left(W^*W\tA^x_aW^*W\tA_a^x-\tA_a^xW^*W\tA_a^x\right)\right|\\
        &+\left|\sum_a\tau^{\infty}\left(\tA^x_aW^*W\tA_a^x-(\tA_a^x)^2\right)\right|\\
        \leq& \tau^{\infty}(I_{M\otimes M_0}-W^*W)\norm{\sum_a(\tA_a^x)^2}_{\infty}\\
        &+\tau^{\infty}(I_{M\otimes M_0}-W^*W)\norm{\sum_a\tA^x_aW^*W\tA_a^x}_{\infty}\\
        \leq& 2\epsilon.
    \end{align*}
    Therefore, we have
    \begin{align*}
        \sum_a\tau^{M\otimes M_0}\left((\wtd{A}^x_a-W^* A^x_a W)^2\right)=&\sum_a\tau^{\infty}\Big((\tA_a^x)^2+WW^*A_a^xWW^*A_a^x\\
        &-\tA_a^xW^*A_a^xW-W^*A_a^xW\tA_a^x\Big)\\
        \leq& 2\epsilon+\sum_a\tau^{\infty}\Big(W^*W\tA_a^xW^*W\tA_a^x+(A_a^x)^2-2\tA_a^xW^*A_a^xW\Big)\\
        =&2\epsilon+\sum_a\tau^{\infty}(P)\tau^N\left(\left(W\wtd{A}^x_aW^*-A^x_a\right)^2\right)
    \end{align*}
    and analogously
    \begin{equation*}
        \tau^N\left((W\wtd{A}^x_aW^*-A^x_a)^2\right)\leq \frac{2\epsilon}{\tau^{\infty}(P)}+\frac{1}{\tau^{\infty}(P)}\tau^{M\otimes M_0}\left((\wtd{A}^x_a-W^* A^x_a W)^2\right).
    \end{equation*}
    Since $1-\epsilon\leq \tau^\infty(P)\leq (1-\epsilon)^{-1}$ by Lemma \ref{lem:ME-strat-dim-est}, we have shown the desired result.
\end{proof}

\section{Robust self-testing based on the spectral gap of the game polynomial}\label{sec4}
In this section we show that a PME-robust self-test whose game polynomial has spectral gap is automatically a robust self-test. The proof consists of three steps. First, we show that any almost perfect strategy $\mcS$ for a PME-robust self-test can be approximately decomposed into `orthogonal' symmetric maximally entangled strategies $\mcS_i$ with high winning probability. Next, we use the PME-robust self-testing properties of the game to find partial isometries from $\mcS_i$ to the ideal strategy and combine and extend them to obtain isometries from $\mcS$ to the ideal strategy. The spectral gap of the game polynomial then allows us to conclude that these isometries form a local dilation.

In the first step of the proof, we will use \cite[Theorem 3.1]{Vid22}, presented as Theorem \ref{thm:Vidick} in this paper, and intermediate results of its proof to construct the decomposition into the strategies $\mcS_i$. One of these intermediate results provides a bound on the commutator of the measurement operators of $\mcS$ and the spectral projections of its reduced density matrix. To make proper use of this, we need the following lemma, which records its consequences.

\begin{lemma}\label{lem:comm-proj-PVM}
	Let $\{A^x_a\}$ be a collection of PVMs on $H_A$ and $\nu$ a symmetric probability measure on $\mcX\times \mcX$ with marginal $\nu_A$ on $\mcX$. Let $\gamma\geq 0$ and let $R$ be a projection on $H_A$ such that
	\begin{equation*}
		\expect{x\sim \nu_A}\sum_a\frac{1}{\Tr(R)}\norm{A_a^{ x}R-RA_a^{x}}_2^2\leq \gamma.
	\end{equation*}
     If $\kp\in (RH_A)\otimes (RH_A)$ is a maximally entangled state, then
    \begin{equation*}
        \dsync(\mcS,\nu_A)=\expect{x\sim\nu_A}\sum_a\frac{1}{\Tr(R)}\big|\Tr(RA_a^x)-\Tr(RA_a^xRA_a^x)\big|\leq \frac{1}{2}\gamma
    \end{equation*}
    for the ME strategy $\mcS=(\kp,\{RA^x_aR\})$.
\end{lemma}

\begin{proof}
    Observe that 
    \begin{equation}\label{eq:comm-lem-observation}
        \Tr(A_a^{ x}R-A_a^{ x}RA_a^{ x}R)=\frac{1}{2}\norm{A_a^{ x}R-RA_a^{ x}}_2^2.
    \end{equation}
    Since
    \begin{equation*}
    		\dsync(\mcS,\nu_A)=1-\expect{x\sim\nu_A}\sum_a\bra{\psi}A_a^x\otimes (A_a^x)^T\kp=\expect{x\sim\nu_A}\sum_a\frac{1}{\Tr(R)}\big(\Tr(RA_a^x)-\Tr(RA_a^xRA_a^x)\big)
    \end{equation*}
    by Ando's formula, we immediately see from \Cref{eq:comm-lem-observation} that $\dsync(\mcS,\nu_A)\leq \frac{1}{2}\gamma$. 
\end{proof}

We will now prove the main theorem of this section. Because the proof requires many computations, we will include several claims in the proof to guide the reader. They serve as announcements of the next step in the proof and indicate the general flow of the argument. For the remainder of this section and the next, we will be working with two probability distributions on $\mcX\times\mcX$; one determining the winning probability and one for the local dilations. The most natural situation is when both distributions are identical, but the Quantum Low Degree Test requires us to treat the case with distinct distributions.

\begin{theorem} \label{thm:constr-isom-gen-strat}
    Let $\mcS=(\kp,A,B)$ be a strategy for a perfect symmetric non-local game $G=(\mcX,\nu,\mcA,D)$ with winning probability $1-\epsilon$. Let $\rho_A$ and $\rho_B$ be the reduced density matrices on $H_A$ and $H_B$, respectively, $\nu_A$ the marginal of $\nu$ on $\mcX$ and $\delta=\dsync(\mcS,\nu_A)$. Let $\hat{\nu}$ be a symmetric probability distribution on $\mcX\times\mcX$ with marginal $\hat{\nu}_A$ on $\mcX$ and $c\geq 1$ such that $\hat{\nu}\leq c\nu$. Suppose that $G$ $(\kappa,\hat{\nu})$-PME-robustly self-tests the optimal strategy $\tilde{\mcS}=(\kpt,\tA,\tB)$ on $\tilde{H}_A\otimes \tilde{H}_B$. Then there exist Hilbert spaces $K_A$ and $K_B$ and isometries $V_A:H_A\rightarrow \tilde{H}_A\otimes K_A$ and $V_B:H_B\rightarrow \tilde{H}_B\otimes K_B$ such that for the game $G_{\hat{\nu}}=(\mcX,\hat{\nu},\mcA,D)$ the strategy $\hat{\mcS}=(\kp,V_A^*(\tA\otimes \Id_{K_A}) V_A,V_B^*(\tB\otimes \Id_{K_B}) V_B)$ has winning probability $\omega(G_{\hat{\nu}};\hat{\mcS})\geq 1-\mathcal{O}\big((c\cdot\id+\kappa)(2\epsilon+\poly(\delta))\big)$, and
    \begin{align*}
        \expect{x\sim\hat{\nu}_A}\sum_a\Tr\left(\left(A_a^{x}-V_A^*(\tA_a^x\otimes \Id_{K_A})V_A\right)^2\rho_A\right)&=\mcO\left((c\cdot\id+\kappa)^2(\poly(\delta)+2\epsilon)\right),\\
        \expect{y\sim\hat{\nu}_A}\sum_b\Tr\left(\left(B_b^{y}-V_B^*(\tB_a^x\otimes \Id_{K_B})V_B\right)^2\rho_B\right)&= \mcO\left((c\cdot\id+\kappa)^2(\poly(\delta)+2\epsilon)\right).
    \end{align*}
\end{theorem}

\begin{proof}
	Let $\mcS'=(\kp,A',B')$ be the projective strategy, obtained from $\mcS$, using Lemma \ref{lem:nearby-proj-strat} and let $\delta'=\dsync(\mcS',\nu_A)$. By Lemma \ref{lem:nearby-proj-strat}, we know that $\delta'=\mathcal{O}(\delta^{\frac{1}{8}})$. Let the measure $\mu$ on $\mbR_+$, the projections $P_{\lambda}$ of $H_A$ onto $H_{\lambda}$ and the family of strategies $\mcS_{\lambda}=(\ket{\psi_{\lambda}},P_{\lambda}AP_{\lambda})$ on $H_{\lambda}$ with reduced density matrix $\rho_{\lambda}$ be as in Theorem \ref{thm:Vidick}. Let $C,\ C'$ and $C^{\lambda}$ be the correlations of $\mcS$, $\mcS'$ and $\mcS_{\lambda}$, respectively. Let
    \begin{equation*}
        \expect{x,y\sim \nu}\sum_{a,b}|C_{x,y,a,b}-\int_0^{\infty}C^{\lambda}_{x,y,a,b}d\mu(\lambda)|=\alpha
    \end{equation*}
    and 
    \begin{equation*}
        \expect{x\sim \nu_A}\sum_a\int_0^{\infty}\frac{1}{\Tr(P_{\lambda})}\norm{A_a^{\prime x}P_{\lambda}-P_{\lambda}A_a^{\prime x}}_2^2d\mu(\lambda)=\beta.
    \end{equation*}
    \begin{claim}\label{clm:constr-Lambda}
    		The set $\Lambda\subset\R_+$, defined by
    		\begin{equation*}
    			\Lambda=\{\lambda\geq 0|\omega(\mcS_{\lambda})\geq 1-\sqrt{\alpha}-\epsilon \text{ and } \expect{x\sim \nu_A}\sum_a\frac{1}{\Tr(P_{\lambda})}\norm{A_a^{\prime x}P_{\lambda}-P_{\lambda}A_a^{\prime x}}_2^2\leq \sqrt{\beta}\},
    		\end{equation*}
    		satisfies $\mu(\Lambda)\geq 1-\sqrt{\alpha}-\sqrt{\beta}$.
    \end{claim}

    By Theorem \ref{thm:Vidick}(iv), Lemma \ref{lem:nearby-proj-strat}, Lemma \ref{lem:win-prob-est} and the triangle inequality, we know that $\alpha\leq \poly(\delta')+\mathcal{O}(\delta^{\frac{1}{8}})\leq \poly(\delta)$. Lemma \ref{lem:win-prob-est} tells us that for any $t>0$ we have that
    \begin{equation*}
        \mu(\{\lambda\geq 0|\omega(\mcS_{\lambda})\geq1-\epsilon-t\})\geq 1-\frac{\alpha}{t}.
    \end{equation*}
    Here we have used the observation that
    \begin{equation}\label{eq:int-est-observation}
    		t\mu(\{\lambda|f(\lambda)\geq s\})\leq \int_0^{\infty}f(\lambda)d\mu(\lambda).
    \end{equation}
    Next, Theorem \ref{thm:Vidick}, Remark \ref{rmk:comm-meas-proj} and Lemma \ref{lem:nearby-proj-strat} state that $\beta\leq 2\sqrt{2\delta'}\leq \poly(\delta)$, and we also have for every $s\geq 0$ that
    \begin{equation*}
        \mu(\{\lambda\geq0|\expect{x\sim \nu_A}\sum_a\frac{1}{\Tr(P_{\lambda})}\norm{A_a^{\prime x}P_{\lambda}-P_{\lambda}A_a^{\prime x}}_2^2\leq s\})\geq 1-\frac{\beta}{s}
    \end{equation*}
    by \Cref{eq:int-est-observation}.
    Choosing $s$ and $t$ is a trade-off between the strength of the bound and the measure of the set for which the bound holds. We choose $s=\sqrt{\beta}$ and $t=\sqrt{\alpha}$, but many choices are possible here. If we define
    \begin{equation*}
        \Lambda=\{\lambda\geq 0|\omega(\mcS_{\lambda})\geq 1-\sqrt{\alpha}-\epsilon \text{ and } \expect{x\sim \nu_A}\sum_a\frac{1}{\Tr(P_{\lambda})}\norm{A_a^{\prime x}P_{\lambda}-P_{\lambda}A_a^{\prime x}}_2^2\leq \sqrt{\beta}\},
    \end{equation*}
    we find that $\mu(\Lambda)\geq 1-\sqrt{\alpha}-\sqrt{\beta}$. From now on we will purely work with $\mcS_{\lambda}$ for $\lambda\in \Lambda$. For $\lambda\in \R_+\backslash\Lambda$ the crudest estimates will suffice.
  \clmprvd{\ref{clm:constr-Lambda}}
  
  The strategies $\mcS_{\lambda}$ for $\lambda\in \Lambda$ give an approximate decomposition of the strategy $\mcS$. However, this is not the right decomposition for us when we want to construct isometries. In essence, the current decomposition is a decomposition into smaller and smaller strategies. What we need is an approximate decomposition into `orthogonal' strategies, so our next aim is to construct this. Note that the projections $P_{\lambda}$ are ordered; $\lambda\leq \lambda'$ implies that $P_{\lambda}\geq P_{\lambda'}$, so the dimensions of the Hilbert spaces $H_{\lambda}$ are also ordered. We recursively define a partition $\{\Lambda_i\}_{i=1}^k$ of $\Lambda$ by setting $\lambda_1=\min(\Lambda)$,
    \begin{equation*}
        \Lambda_i=\{\lambda\in\Lambda|\dim H_{\lambda_i}\geq \dim H_{\lambda}>\frac{1}{2}\dim H_{\lambda_i}\}\text{ and }
        \lambda_{i+1}=\min(\Lambda\backslash\bigcup_{j=1}^i\Lambda_j),
    \end{equation*}
    where $k\in \N$ is such that $\Lambda=\bigcup_{j=1}^k\Lambda_j$. Since $P_{\lambda_i} \geq P_{\lambda_{i+1}}$, we can define projections $Q_i=P_{\lambda_i}-P_{\lambda_{i+1}}$, where $1\leq i\leq k-1$ and $Q_k=P_{\lambda_k}$ and set $H_i=Q_iH_A$. Let $\ket{\psi_i}\in H_i\otimes H_i$ be a maximally entangled state. We first aim to prove that the strategies $\mcS_i=(\ket{\psi_i},Q_iA'Q_i)$ have high winning probability. An intermediate step is to show that $Q_i$ almost commutes with $A_a^{\prime x}$ in an averaged sense.
    \begin{claim}\label{clm:comm-and-strat-props-of-Q}
    		For $1\leq i\leq k$ the $Q_i$ and $\mcS_i$ constructed above satisfy
    		\begin{equation*}
    			\expect{x\sim\nu_A}\sum_a\frac{1}{\Tr(Q_i)}\norm{A_a^{\prime x}Q_i-Q_iA_a^{\prime x}}_2^2\leq (\sqrt{2}+1)^2\sqrt{\beta}
    		\end{equation*}
    		and
    		\begin{equation*}
    			\dsync(\mcS_i,\nu_A)\leq \frac{1}{2}(\sqrt{2}+1)^2\sqrt{\beta}\text{ and }\omega(\mcS_i)\geq 1-2\epsilon-2\sqrt{\alpha}-6\sqrt{\beta}-8\sqrt{\frac{5}{2}+\sqrt{2}}\sqrt[4]{\beta}.
    		\end{equation*}
    \end{claim}
    Let $1\leq i\leq k-1$. Since $\lambda_i,\lambda_{i+1}\in \Lambda$, we know that
    \begin{equation*}
        \expect{x\sim \nu_A}\sum_a\frac{1}{\Tr(P_{\lambda})}\norm{A_a^{\prime x}P_{\lambda}-P_{\lambda}A_a^{\prime x}}_2^2\leq \sqrt{\beta}
    \end{equation*}
    for $\lambda\in \{\lambda_i,\lambda_{i+1}\}$. From this we get
    \begin{align*}
        \left(\expect{x\sim \nu_A}\sum_a\norm{A_a^{\prime x}Q_i-Q_iA_a^{\prime x}}_2^2\right)^{\frac{1}{2}}&\leq \sqrt{\Tr(P_{\lambda_i})}\sqrt[4]{\beta}+\sqrt{\Tr(P_{\lambda_{i+1}})}\sqrt[4]{\beta}\\
        &\leq (\sqrt{2}+1)\sqrt{\Tr(Q_i)}\sqrt[4]{\beta}
    \end{align*}
    by using the triangle inequality in the first inequality and the fact that $\dim H_{\lambda_i}\geq 2\dim H_{\lambda_{i+1}}$ in the second. Taking $\sqrt{\Tr(Q_i)}$ to the other side gives us
    \begin{equation}\label{eq:Qi-almost-commute}
        \left(\expect{x\sim \nu_A}\sum_a\frac{1}{\Tr(Q_i)}\norm{A_a^{\prime x}Q_i-Q_iA_a^{\prime x}}_2^2\right)^{\frac{1}{2}}\leq (\sqrt{2}+1)\sqrt[4]{\beta}.
    \end{equation}
    By Lemma \ref{lem:comm-proj-PVM}, we find that
    \begin{equation*}
        \dsync(\mcS_i,\nu_A)\leq\frac{(\sqrt{2}+1)^2}{2}\sqrt{\beta}.
    \end{equation*}

    Using \Cref{eq:Qi-almost-commute} and the fact that $\omega(\mcS_{\lambda_i}),\omega(\mcS_{\lambda_{i+1}})\geq 1-\sqrt{\alpha}-\epsilon$, we are now ready to show that $\omega(\mcS_i)\geq 1-\sqrt{\alpha}-\epsilon-(1+\frac{3}{2}\sqrt{2})\sqrt[4]{\beta}$. We will do this by showing that the winning probability of $\mcS_{\lambda_i}$ is close to the winning probability of the strategy $(\ket{\psi_{\lambda_i}},Q_iAQ_i+P_{\lambda_{i+1}}AP_{\lambda_{i+1}})$. Since $P_{\lambda_i} \geq P_{\lambda_{i+1}}$, we have that $Q_i$ and $P_{\lambda_{i+1}}$ are orthogonal. We compute
    \begin{align*}
        \expect{x\sim \nu_A}\sum_a\frac{1}{\Tr(P_{\lambda_i})}\Tr\left((P_{\lambda_i}A_a^{\prime x}P_{\lambda_i}-Q_iA_a^{\prime x}Q_i-P_{\lambda_{i+1}}A_a^{\prime x}P_{\lambda_{i+1}})^2P_{\lambda_i}\right)\\
        =\expect{x\sim \nu_A}\sum_a\frac{1}{\Tr(P_{\lambda_i})}\Tr\big(P_{\lambda_i}A_a^{\prime x}P_{\lambda_i}A_a^{\prime x}-Q_iA_a^{\prime x}Q_iA_a^{\prime x}&-P_{\lambda_{i+1}}A_a^{\prime x}P_{\lambda_{i+1}}A_a^{\prime x}\big),
    \end{align*}
    Using Lemma \ref{lem:comm-proj-PVM} for each of the three terms, we find that
    \begin{align*}
        \expect{x\sim \nu_A}\sum_a\frac{1}{\Tr(P_{\lambda_i})}\Tr\big(&(P_{\lambda_i}A_a^{\prime x}P_{\lambda_i}-Q_iA_a^{\prime x}Q_i-P_{\lambda_{i+1}}A_a^{\prime x}P_{\lambda_{i+1}})^2P_{\lambda_i}\big)\\
        &\leq(\frac{5}{2}+\sqrt{2})\sqrt{\beta}+\expect{x\sim \nu_A}\sum_a\big|\frac{1}{\Tr(P_{\lambda_i})}\Tr\big(P_{\lambda_i}A_a^{\prime x}-Q_iA_a^{\prime x}-P_{\lambda_{i+1}}A_a^{\prime x}\big)\big|\\
        &=(\frac{5}{2}+\sqrt{2})\sqrt{\beta}.
    \end{align*}
    By Lemma \ref{lem:subs-meas-ops-corr-est} and Lemma \ref{lem:win-prob-est}, we conclude that
    \begin{equation*}
        |\omega(\mcS_{\lambda_i})-\omega((\ket{\psi_{\lambda_i}},Q_iAQ_i+P_{\lambda_{i+1}}AP_{\lambda_{i+1}}))|\leq 6\sqrt{\beta}+8\sqrt{\frac{5}{2}+\sqrt{2}}\sqrt[4]{\beta}.
    \end{equation*}
    As 
    \begin{equation*}
        \omega((\ket{\psi_{\lambda_i}},Q_iAQ_i+P_{\lambda_{i+1}}AP_{\lambda_{i+1}}))=\frac{\Tr(Q_i)}{\Tr(P_{\lambda_i})}\omega(\mcS_i)+\frac{\Tr(P_{\lambda_{i+1}})}{\Tr(P_{\lambda_i})}\omega(\mcS_{\lambda_{i+1}}),
    \end{equation*}
    since it is a classical mixture of strategies, we find that
    \begin{align*}
        \omega(\mcS_i)&\geq \frac{\Tr(P_{\lambda_i})}{\Tr(Q_i)}(1-\epsilon-\sqrt{\alpha})-\frac{\Tr(P_{\lambda_{i+1}})}{\Tr(Q_i)}-6\sqrt{\beta}-8\sqrt{\frac{5}{2}+\sqrt{2}}\sqrt[4]{\beta}\\
        &\geq 1-2\epsilon-2\sqrt{\alpha}-6\sqrt{\beta}-8\sqrt{\frac{5}{2}+\sqrt{2}}\sqrt[4]{\beta}.
    \end{align*}
    Note that this proof works for $1\leq i\leq k-1$. However, $Q_k=P_{\lambda_k}$ and $\mcS_k=\mcS_{\lambda_k}$, so in that case the result holds immediately since $\lambda_i\in\Lambda$. \clmprvd{\ref{clm:comm-and-strat-props-of-Q}}
     
     For convenience, we define the constant
     \begin{equation*}
     	\gamma=2\epsilon+2\sqrt{\alpha}+6\sqrt{\beta}+8\sqrt{\frac{5}{2}+\sqrt{2}}\sqrt[4]{\beta}
     \end{equation*}
     going forward.
     
     Having obtained our approximate decomposition into orthogonal strategies that have high winning probability, we are in a position to use the PME-robust self-testing properties of $G$. Our goal will be to construct a single partial isometry $W$ that works well for each set of measurement operators $P_{\lambda_i}AP_{\lambda_i}$. Let $1\leq i\leq k$. Recall that $H_i=Q_iH_A$, consider $(B(H_i),\tr_i)$, where $\tr_i$ is the normalised trace on $B(H_i)$, and let $\delta_i=\dsync(\mcS_i,\nu_A)$. Since $\tilde{\mcS}$ and $\mcS_i$ are ME strategies, we can view them in the von Neumann algebra picture as strategies $(M,\tau^M,\tA)$ and $(N_i,\tau^{N_i},A_i)$, respectively. Since $G$ is a $(\kappa,\hat{\nu})$-PME-robust self-test for $\tilde{\mcS}$, we know by Lemmas \ref{lem:PME-robust-self-testing-vNA} and \ref{lem:ME-robust-self-testing-vNA} that $\tilde{\mcS}$ is a local $(\theta_i,\hat{\nu}_A)$-vNA-dilation of $\mcS_i$, with
     \begin{equation*}
         \theta_i=1700\left(\kappa\big(24\sqrt{\delta_i}+(1-\omega(\mcS_i))\big)+9c\delta_i\right)^2.
     \end{equation*}
     This means that there exists a finite dimensional von Neumann algebra $M_i$, a corresponding von Neumann algebra $(M\otimes M_i)^{\infty}$ with trace denoted by $\tau^{\infty}_i$, a projection $R_i\in (M\otimes M_i)^{\infty}$ of finite trace such that $N_i\cong R_i(M\otimes M_i)^{\infty}R_i$ and $\tr^{N_i}=\tau^{\infty}_i/\tau^{\infty}_i(R_i)$, and a partial isometry $V_i\in R_i(M\otimes M_i)^{\infty}I_{M\otimes M_i}$ such that
     \begin{equation*}
         \frac{1}{\Tr(Q_i)}\Tr(Q_i-V_iV_i^*)=\tau^{N_i}(R_i-V_iV_i^*)\leq \theta_i, \tau^{\infty}_i(I_{M\otimes M_i}-V_i^*V_i)\leq \theta_i
     \end{equation*}
     and
     \begin{align}\label{eq:quality-part-iso-Q-i}
         \expect{x\sim\hat{\nu}_A}\sum_a\frac{1}{\Tr(Q_i)}\Tr\Big((Q_iA_a^{\prime x}Q_i-V_i\tA_a^x&V_i^*)^2Q_i\Big) \\
         =\expect{x\sim\hat{\nu}_A}\sum_a\tau^{N_i}&\left(\left((A_i)_a^x-V_i\tA_a^xV_i^*\right)^2\right)\leq 2\theta_i+\frac{\theta_i}{1-\theta_i},\nonumber
     \end{align}
     by Lemma \ref{lem:vNA-close-strat-N-measure}. Note that we have identified $B(H_i)=N_i$ and $R_i(M\otimes M_i)^{\infty}R_i$, leading us to identify $Q_i$ and $R_i$. We have also identified $\tA_a^x\in M$ and $\tA_a^x\otimes \Id_{H_i}\otimes \Id_{B(\ell^2(\N))}\in (M\otimes M_i)^{\infty}$, and we will continue to do so in the rest of this proof. Lastly, remark that the expectations in \Cref{eq:quality-part-iso-Q-i} are with respect to $\hat{\nu}_A$ instead of $\nu_A$. From now on, we will mainly need expectations with respect to $\hat{\nu}_A$. We will freely use that we can obtain estimates on expectations with respect to $\hat{\nu}_A$ by multiplying expectations with respect to $\nu_A$ by $c$, as $\hat{\nu}_A\leq c\nu_A$. 

     Define
     \begin{equation*}
         \hat{M}=\bigoplus_{i=1}^kM_i,\ 
         \check{M}=(M\otimes \hat{M})^{\infty}\text{ and }         \hat{R}=\sum_{i=1}^kR_i.
     \end{equation*}
     We observe that $R\check{M}R\cong B(H_{\lambda_1})$, and we are now ready to define our desired partial isometry $V$ by 
     \begin{equation*}
         V=\sum_{i=1}^kV_i\in R\check{M}I_{M\otimes \hat{M}}.
     \end{equation*}
     
    The definition of the $\lambda_i$ allows us to compute that
    \begin{equation*}
    		\frac{1}{\Tr(P_{\lambda_1})}\Tr(P_{\lambda_1}-VV^*)=\frac{1}{\Tr(P_{\lambda_1})}\sum_{i=1}^k\Tr(Q_i-V_iV_i^*)\leq\sum_{i=1}^k\frac{\Tr(Q_i)}{\Tr(P_{\lambda_1})}\theta_i.
    \end{equation*}
    Let us define
    \begin{equation*}
        \theta=1700\left(\kappa\Big(24(1+\frac{1}{2}\sqrt{2})\sqrt[4]{\beta}+\gamma\Big)+\frac{9}{2}(\sqrt{2}+1)^2c\sqrt{\beta}\right)^2,
    \end{equation*}
    so $\theta_i\leq \theta$ for all $i$. We also have that $\sum_{i=j}^k\Tr(Q_i)=\Tr(P_{\lambda_j})$ for all $1\leq j\leq k$, so we find that 
    \begin{equation*}
    		\frac{1}{\Tr(P_{\lambda_j})}\Tr(P_{\lambda_j}-VV^*)\leq \theta.
    \end{equation*}
    If $\lambda\in\Lambda_j$, then it follows from this equation that
    \begin{equation}\label{eq:constr-iso-tr-est}
        \frac{1}{\Tr(P_{\lambda})}\Tr(P_{\lambda}-VV^*P_{\lambda}VV^*)= \frac{1}{\Tr(P_{\lambda})}\Tr((P_{\lambda_j}-VV^*)P_{\lambda})\leq 2\theta,
    \end{equation}
    where we use the cyclicity of the trace in the first step and Hölder together with $\Tr(P_{\lambda_j})\leq 2\Tr(P_{\lambda})$ in the second. 
    
    Next, we set out to show a relation between $A^{\lambda}$ and $V\tA V^*$. Our goal is to prove the following claim.
    \begin{claim}\label{clm:univ-part-iso}
    		The partial isometry $V$ is a `good' partial isometry on each subspace $H_{\lambda}$ with $\lambda\in\Lambda$, in the sense that
    		\begin{equation*}
    			\expect{x\sim\hat{\nu}_A}\sum_a\frac{1}{\Tr(P_{\lambda})}\Tr\left((A_a^{\prime x}-V\tA_a^xV^*)^2P_{\lambda}\right)\leq (\sqrt{2}+1)^2c\sqrt{\beta}+4\theta+\frac{2\theta}{1-\theta}.
    		\end{equation*}
    \end{claim}
    By the way $V$ is constructed, we have
    \begin{equation*}
    		V\tA_a^xV^*Q_i=V_i\tA_a^xV_i^*Q_i.
    \end{equation*}
    Let $\lambda\in \Lambda$ and let $1\leq j\leq k$ be such that $\lambda\in \Lambda_j$. By decomposing $P_{\lambda}$ into $P_{\lambda}-P_{\lambda_{j+1}},Q_{j+1},\dots,Q_k$, we find that
    \begin{align*}
    		\expect{x\sim\hat{\nu}_A}\sum_a\Tr\left((A_a^{\prime x}-V\tA_a^xV^*)^2P_{\lambda}\right)=&\expect{x\sim\hat{\nu}_A}\sum_a\Tr\left((A_a^{\prime x}-V_j\tA_a^xV_j^*)^2(P_{\lambda}-P_{\lambda_{j+1}})\right)\\
    		&+\expect{x\sim\hat{\nu}_A}\sum_a\sum_{i=j+1}^k\Tr\left((A_a^{\prime x}-V_i\tA_a^xV_i^*)^2Q_i\right).
    \end{align*}
    Since $(A_a^{\prime x}-V_j\tA_a^xV_j^*)^2$ is positive and $P_{\lambda}\leq P_{\lambda_j}$, we have the estimate
    \begin{equation*}
    		\expect{x\sim\hat{\nu}_A}\sum_a\Tr\left((A_a^{\prime x}-V_j\tA_a^xV_j^*)^2(P_{\lambda}-P_{\lambda_{j+1}})\right)\leq \expect{x\sim\hat{\nu}_A}\sum_a\Tr\left((A_a^{\prime x}-V_j\tA_a^xV_j^*)^2Q_j\right),
    \end{equation*}
    so
    \begin{equation*}
    		\expect{x\sim\hat{\nu}_A}\sum_a\Tr\left((A_a^{\prime x}-V\tA_a^xV^*)^2P_{\lambda}\right)\leq \expect{x\sim\hat{\nu}_A}\sum_a\sum_{i=j}^k\Tr\left((A_a^{\prime x}-V_i\tA_a^xV_i^*)^2Q_i\right).
    \end{equation*}
    Note that $V_i^*Q_i=V_i^*$. Using this knowledge, expanding the inner term of the right hand side gives
    \begin{equation*}
    		\Tr\left((A_a^{\prime x}-V_i\tA_a^xV_i^*)^2Q_i\right)=\Tr\left(A_a^{\prime x}Q_i-A_a^{\prime x}V_i\tA_a^xV_i^*-V_i\tA_a^xV_i^*A_a^{\prime x}+V_i\tA_a^xV_i^*V_i\tA_a^xV_i^*\right).
    \end{equation*}
    Note that
    \begin{equation*}
        \dsync(\mcS_i,\hat{\nu}_A)=\expect{x\sim\hat{\nu}_A}\sum_a\frac{1}{\Tr(Q_i)}\Tr(A_a^{\prime x}Q_i-A_a^{\prime x}Q_iA_a^{\prime x}Q_i).
    \end{equation*}
    By using Claim \ref{clm:comm-and-strat-props-of-Q} and Lemma \ref{lem:comm-proj-PVM}, we find
    \begin{equation*}
    		\expect{x\sim\hat{\nu}_A}\sum_a\Tr\big((A_a^{\prime x}-V_i\tA_a^xV_i^*)^2Q_i\big)\leq \Tr(Q_i)\frac{(\sqrt{2}+1)^2}{2}c\sqrt{\beta}+\expect{x\sim\hat{\nu}_A}\sum_a\norm{Q_iA_a^{\prime x}Q_i-V_i\tA_a^xV_i^*}_2^2.
    \end{equation*}
    Summing the first term on the right hand side gives
    \begin{equation*}
    		\sum_{i=j}^k\Tr(Q_i)\frac{(\sqrt{2}+1)^2}{2}c\sqrt{\beta}=\Tr(P_{\lambda_j})\frac{(\sqrt{2}+1)^2}{2}c\sqrt{\beta}\leq \Tr(P_{\lambda})(\sqrt{2}+1)^2c\sqrt{\beta}.
    \end{equation*}
    Summing the second term we use \Cref{eq:quality-part-iso-Q-i} and an analogous summation estimate to see that
    \begin{equation*}
    		\expect{x\sim\hat{\nu}_A}\sum_a\sum_{i=j}^k\norm{Q_iA_a^{\prime x}Q_i-V_i\tA_a^xV_i^*}_2^2\leq 2\Tr(P_{\lambda})\left(2\theta+\frac{\theta}{1-\theta}\right).
    \end{equation*}
    All in all we conclude that
    \begin{equation*}
    		\expect{x\sim\hat{\nu}_A}\sum_a\frac{1}{\Tr(P_{\lambda})}\Tr\left((A_a^{\prime x}-V\tA_a^xV^*)^2P_{\lambda}\right)\leq (\sqrt{2}+1)^2c\sqrt{\beta}+4\theta +\frac{2\theta}{1-\theta}
    \end{equation*}
    by pulling $\Tr(P_{\lambda})$ to the other side.\clmprvd{\ref{clm:univ-part-iso}}
	
	We are now ready to finish the proof of the theorem. What is left to do is to pad $V$ by an arbitrary partial isometry which has full support on the kernel of $V$. This only introduces a small error because $\mu(\R_+\backslash\Lambda)\leq \sqrt{\alpha}+\sqrt{\beta}$ by Claim \ref{clm:constr-Lambda} and because $\Tr(P_{\lambda_1}-VV^*)\leq \theta\Tr(P_{\lambda_1})$. Let $P^{\perp}=\Id_{B(H_A)}-VV^*$, $R^{\perp}\in \check{M}$ a projection such that $R^{\perp}\check{M}R^{\perp}\cong P^{\perp}B(H_A)P^{\perp}$ and $\hat{R}R^{\perp}=0$. We identify $P^{\perp}B(H_A)P^{\perp}$ and $R^{\perp}\check{M}R^{\perp}$ and let $M^{\perp}\in \hat{M}\otimes B(\ell^2(\N))$ be a von Neumann algebra and $V^{\perp}\in R^{\perp}\check{M}I_{M\otimes M^{\perp}}$ a partial isometry such that $V^{\perp}(V^{\perp})^*=R^{\perp}$ and $I_{\hat{M}}M^{\perp}I_{\hat{M}}=\{0\}$. Let $W=V+V^{\perp}\in \check{M}I_{M\otimes (\hat{M}\oplus M^{\perp})}$. First, observe that
	\begin{equation*}
		\norm{\rho_A-VV^*\rho_A VV^*}_1\leq\norm{\int_{\lambda\in \R_+\backslash \Lambda}\rho_{\lambda}d\mu(\lambda)}_1+\int_{\lambda\in\Lambda}\norm{\rho_{\lambda}-VV^*\rho_{\lambda}VV^*}_1d\mu(\lambda)\leq \sqrt{\alpha}+\sqrt{\beta}+2\theta
	\end{equation*}
	by construction of $\Lambda$ and \Cref{eq:constr-iso-tr-est}, and
	\begin{equation*}
		\norm{\sum_a(A_a^{\prime x}-\tilde{W}\tA_a^x\tilde{W}^*)^2}_{\infty}\leq 4
	\end{equation*}
	by Lemma \ref{lem:op-norm-bound-game-poly}. This allows us to compute that 
	\begin{align*}
		\expect{x\sim\hat{\nu}_A}\sum_a\Tr&\left((A_a^{\prime x}-\tilde{W}\tA_a^x\tilde{W}^*)^2\rho_A\right)\\
        &\leq 4(\sqrt{\alpha}+\sqrt{\beta}+2\theta)+\expect{x\sim\hat{\nu}_A}\sum_a\int_{\lambda\in \Lambda}\Tr\left((A_a^{\prime x}-\tilde{W}\tA_a^x\tilde{W}^*)^2VV^*\rho_{\lambda}VV^*\right)d\mu(\lambda)\\
		&\leq 4(\sqrt{\alpha}+\sqrt{\beta}+2\theta)+(\sqrt{2}+1)^2c\sqrt{\beta}+4\theta+\frac{2\theta}{1-\theta}.
	\end{align*}
	By the construction of $\mcS'$ using Lemma \ref{lem:nearby-proj-strat}, the triangle inequality and the fact that the estimate becomes trivial if $\theta\geq 0.28$, we obtain
	\begin{equation*}
		\expect{x\sim\hat{\nu}_A}\sum_a\Tr\left((A_a^{x}-\tilde{W}\tA_a^x\tilde{W}^*)^2\rho_A\right)= \mcO(\sqrt{\alpha}+c\sqrt{\beta})+15\theta,
	\end{equation*}
	showing that $\tilde{W}$ is a good isometry. Now define $V_A=\tilde{W}$, and by repeating the proof for the B-side, we obtain an isometry $V_B$ with the corresponding properties. The spaces $K_A$ and $K_B$ are found by taking the Hilbert space upon which $\hat{M}\oplus M^{\perp}$ acts. We now define the strategy $\hat{\mcS}=(\kp,V_A\tA V_A^*,V_B\tB V_B^*)$. For an estimate on the success probability of $\hat{\mcS}$, we can just use Lemma \ref{lem:subs-meas-ops-corr-est} to show that $\omega(G_{\hat{\nu}};\hat{\mcS})= 1-c\epsilon-\mcO(\delta+\sqrt{\sqrt{\alpha}+c\sqrt{\beta}+\theta})$, so $\omega(G_{\hat{\nu}};\hat{\mcS})= 1-\mcO\big((c\cdot \id+\kappa)(\poly(\delta)+2\epsilon)\big)$.
\end{proof}

\begin{theorem} \label{thm:lifting-PME-assumtion-if-spec-gap}
    Let $\mcS$ be a strategy for a perfect symmetric non-local game $G=(\mcX,\nu,\mcA,D)$ with winning probability $1-\epsilon$. Let $\nu_A$ be the marginal of $\nu$ on $\mcX$ and let $\delta=\dsync(\mcS,\nu_A)$. Let $\hat{\nu}$ be a symmetric probability distribution on $\mcX\times\mcX$ with marginal $\hat{\nu}_A$ on $\mcX$ and $c\geq 1$ such that $\hat{\nu}\leq c\nu$. Define $G_{\hat{\nu}}=(\mcX,\hat{\nu},\mcA,D)$. Suppose that $G$ $(\kappa,\hat{\nu})$-PME-robustly self-tests a strategy $\tilde{\mcS}$ and that the game polynomial $T_{G_{\hat{\nu}},\tilde{\mcS}}$ has spectral gap $\alpha$. Then $\tilde{\mcS}$ is a local $\Big(\mathcal{O}\big(\frac{1}{\sqrt{\alpha}}\sqrt{c\cdot\id+\kappa}(2\epsilon+\poly(\delta))\big),\hat{\nu}\Big)$-dilation of $\mcS$.
\end{theorem}
\begin{proof}
    Let the optimal strategy $\tilde{\mcS}$ be given by $\tilde{\mcS}=(\kpt,\tA,\tB)$. Because $\mcS$ and $G$ satisfy the requirements of Theorem \ref{thm:constr-isom-gen-strat}, we know that there exist Hilbert spaces $K_A$ and $K_B$ and isometries $V_A:H_A\rightarrow \tH_A\otimes K_A$ and $V_B:H_B\rightarrow \tH_B\otimes K_B$ such that $\hat{\mcS}=(\kp,V_A^*(\tA\otimes \Id_{K_A})V_A,V_B^*(\tB\otimes \Id_{K_B})V_B)$ has winning probability $\omega(\hat{\mcS})= 1-\mcO\big((c\cdot\id+\kappa)(2\epsilon+\poly(\delta))\big)$ and 
    \begin{align}
        \expect{x\sim\hat{\nu}_A}\sum_a\Tr\left((A_a^{x}-V_A^*\tA_a^xV_A)^2\rho_A\right)&= \mcO\big((c\cdot\id+\kappa)^2(\poly(\delta)+2\epsilon)\big),\label{eq:input-POVM-est-A-thm-self-test-implication}\\
        \expect{y\sim\hat{\nu}_A}\sum_b\Tr\left((B_b^{y}-V_B^*\tB_a^xV_B)^2\rho_B\right)&= \mcO\big((c\cdot\id+\kappa)^2(\poly(\delta)+2\epsilon)\big).\label{eq:input-POVM-est-B-thm-self-test-implication}
    \end{align}
    Here we have identified $\tA_a^x\in B(\tH_A)$ and $\tA_a^x\otimes \Id_{K_A}\in B(\tH_A\otimes K_A)$ and $\tB_b^y\in B(\tH_B)$ and $\tB_b^y\otimes \Id_{K_B}\in B(\tH_B\otimes K_B)$, and we will continue to do so in this proof. From the winning probability estimate it follows that 
    \begin{align*}
        \bra{\psi}\left(\sum_{x,y}\hat{\nu}(x,y)\sum_{a,b}D(a,b|x,y)(V_A\otimes V_B)^*(\tA_a^x\otimes \tB_b^y)(V_A\otimes V_B)\right)\kp&\\
        \geq 1-\mathcal{O}\big((c\cdot\id+\kappa)&(2\epsilon+\poly(\delta))\big),
    \end{align*}
    which implies that
    \begin{equation*}
        \bra{\psi}(V_A\otimes V_B)^*(T_{G_{\hat{\nu}},\tilde{\mcS}}\otimes \Id_{K_A\otimes K_B})(V_A\otimes V_B)\kp\geq 1-\mathcal{O}\big((c\cdot\id+\kappa)(2\epsilon+\poly(\delta))\big).
    \end{equation*}
    Let $P\in B(\tH_A\otimes \tH_B\otimes K_A\otimes K_B)$ be the projection onto the 1-eigenspace of $T_{G,\tilde{\mcS}}\otimes \Id_{K_A\otimes K_B}$, which is the largest eigenvalue since $G$ is perfect. Then we have
    \begin{equation*}
        \bra{\psi}(V_A\otimes V_B)^*(1-P)(T_{G_{\hat{\nu}},\tilde{\mcS}}\otimes \Id_{K_A\otimes K_B})(1-P)(V_A\otimes V_B)\kp\leq (1-\alpha)\norm{(1-P)(V_A\otimes V_B)\kp}^2
    \end{equation*}
    and 
    \begin{equation*}
        \bra{\psi}(V_A\otimes V_B)^*P(T_{G_{\hat{\nu}},\tilde{\mcS}}\otimes \Id_{K_A\otimes K_B})P(V_A\otimes V_B)\kp= \norm{P(V_A\otimes V_B)\kp}^2,
    \end{equation*}
    so
    \begin{equation*}
        \bra{\psi}(V_A\otimes V_B)^*(T_{G_{\hat{\nu}},\tilde{\mcS}}\otimes \Id_{K_A\otimes K_B})(V_A\otimes V_B)\kp\leq 1-\alpha \norm{(1-P)(V_A\otimes V_B)\kp}^2.
    \end{equation*}
    This directly implies that 
    \begin{equation*}
        \norm{(1-P)(V_A\otimes V_B)\kp}\leq \mathcal{O}\big(\frac{1}{\sqrt{\alpha}}\sqrt{c\cdot\id+\kappa}(2\epsilon+\poly(\delta))\big).
    \end{equation*}
    Since the 1-eigenspace of $T_{G,\tilde{\mcS}}$ is one-dimensional, we know that there exists a state $\ket{aux}\in K_A\otimes K_B$ such that
    \begin{equation*}
        \frac{1}{\norm{P(V_A\otimes V_B)\kp}}P(V_A\otimes V_B)\kp=\kpt\otimes \ket{aux}.
    \end{equation*}
    By the triangle inequality, we find that
    \begin{equation*}
        \norm{(V_A\otimes V_B)\kp-\kpt\otimes\ket{aux}}\leq \mathcal{O}\big(\frac{1}{\sqrt{\alpha}}\sqrt{c\cdot\id+\kappa}(2\epsilon+\poly(\delta))\big).
    \end{equation*}
    Next, note that
    \begin{equation}\label{eq:est-1-thm-self-test-implication}
        \Tr\big((A_a^x-V_A^*\tA_a^xV_A)^2\rho_A\big)= \norm{(V_A(A_a^x-V_A^*\tA_a^xV_A)\otimes V_B)\kp}^2
    \end{equation}
    since $V_A$ is an isometry. Moreover, observe that
    \begin{align}\label{eq:est-2-thm-self-test-implication}
        \sum_a\norm{(V_AV_A^*\tA_a^x\otimes \Id_{\tH_B\otimes K_B})(V_A\otimes V_B\kp-\kpt\otimes&\ket{aux})}^2\\
        &\leq\norm{(V_A\otimes V_B)\kp-\kpt\otimes\ket{aux}}^2,\nonumber
    \end{align}
    using the fact that $V_AV_A^*$ is a contraction and 
    \begin{equation}\label{eq:sum-square-POVM-est}
        \sum_a(\tA_a^x)^2\leq \sum_a\tA_a^x=\Id_{\tH_A\otimes K_A}.
    \end{equation}
    Consequently, our goal will be to provide an estimate for
    \begin{equation*}
        \sum_a\norm{((V_AV_A^*\tA_a^x-\tA_a^x)\otimes \Id_{\tH_B\otimes K_B})\kpt\otimes\ket{aux}}^2.
    \end{equation*}
    First, we get rid of the sum over $a$. Let $m=\dim(\tH_A)$.  Since $\kpt$ is maximally entangled, we have
    \begin{align*}
        \sum_a\norm{((&V_AV_A^*\tA_a^x-\tA_a^x)\otimes \Id_{\tH_B\otimes K_B})\kpt\otimes\ket{aux}}^2\\
        &=\sum_a\Tr\left(((V_AV_A^*\tA_a^x-\tA_a^x)^*(V_AV_A^*\tA_a^x-\tA_a^x)\otimes \Id_{\tH_B\otimes K_B})\kpt\otimes\ket{aux}\bra{\tilde{\psi}}\otimes\bra{aux}\right)\\
        &=\frac{1}{m}\sum_a\Tr\left(((V_AV_A^*\tA_a^x-\tA_a^x)^*(V_AV_A^*\tA_a^x-\tA_a^x)\otimes \Id_{ K_B})(\Id_{\tH_A}\otimes\ket{aux}\bra{aux})\right)\\
        &\leq\frac{1}{m}\Tr\left(((V_AV_A^*-\Id_{\tH_A\otimes K_A})^*(V_AV_A^*-\Id_{\tH_A\otimes K_A})\otimes \Id_{K_B})(\Id_{\tH_A}\otimes\ket{aux}\bra{aux})\right)\\
        &=\norm{((V_AV_A^*-\Id_{\tH_A\otimes K_A})\otimes \Id_{\tH_B\otimes K_B})\kpt\otimes\ket{aux}}^2,
    \end{align*}
    where we use the cyclicity of the trace, the fact that $\Tr_{\tH_B}(\ket{\tilde{\psi}}\bra{\tilde{\psi}})=\frac{1}{m}\Id_{\tH_A}$, and the standard estimate (\ref{eq:sum-square-POVM-est}) for the inequality. By adding $0$ in a clever way, we find
    \begin{align*}
        \norm{((V_AV_A^*-\Id_{\tH_A\otimes K_A})&\otimes \Id_{\tH_B\otimes K_B})\kpt\otimes\ket{aux}}^2\\
        =&\lVert((V_AV_A^*-\Id_{\tH_A\otimes K_A})\otimes \Id_{\tH_B\otimes K_B})(\kpt\otimes\ket{aux}-(V_A\otimes V_B)\kp)\\
        &+((V_AV_A^*-\Id_{\tH_A\otimes K_A})\otimes \Id_{\tH_B\otimes K_B})(V_A\otimes V_B)\kp\rVert^2.
    \end{align*}
    Since $V_A$ is an isometry, $V_A^*V_A=\Id_{\tH_A\otimes K_A}$, so the second term is zero. Furthermore, 
    \begin{equation*}
        0\leq V_AV_A^*\leq \Id_{\tH_A\otimes K_A},
    \end{equation*}
    so 
    \begin{equation*}
        \norm{\Id_{\tH_A\otimes K_A}-V_AV_A^*}_{\infty}\leq 1.
    \end{equation*}
    Consequently,
    \begin{equation*}
        \norm{((V_AV_A^*-\Id_{\tH_A\otimes K_A})\otimes \Id_{\tH_B\otimes K_B})\kpt\otimes\ket{aux}}\leq \norm{(V_A\otimes V_B)\kp-\kpt\otimes\ket{aux}}.
    \end{equation*}
    This equation, together with Equations (\ref{eq:est-1-thm-self-test-implication}) and (\ref{eq:est-2-thm-self-test-implication}), the starting estimate (\ref{eq:input-POVM-est-A-thm-self-test-implication}) and the triangle inequality yields
    \begin{align*}
        \expect{x\sim \hat{\nu}_A}\Big(\sum_a\norm{(V_A\otimes V_B)(A_a^x\otimes \Id_{H_B})\kp-(\tA_a^x\otimes &\Id_{\tH_B\otimes K_B})\kpt\otimes\ket{aux}}^2\Big)^{\frac{1}{2}}\\
        &\leq \mathcal{O}\big(\frac{1}{\sqrt{\alpha}}\sqrt{c\cdot\id+\kappa}(2\epsilon+\poly(\delta))\big).
    \end{align*}
    By doing the analogous proof for the B-side, which gives us the same $\ket{aux}$, we have obtained the three inequalities that show that $\tilde{\mcS}$ is a local $(\mathcal{O}\big(\alpha^{-1/2}\sqrt{c\cdot\id+\kappa}(2\epsilon+\poly(\delta))\big),\hat{\nu})$-dilation of $\mcS$.
\end{proof}

\begin{corollary}\label{cor:lifting-PME-assumption-if-spec-gap-beta-sync}
    Let $G=(\mcX,\nu,\mcA,D)$ be a perfect $\beta$-synchronous non-local game  $(\kappa,\hat{\nu})$-PME-robust self-testing $\tilde{\mcS}$ and let $G_{\hat{\nu}}=(\mcX,\hat{\nu},\mcA,D)$. Suppose that the game polynomial $T_{G_{\hat{\nu}},\tilde{\mcS}}$ has spectral gap $\alpha$ and $\hat{\nu}\leq c\nu$ for $c\geq 1$. Then $G$ is a $(\kappa',\hat{\nu})$-robust self-test for $\kappa'(\epsilon)=\mcO(\alpha^{-1/2}\sqrt{c\cdot\id+\kappa}(\poly(\beta^{-1}\epsilon))$.
\end{corollary}
\begin{proof}
    This is immediate from the previous theorem after realising that the asynchronicity is bounded above by $\beta^{-1}\epsilon$ for a $\beta$-synchronous game as shown in Lemma \ref{lem:value-dsync}.
\end{proof}
\begin{remark}
    The constants hidden in $\mcO$ and $\poly$ are universal and do not depend on any property of the game.
\end{remark}

\section{Spectral gap for robust synchronous self-tests}\label{sec5}
In the previous section we proved that every $\beta$-synchronous PME-robust self-test whose game polynomial has spectral gap is a robust self-test. This raises the question whether the spectral gap condition is necessary. Are there examples of such games for which the spectral gap can be arbitrarily small, or is there some lower bound on the spectral gap? It turns out that a lower bound exists if the non-local game is a $\kappa$-PME-robust self-test, i.e. if the probability distribution for robust self-test is the same as the probability distribution of the game. This is a consequence of several elements of the proof of Theorem \ref{thm:constr-isom-gen-strat} under the additional assumption that the Hilbert space $H_A\otimes H_B$ for a strategy $\mcS$ equals the Hilbert space $\tH_A\otimes \tH_A$ for the strategy $\tilde{\mcS}$ self-tested by the game. We rely on a dimension estimate, given in Lemma \ref{lem:contained-ME-strat-dim-est}, to prove Lemma \ref{lem:nearby-max-ent-strat}, stating that a PME strategy exists with high winning probability such that its reduced density matrix is close to the one of $\mcS$. From there we are able to prove the desired theorem.

\begin{lemma}\label{lem:contained-ME-strat-dim-est}
    Let $\kp\in H\otimes H$ and $\kpt\in \tH\otimes\tH$ be maximally entangled states on Hilbert spaces such that $\dim(H)\leq \dim(\tH)$ and $\epsilon\geq 0$. Suppose that there exist isometries $V_A:H\rightarrow \tH\otimes K_A$ and $V_B:H\rightarrow \tH\otimes K_B$ and a state $\ket{aux}\in K_A\otimes K_B$ such that
    \begin{equation*}
        \norm{(V_A\otimes V_B)\kp-\kpt\otimes\ket{aux}}\leq \epsilon.
    \end{equation*}
    Then $\dim(H)\geq (1-\epsilon^2)\dim(\tH)$.
\end{lemma}
\begin{proof}
    Let $n=\dim(H)$ and $m=\dim(\tH)$, so we know that $n\leq m$. Let $d=\dim(\tH\otimes K_A)$ and
    \begin{equation*}
        (V_A\otimes V_B)\kp=\sum_{i=1}^d\lambda_i\ket{u_i}\otimes \ket{v_i}\text{ and } \kpt\otimes \ket{aux}=\sum_{i=1}^d\mu_i\ket{\tilde{u}_i}\otimes \ket{\tilde{v}_i}
    \end{equation*}
    be Schmidt decompositions, where $\lambda_1\geq\lambda_2\geq\dots\geq\lambda_d$ and $\mu_1\geq\mu_2\geq \dots\geq\mu_d$. Since $\kp$ is maximally entangled on $H\otimes H$, we know that $\lambda_i=1/\sqrt{n}$ if $i\leq n$ and $\lambda_i=0$ otherwise. Since $\kpt$ is maximally entangled on $\tH\otimes\tH$, we know that $\mu_1\leq 1/\sqrt{m}$. Consequently,
    \begin{equation*}
        \norm{(\lambda_i)_i-(\mu_i)_i}_2^2\geq n(\frac{1}{\sqrt{n}}-\frac{1}{\sqrt{m}})^2+\frac{m-n}{m}\geq \frac{m-n}{m}.
    \end{equation*}
    Combining this with Lemma \ref{lemma:Schmidt-bound}, we conclude that 
    \begin{equation*}
        m\cdot (1-\epsilon^2)\leq n\leq m.\qedhere
    \end{equation*}
\end{proof}

\begin{lemma}\label{lem:nearby-max-ent-strat}
    Let $\mcS=(\kp,A,B)$ be a strategy for a perfect symmetric non-local game $G=(\mcX,\nu,\mcA,D)$ with winning probability $1-\epsilon$. Let $\rho_A$ be the reduced density matrix on $H_A$, $\nu_A$ the marginal of $\nu$ on $\mcX$, and $\delta=\dsync(\mcS,\nu_A)$. Let $\hat{\nu}$ be a probability distribution on $\mcX\times\mcX$. Suppose that $G$ $(\kappa,\hat{\nu})$-PME-robustly self-tests the optimal PME strategy $\tilde{\mcS}=(\kpt,\tA)$ on $\tH_A\otimes \tH_A$ and that $\dim(\tH_A)=\dim(H_A)$. Then there exists a PME strategy $\hat{\mcS}_A=(\ket{\hat{\psi}},\hat{A})$ on $\hat{H}_A\subset H_A$ with reduced density matrix $\hat{\rho}_A$, such that
    \begin{align}
        \norm{\rho_A-\hat{\rho}_A}_1&\leq \mathcal{O}\big(\kappa(\epsilon+\poly(\delta))^2+\poly(\delta)\big)\text{ and}\label{eq:lem-nearby-max-ent-strat-rho-est}\\
        \omega(G;\hat{\mcS}_A)&\geq 1-\epsilon-\poly(\delta).\label{eq:lem-nearby-max-ent-strat-succ-prob-est}
    \end{align}
\end{lemma}
\begin{remark}
    In both the above and the subsequent lemma, $\hat{\nu}$ does not affect the conclusions. This happens because the condition that $G$ is a $(\kappa,\hat{\nu})$-PME-robust self-test is slightly stronger than necessary, since we do not need the equations in the definition of a local dilation concerning the measurement operators. As we are not aware of any use of the slightly more general statement, we refrain from introducing the nomenclature required to formally state this more general statement.
\end{remark}

\begin{proof}
	Let $\mcS'=(\kp,A',B')$ be the projective strategy given by Lemma \ref{lem:nearby-proj-strat} and let $\delta'=\dsync(\mcS',\nu_A)$. By Lemma \ref{lem:nearby-proj-strat}, we know that $\delta'=\mathcal{O}(\delta^{\frac{1}{8}})$. Let the probability measure $\mu$ on $\mbR_+$, the projections $P_{\lambda}$ of $H_A$ onto $H_{\lambda}$ and the family of strategies $\mcS_{\lambda}'=(\ket{\psi_{\lambda}},A^{\lambda})$ on $H_{\lambda}$ with reduced density matrix $\rho_{\lambda}$ be as in Theorem \ref{thm:Vidick}. Let $C,\ C'$ and $C^{\lambda}$ be the correlations of $\mcS$, $\mcS'$ and $\mcS_{\lambda}'$, respectively. Let
    \begin{equation*}
        \expect{x,y\sim \nu}\sum_{a,b}|C_{x,y,a,b}-\int_0^{\infty}C^{\lambda}_{x,y,a,b}d\mu(\lambda)|=\alpha.
    \end{equation*}
    
    By Theorem \ref{thm:Vidick}(v), Lemma \ref{lem:nearby-proj-strat} and the triangle inequality, we know that $\alpha\leq \poly(\delta')+\mathcal{O}(\delta^{\frac{1}{8}})\leq \poly(\delta)$. Lemma \ref{lem:win-prob-est} tells us that for any $t>0$ we have that
    \begin{equation*}
        \mu(\{\lambda\geq 0|\omega(\mcS_{\lambda}')\geq1-\epsilon-t\})\geq 1-\frac{\alpha}{t}.
    \end{equation*}
    As in the proof of Theorem \ref{thm:constr-isom-gen-strat}, choosing $t$ is a trade-off between the strength of the bound and the measure of the set for which the bound holds. We choose $t=\sqrt{\alpha}$, but many choices are possible here. If we define
    \begin{equation*}
        \Lambda=\{\lambda\geq 0|\omega(\mcS_{\lambda}')\geq 1-\sqrt{\alpha}-\epsilon\},
    \end{equation*}
    we find that $\mu(\Lambda)\geq 1-\sqrt{\alpha}$.

    Let $\lambda\in\Lambda$ and let $m$ be the dimension of $\tH_A$. Since $\omega(\mcS_{\lambda}')\geq 1-\sqrt{\alpha}-\epsilon$ and $G$ is a $(\kappa,\hat{\nu})$-PME-robust self-test, we know that $\tilde{\mcS}$ is a local $(\kappa(\sqrt{\alpha}+\epsilon),\hat{\nu})$-dilation of $\mcS_{\lambda}'$. Consequently, there exist isometries $V_A:H_\lambda\rightarrow \tH_A\otimes K_A$ and $V_B: H_{\lambda}\rightarrow \tH_A\otimes K_B$ and a unit vector $\ket{aux}\in K_A\otimes K_B$ such that 
    \begin{equation*}
        \norm{(V_A\otimes V_B)\ket{\psi_{\lambda}}-\kpt\otimes\ket{aux}}\leq \kappa(\sqrt{\alpha}+\epsilon).
    \end{equation*}
    By Lemma \ref{lem:contained-ME-strat-dim-est}, we know that $\dim(H_{\lambda})\geq (1-\kappa(\sqrt{\alpha}+\epsilon)^2)m$.
    
    Let $P_{\lambda}$ be the projection of $H_A$ onto $H_{\lambda}$, so $\rho_{\lambda}=\Tr(P_{\lambda})^{-1}P_{\lambda}$. The $P_{\lambda}$ are constructed in Theorem \ref{thm:Vidick} as the spectral projections of $\rho_A$, so $P_{\lambda}=\chi_{\geq \lambda}(\rho_A)$. This implies that $P_{\lambda_1}\leq P_{\lambda_2}$ if and only if $\lambda_1\geq \lambda_2$. Let $\lambda_0=\min(\Lambda)$, which exists because $\Lambda$ is closed. Using the estimate of $\dim(H_{\lambda})$ and the order on the projections $P_{\lambda}$, we now calculate that
\begin{align*}
        \norm{\rho_{\lambda}-\rho_{\lambda_0}}_1=&\norm{\frac{1}{\Tr(P_{\lambda})}P_{\lambda}-\frac{1}{\Tr(P_{\lambda_0})}P_{\lambda_0}}_1\\
	\leq& \left(\frac{1}{\Tr(P_{\lambda})}-\frac{1}{\Tr(P_{\lambda_0})}\right)\norm{P_{\lambda}}_1+\frac{1}{\Tr(P_{\lambda_0})}\norm{P_{\lambda}-P_{\lambda_0}}_1\\
	\leq& \frac{1}{m}\left(\frac{1}{1-\kappa(\epsilon+\sqrt{\alpha})^2}-1\right)m+\frac{m\left(1-\left(1-\kappa(\epsilon+\sqrt{\alpha})^2\right)\right)}{m(1-\kappa(\epsilon+\sqrt{\alpha})^2)}\\
    = & \frac{2\kappa(\epsilon+\sqrt{\alpha})^2}{1-\kappa(\epsilon+\sqrt{\alpha})^2}\leq 4\kappa(\epsilon+\sqrt{\alpha})^2, 
    \end{align*}
    if $\kappa(\epsilon+\sqrt{\alpha})^2\leq 1/2$. 
    By directly estimating that $\norm{\rho_{\lambda}-\rho_{\lambda_0}}_1\leq 4\kappa(\epsilon+\sqrt{\alpha})^2$ if $\kappa(\epsilon+\sqrt{\alpha})^2\geq 1/2$, we know that 
    \begin{equation*}
        \norm{\rho_{\lambda}-\rho_{\lambda_0}}_1\leq 4\kappa(\epsilon+\sqrt{\alpha})^2
    \end{equation*}
    always holds. Since $\mu(\Lambda)\geq 1-\sqrt{\alpha}$ and 
    \begin{equation*}
        \rho_A=\int_0^{\infty}\rho_{\lambda}d\mu(\lambda),
    \end{equation*}
    it follows that $\norm{\rho_A-\rho_{\lambda_0}}_1\leq 4\kappa(\epsilon+\sqrt{\alpha})^2+\sqrt{\alpha}$. Since $\alpha,\beta\leq \poly(\delta)$, this shows \Cref{eq:lem-nearby-max-ent-strat-rho-est}. Note that we automatically satisfy  \Cref{eq:lem-nearby-max-ent-strat-succ-prob-est} since $\lambda_0\in \Lambda$, so the proof is complete.
\end{proof}

\begin{theorem} \label{thm:auto-spec-gap}
    Let $G=(\mcX,\nu,\mcA,D)$ be a $\beta$-synchronous non-local game and let $\hat{\nu}$ be a probability distribution on $\mcX\times\mcX$. Suppose that $G$  $(\kappa,\hat{\nu})$-PME-robustly self-tests a perfect strategy $\mcS$. There exists universal constants $C_1,C_2,\zeta>0$ such that the spectral gap of the game polynomial $T_{G,\mcS}$ is at least $C_1\beta((\mathrm{id}+\kappa^2)^{-1}(C_2))^{\zeta}$.
\end{theorem}

\begin{proof}
    Let $\mcS=(\kpz,A)$ on $H_A\otimes H_A$ be the perfect strategy for $G$, $n$ be the dimension of $H_A$, $T=T_{G,\mcS}$ and $\Delta$ its spectral gap. Let $\kpo$ be a state in $H_A\otimes H_A$ orthogonal to $\kpz$ such that $\bra{\psi_1}T\kpo=1-\Delta$. Our goal is to use $\kpz$ and $\kpo$ to construct a state $\kp$ with high winning probability, but whose reduced density matrix on $H_A$ is not close to the normalised identity. Lemma \ref{lem:nearby-max-ent-strat} allows us to relate this to $\kappa$, from which we can derive the theorem. Note that we have the freedom to multiply $\kpo$ by a phase, which we will use later.

    Let $\ket{u_j}_{j=1}^n$ be an orthonormal basis of $H_A$ such that
    \begin{equation*}
        \kpz=\frac{1}{\sqrt{n}}\sum_j\ket{u_j}\ket{u_j}
    \end{equation*}
    and let $(S_{jk})_{jk}$ be a matrix such that
    \begin{equation*}
        \kpo=\sum_{jk}S_{jk}\ket{u_j}\ket{u_k}.
    \end{equation*}
    Define 
    \begin{equation*}
        \kp=\frac{1}{\sqrt{2}}\kpz+\frac{1}{\sqrt{2}}\kpo,
    \end{equation*}
    so 
    \begin{equation*}
        \kp=\frac{1}{\sqrt{2}}\sum_{jk}(S_{jk}+\frac{1}{\sqrt{n}}\delta_{jk})\ket{u_j}\ket{u_k}.
    \end{equation*}
    First note that the winning probability of $(\kp,A)$ is given by
    \begin{equation*}
        \omega((\kp,A))=\bra{\psi}T\ket{\psi}=\frac{1}{2}(\bra{\psi_0}T\kpz+\bra{\psi_1}T\kpo)=1-\frac{1}{2}\Delta,
    \end{equation*}
    and that this probability does not depend on the phase of $\kpo$. Next, we calculate its reduced density matrix $\rho$. For a state $\ket{\phi}$, given by
    \begin{equation*}
        \ket{\phi}=\sum_{jk}F_{jk}\ket{u_j}\ket{u_k},
    \end{equation*}
    the reduced density matrix on $H_A$ is given by
    \begin{equation*}
        \Tr_B(\ket{\phi}\bra{\phi})=\sum_{jkj'k'}F_{jk}\overline{F_{j'k'}}\delta_{kk'}\ket{u_j}\bra{u_{j'}}=\sum_{jj'}(FF^*)_{jj'}\ket{u_j}\bra{u_{j'}}.
    \end{equation*}
    Consequently, we find that
    \begin{equation*}
        \rho=\frac{1}{2}(S+\frac{\Id}{\sqrt{n}})(S+\frac{\Id}{\sqrt{n}})^*
    \end{equation*}
    with respect to the basis $\ket{u_j}$. Therefore, 
    \begin{equation*}
        \norm{\frac{\Id}{n}-\rho}_1=\norm{\frac{\Id}{n}-\frac{1}{2}(S+\frac{\Id}{\sqrt{n}})(S+\frac{\Id}{\sqrt{n}})^*}_1=\norm{\frac{\Id}{2n}-\frac{1}{2}SS^*-\frac{1}{2\sqrt{n}}(S+S^*)}_1.
    \end{equation*}
    If we consider the state $\ket{\psi'}$, given by
    \begin{equation*}
        \kp=\frac{1}{\sqrt{2}}\kpz-\frac{1}{\sqrt{2}}\kpo
    \end{equation*}
    and its reduced density matrix $\rho'$, then we have that
    \begin{equation*}
        \norm{\frac{\Id}{n}-\rho'}_1=\norm{\frac{\Id}{n}-\frac{1}{2}(-S+\frac{\Id}{\sqrt{n}})(-S+\frac{\Id}{\sqrt{n}})^*}_1=\norm{\frac{\Id}{2n}-\frac{1}{2}SS^*+\frac{1}{2\sqrt{n}}(S+S^*)}_1.
    \end{equation*}
    Using the triangle inequality, it follows that
    \begin{equation*}
        \frac{1}{\sqrt{n}}\norm{S+S^*}_1\leq \norm{\frac{\Id}{n}-\rho}_1+\norm{\frac{\Id}{n}-\rho'}_1.
    \end{equation*}
    We will now use our freedom in the phase of $\kpo$ to assume some properties of $S$ without loss of generality. It holds in general that $\norm{A+iB}_2^2=\norm{A}_2^2+\norm{B}_2^2$ for self-adjoint matrices $A$ and $B$, so we can assume that $\norm{S+S^*}_2^2\geq 2$, since $\norm{S}_2^2=\Tr(SS^*)=1$ and $S+S^*$ is twice the self-adjoint part of $S$. Next, we can possibly multiply $S$ by $-1$ to assume that
    \begin{equation*}
        \frac{1}{2\sqrt{n}}\norm{S+S^*}_1\leq \norm{\frac{\Id}{n}-\rho}_1.
    \end{equation*}
    We now claim that at least one of
    \begin{equation*}
        \frac{1}{16+12\sqrt{2}}\leq \frac{1}{2\sqrt{n}}\norm{S+S^*}_1 \text{ and } \frac{1}{8+6\sqrt{2}}\leq \frac{1}{2}\norm{\frac{\Id}{n}-SS^*}_1
    \end{equation*}
    holds. Let $(\lambda_j)_{j=1}^n$ be the eigenvalues of $S+S^*$ and define the sets
    \begin{equation*}
        \Lambda_>=\{1\leq j\leq n||\lambda_j|\geq \frac{8+6\sqrt{2}}{\sqrt{n}}\} \text{ and } \Lambda_<=\{1\leq j\leq n||\lambda_j|< \frac{8+6\sqrt{2}}{\sqrt{n}}\}.
    \end{equation*}
    Now we have that 
    \begin{equation*}
        \sum_{j\in \Lambda_>}\lambda_j^2\geq 1\text{ or } \sum_{j\in \Lambda_<}\lambda_j^2\geq 1.
    \end{equation*}
    In the first case, we use that the sequence of eigenvalues of $4SS^*$ majorises the squares of the eigenvalues of both the positive and negative parts of $S+S^*$ \cite[Theorem III.5.1]{Bha97}, so 
    \begin{equation*}
        \norm{\frac{\Id}{n}-SS^*}_1\geq \sum_{j\in\Lambda_>}\Bigg(\frac{\lambda_j^2}{8}-\frac{1}{n}\Bigg)\geq \sum_{j\in\Lambda_>}\Bigg(\frac{\lambda_j^2}{8}-\frac{\lambda_j^2}{136+96\sqrt{2}}\Bigg)\geq\frac{1}{8}-\frac{1}{136+96\sqrt{2}}= \frac{2}{8+6\sqrt{2}}.
    \end{equation*}
    In the other case we find that
    \begin{equation*}
        1\leq \sum_{j\in \Lambda_<}\lambda_j^2\leq \sum_{j\in \Lambda_<}\frac{8+6\sqrt{2}}{\sqrt{n}}|\lambda_j|,
    \end{equation*}
    so $(8+6\sqrt{2})\norm{S+S^*}_1\geq\sqrt{n}$. This shows the claim. Since 
    \begin{equation*}
        \norm{\frac{\Id}{n}-\rho}_1\geq \frac{1}{2}\norm{\frac{\Id}{n}-SS^*}_1-\frac{1}{2\sqrt{n}}\norm{S+S^*}_1,
    \end{equation*}
    we can now conclude that
    \begin{equation}\label{eq:thm-auto-spec-gap-univ-low-bound}
        \norm{\frac{\Id}{n}-\rho}_1\geq \frac{1}{16+12\sqrt{2}}>\frac{1}{33}
    \end{equation}
    in either case of the claim. We have now constructed a state $\kp$ with high winning probability such that the reduced density matrix is not close to the normalised identity.

    We are now in a position to use Lemma \ref{lem:nearby-max-ent-strat}. Note that the winning probability controls the asynchronicity by Lemma \ref{lem:value-dsync} since $G$ is $\beta$-synchronous, so $\dsync((\kp,A))\leq \Delta/(2\beta)$. Using Lemma \ref{lem:nearby-max-ent-strat} we obtain a maximally entangled strategy $\hat{\mcS}=(\ket{\hat{\psi}},\hat{A})$ on $\hat{H}\otimes \hat{H}$ with reduced density matrix $\hat{\rho}$ such that $\omega(\hat{\mcS})\geq 1-\poly(\Delta/\beta)$ and
    \begin{equation*}
        \norm{\rho-\hat{\rho}}_1= \mcO(\kappa(\poly(\frac{\Delta}{\beta}))^2+\poly(\frac{\Delta}{\beta})).
    \end{equation*}
    By the $(\kappa,\hat{\nu})$-PME-robust self-testing of $G$ and Lemma \ref{lem:contained-ME-strat-dim-est}, we know that $\dim(\hat{H})\geq (1-\kappa(\poly(\Delta/\beta))^2)n$. From this it follows that
    \begin{equation*}
        \norm{\hat{\rho}-\frac{\Id}{n}}_1=\dim(\hat{H})\left(\frac{1}{\dim(\hat{H})}-\frac{1}{n}\right)+(n-\dim(\hat{H}))\frac{1}{n}\leq 2\kappa\left(\poly\left(\frac{\Delta}{\beta}\right)\right)^2,
    \end{equation*}
    so 
    \begin{equation*}
        \norm{\rho-\frac{\Id}{n}}_1= \mcO(\kappa(\poly(\frac{\Delta}{\beta}))^2+\poly(\frac{\Delta}{\beta})).
    \end{equation*}
    By \Cref{eq:thm-auto-spec-gap-univ-low-bound}, we have now shown that
    \begin{equation*}
        \frac{1}{33}= \mcO(\kappa(\poly(\frac{\Delta}{\beta}))^2+\poly(\frac{\Delta}{\beta})).
    \end{equation*}
    Inverting this inequality shows that there exist constants $C_1,C_2,\zeta>0$ such that
    \begin{equation*}
        \Delta\geq C_1\beta ((\mathrm{id}+\kappa^2)^{-1}(C_2))^\zeta,
    \end{equation*}
    proving the theorem.
\end{proof}

\begin{corollary} \label{cor:PME-robust-self-testing-implies-robust-self-testing}
    There exist universal constants $C_1,C_2,C_3,\zeta_1,\zeta_2>0$ such that every $\beta$-synchronous $\kappa$-PME-robust self-test is a $\kappa'$-robust self-test with
    \begin{equation*}
        \kappa'(\epsilon)\leq C_1\frac{\sqrt{(\mathrm{id}+\kappa)}(C_2(\frac{\epsilon}{\beta})^{\zeta_1})}{\beta ((\mathrm{id}+\kappa^2)^{-1}(C_3))^{\zeta_2}}.
    \end{equation*}
\end{corollary}

\section{Spectral gap of the Quantum Low Degree Test} \label{sec:QLDT}
In Corollary \ref{cor:PME-robust-self-testing-implies-robust-self-testing}, we have found that every $\kappa$-PME-robust self-test is also a robust self-test with a related robustness. However, it is not known that the Quantum Low Degree Test is a $(\kappa,\nu_{\mathrm{qldt}})$-PME-robust self-test. Using the notation in Definition \ref{def:qubit-test}, we only know that the Quantum Low Degree Test is a $(\kappa,\nu')$-PME-robust self-test. We are therefore forced to explicitly compute a lower bound for the spectral gap for this distribution. In fact, we are able to calculate the exact spectral gap.

\begin{theorem}\label{thm:QLDT-spec-gap}
    Let $G_{\mathrm{qldt}}=(\mcX_{\mathrm{qldt}},\nu_{\mathrm{qldt}},\mcA_{\mathrm{qldt}},D_{\mathrm{qldt}})$ be the Quantum Low Degree Test for $k$ qubits with optimal strategy $\mcS$ and let $G'=(\mcX_{\mathrm{qldt}},\nu_{\mathrm{qldt}}',\mcA_{\mathrm{qldt}},D_{\mathrm{qldt}})$. Let $d$ be the relative distance of the code used in the Quantum Low Degree Test. Then the spectral gap of the game polynomial $T_{G',\mcS}$ is $\frac{d}{2}$.   
\end{theorem}

\begin{proof}
    For an element $a\in\mbF_2$ let $\overline{a}=1-a$. Let $B=\{\ket{\psi_{ab}}|a,b\in\mbF_2\}$ be the Bell basis of $\C^2\otimes \C^2$ given by 
	\begin{equation*}
		\ket{\psi_{ab}}=\frac{1}{2}\left(\ket{0a}+(-1)^b\ket{1\overline{a}}\right)
	\end{equation*}
	for all $a,b\in \mbF_2$. For $\bfa,\bfb\in \mbF_2^k$ define 
	\begin{equation*}
		\ket{\psi_{\bfa\bfb}}=\bigotimes_{i=1}^k\ket{\psi_{a_ib_i}}.
	\end{equation*}
	Note that 
	\begin{equation*}
		\bra{\psi_{ab}}X\otimes X\ket{\psi_{ab}}=(-1)^b\text{ and } \bra{\psi_{ab}}Z\otimes Z\ket{\psi_{ab}}=(-1)^a
	\end{equation*}
	for all $a,b\in \mbF_2$, so
	\begin{equation}\label{eq:sigma-W-meas-outcome}
		\bra{\psi_{\bfa\bfb}}\sigma^X(\bfc)\otimes \sigma^X(\bfc)\ket{\psi_{\bfa\bfb}}=(-1)^{\bfb\cdot \mathbf{c}}\text{ and } \bra{\psi_{\bfa\bfb}}\sigma^Z(c)\otimes \sigma^Z(c)\ket{\psi_{\bfa\bfb}}=(-1)^{\bfa\cdot \mathbf{c}}
	\end{equation}
	for all $\bfa,\bfb,\mathbf{c}\in\mbF_2^k$. 

    The game polynomial $T_{G',\mcS}$ is given by
	\begin{equation*}
		T_{G',\mcS}=\expect{(x,y)\sim\nu_{\qldt}'}\sum_{a,b}D_{\qldt}(a,b|x,y)A_a^x\otimes (A_b^y)^T,
	\end{equation*}
	where $\{A^x\}_{x\in\mcX_{\qldt}}$ are the measurement operators of the ideal strategy. Let $E$ be the generating matrix for the code $C_{\mathrm{RM2}}$ used in the construction of the Quantum Low Degree Test, so $S_X=S_Z=\{(E_{ij})_{i=1}^k|1\leq j\leq n\}$, and let $n$ be the length of the code. For each $W\in \{X,Z\}$ and $\bfa\in S_W$ there are two questions $x_{W,\bfa}^1$ and $x_{W,\bfa}^2$ in $\mcX_{\qldt}$ such that
	\begin{equation*}
		U(A^y)=\sigma^W(\bfa)\text{ for } y\in \{x_{W,\bfa}^1,x_{W,\bfa}^2\}.
	\end{equation*}
	Moreover, 
	\begin{equation*}
		\nu'_{\qldt}(x_{W,\bfa}^1,x_{W,\bfa}^2)=\nu'_{\qldt}(x_{W,\bfa}^2,x_{W,\bfa}^1)= \frac{1}{2}\cdot \frac{1}{2}\cdot\frac{1}{n},
	\end{equation*}
	where this propability is obtained by multiplying the probability that this $W$ is chosen, the probability that question 1 goes to Alice and the probability that $\bfa$ is chosen from $S_W$. For these questions, Alice and Bob win if they give the same answer. Since
	\begin{equation*}
		A^{x_{W,\bfa}^1}_0\otimes (A^{x_{W,\bfa}^2}_0)^T+A^{x_{W,\bfa}^1}_1\otimes (A^{x_{W,\bfa}^2}_1)^T=\frac{1}{2}\left(1+\sigma^W(\bfa)\otimes \sigma^W(\bfa)\right),
	\end{equation*}
	this implies that 
	\begin{equation*}
		T_{G',\mcS}= \frac{1}{2}+\frac{1}{4n}\sum_{W\in\{X,Z\}}\sum_{\bfa\in S_X}\sigma^W(\bfa)\otimes \sigma^W(\bfa),
	\end{equation*}
	where the factor arises from the fact that we are considering both $W=X$ and $W=Z$ and the observation that both orders of questions give the same expression.
	
	 By \Cref{eq:sigma-W-meas-outcome}, the above inequality tells us that
	\begin{equation*}
		\bra{\psi_{\bfa\bfb}}T_{G',\mcS}\ket{\psi_{\bfa\bfb}}= \frac{1}{2}+\frac{1}{4n}\sum_{\bfc\in S_X}\left((-1)^{\bfa\cdot \bfc}+(-1)^{\bfb\cdot \bfc}\right).
	\end{equation*}
	Therefore, the question we need to answer is for given $\bfa\in \mbF_2^k$, how many $\bfc\in S_X$ are there such that $\bfa\cdot \bfc=1$? For this, we need to recall how to calculate the physical representation $\bfa_{\mathrm{phys}}\in \mbF_2^n$ of a logical code word $\bfa_{\mathrm{logic}}\in\mbF_2^k$. This connection is defined by the generating matrix and is given by
	\begin{equation*}
		\bfa_{\mathrm{phys}}=\sum_{i=1}^ka_{\mathrm{logic},i}\cdot (E_{ij})_{j=1}^n.
	\end{equation*}
	Another way to represent this is by describing each $a_{\mathrm{phys},i}$, which gives
	\begin{equation*}
		a_{\mathrm{phys},i}=\bfa_{\mathrm{logic}}\cdot (E_{ij})_{i=1}^k.
	\end{equation*}
	This means that the number of $\bfc\in S_X$, i.e. the number of columns of $E$, for which $\bfa\cdot \bfc=1$ precisely equals the number of ones in the physical representation of the logical code word $\bfa$. By the nature of the code, this is at least $dn$ if $\bfa\neq 0$. Consequently, 
	\begin{equation*}
		\bra{\psi_{\bfa\bfb}}T_{G',\mcS}\ket{\psi_{\bfa\bfb}}\leq 1-\frac{d}{2}
	\end{equation*}
	unless $\bfa=0=\bfb$. Since the Quantum Low Degree Test is perfect and the distance in a code is attained, this shows that the spectral gap is $\frac{d}{2}$. 
\end{proof}

\begin{remark}
    The above proof is closely related to \cite[Example 1.2]{dlS22b}. Since $\bfa\mapsto \sigma^W(\bfa)\otimes \sigma^W(\bfa)$ is a unitary representation of $\mbF_2^k$, \cite[Example 1.2]{dlS22b} tells us that the operator 
    \begin{equation*}
        \frac{1}{n}\sum_{\bfa\in S_X}\sigma^W(\bfa)\otimes \sigma^W(\bfa)
    \end{equation*}
    has spectral gap $2d$, if you do not account for multiplicities. Doing this for both $W=X$ and $W=Z$ and arguing that $T_{G',\mcS}$ only has a one-dimensional eigenspace for eigenvalue 1, using much of the proof above, also yields \ref{thm:QLDT-spec-gap}.
\end{remark}

\begin{corollary}\label{cor:QLDT}
    The $1/2$-synchronised version of the Quantum Low Degree Test is a $\big(k,\poly(\log(k))\cdot\poly(\epsilon)\big)$-qubit test.
\end{corollary}
\begin{proof}
    This is the consequence of combining Corollary \ref{cor:lifting-PME-assumption-if-spec-gap-beta-sync}, Theorem \ref{thm:QLDT-spec-gap} and Lemma \ref{lem:beta-synchronised} after observing that $\nu'_{\qldt}\leq 4\nu_{\qldt}$.
\end{proof}
\bigskip

\appendix
\section{Proof of Lemma \ref{lem:PME-robust-self-testing-vNA}}
In this appendix, we prove \Cref{lem:PME-robust-self-testing-vNA}. The proof relies on several technical lemmas that we outline below.

\begin{lemma}\label{lemma:Schmidt-bound}
Let $\ket{\psi}$ and $\ket{\wtd{\psi}}$ be two unit vectors in $\C^d\otimes \C^d$ with Schmidt decompositions $\ket{\psi}=\sum_{i=1}^d\lambda_{i}\ket{u_i}\otimes\ket{v_i}$ and $\ket{\wtd{\psi}}=\sum_{i=1}^d\mu_{i}\ket{\wtd{u}_i}\otimes\ket{\wtd{v}_i}$, where $\lambda_1\geq \lambda_2\geq\cdots\geq \lambda_d\geq 0$ and $\mu_1\geq \mu_2\geq\cdots\geq \mu_d\geq 0$. Then the $\ell_2$-distance of the sequences $\bflambda=(\lambda_1,\ldots,\lambda_d)$ and $\bfmu=(\mu_1,\ldots.\mu_d)$ is bounded above by the norm-distance of $\ket{\psi}$ and $\ket{\wtd{\psi}}$. That is,
\begin{equation}
    \norm{(\lambda_i)_i-(\mu_i)_i}_2\leq \norm{\ket{\psi}-\ket{\wtd{\psi}}}. \label{eq:Schmidt-bound}
\end{equation}
\end{lemma}
\begin{proof}
    Let $a_{k\ell}:=\latRe \big(\braket{u_k|\wtd{u}_\ell}\braket{v_k|\wtd{v}_\ell}\big)$ for all $1\leq k,\ell\leq d$. We first show that there exists a bistochastic matrix $X=(x_{k\ell})_{k,\ell}$ such that $a_{k\ell}\leq x_{k\ell}$ for all $k,\ell$. For any $\ell$,
    \begin{align*}
        c_\ell:=&\sum_{k=1}^d a_{kl}\leq \sum_{k=1}^d \abs{\braket{u_k|\wtd{u}_\ell}}\cdot \abs{\braket{v_k|\wtd{v}_\ell}}\\
        \leq & \left(\sum_{k=1}^d\abs{\braket{u_k|\wtd{u}_\ell}}^2 \right)^{1/2}\cdot \left(\sum_{k=1}^d\abs{\braket{v_k|\wtd{v}_\ell}}^2 \right)^{1/2}\\
        =& \norm{\ket{\wtd{u}_\ell}}\cdot \norm{\ket{\wtd{v}_\ell}}=1.
    \end{align*}
The second line follows from the Cauchy-Schwartz inequality. The third line holds because $\{\ket{u_k}:1\leq k\leq d\}$ and  $\{\ket{v_k}:1\leq k\leq d\}$ are orthonormal bases for $\C^d$. Similarly,
\begin{equation*}
    r_k:=\sum_{\ell=1}^d a_{k\ell}\leq 1
\end{equation*}
for all $k$. So 
\begin{equation*}
    T:=\sum_{k,\ell}a_{k\ell}=\sum_\ell c_\ell=\sum_k r_k\leq d.
\end{equation*}
If $T=d$, then $c_\ell=r_k=1$ for all $k,\ell$ and hence $(a_{kl})_{k,l}$ is a bistochastic matrix. Now assume $T<d$. Then 
\begin{equation*}
    b_{k\ell}:=\frac{(1-c_\ell)(1-r_k)}{d-T}\geq 0
\end{equation*}
for all $k,\ell$, and
\begin{align*}
    \sum_{k=1}^d (a_{k\ell}+b_{k\ell}) = c_\ell +\frac{\sum_{k=1}^d(1-r_k)}{d-T}(1-c_\ell)=c_\ell +\frac{d-T}{d-T}(1-c_\ell)=1
\end{align*}
for all $\ell$. Similarly, $ \sum_{\ell=1}^d (a_{k\ell}+b_{k\ell})=1$ for all $k$. Hence $(a_{k\ell}+b_{k\ell})_{k,\ell}$ is a bistochastic matrix. We conclude that $(a_{k\ell})_{k,\ell}$ is entrywise smaller than some  bistochastic matrix $X$.

By Birkhoff–von Neumann theorem, $X=\theta_1P_1+\cdots+\theta_tP_t$ is a convex combination of some permutation matrices $P_1,\ldots,P_t$. It follows that
\begin{align*}
    \latRe\big(\braket{\psi|\wtd{\psi}}\big)&=\sum_{k,\ell} \lambda_k \mu_\ell a_{k\ell}=\begin{pmatrix}
        \lambda_1 & \cdots & \lambda_d
    \end{pmatrix} (a_{k\ell})_{k,\ell}\begin{pmatrix}
        \mu_1 \\ \vdots \\ \mu_d
    \end{pmatrix}\\
    &\leq \begin{pmatrix}
        \lambda_1 & \cdots & \lambda_d
    \end{pmatrix} X\begin{pmatrix}
        \mu_1 \\ \vdots \\ \mu_d
    \end{pmatrix} = \sum_{i=1}^t \theta_i \begin{pmatrix}
        \lambda_1 & \cdots & \lambda_d
    \end{pmatrix} P_i\begin{pmatrix}
        \mu_1 \\ \vdots \\ \mu_d
    \end{pmatrix} \\
    &\leq \sum_{i=1}^t \theta_i \begin{pmatrix}
        \lambda_1 & \cdots & \lambda_d
    \end{pmatrix}\begin{pmatrix}
        \mu_1 \\ \vdots \\ \mu_d
    \end{pmatrix} =\sum_{k=1}^d\lambda_k \mu_k.
\end{align*}
The last line follows from the rearrangement inequality. This implies
\begin{align*}
    \norm{\ket{\psi}-\ket{\wtd{\psi}}}^2=2-2\latRe\big(\braket{\psi|\wtd{\psi}}\big)\geq 2-2\sum_{k=1}^d\lambda_k \mu_k=\norm{\bflambda-\bfmu}_2^2.
\end{align*}
So \Cref{eq:Schmidt-bound} follows.
\end{proof}

\begin{lemma}\label{lem:me-aux}
    Let $\mcS=(\kp, A, B)$ and $\tilde{\mcS}=(\kpt, \tA, \tB)$ be two maximally entangled strategies on $H_A\otimes H_B$ and $\tilde{H}_A\otimes \tilde{H}_B$ such that $\tilde{\mcS}$ is a local $(\epsilon,\nu)$-dilation of $\mcS$ for some $\epsilon\geq 0$ and distribution $\nu$ on $\mcX\times\mcY$ with isometries $V_A$ and $V_B$ and unit vector $\ket{aux}\in K_A\otimes K_B$. Let $\nu_A$ and $\nu_B$ be the marginal distributions of $\nu$ on $\mcX$ and $\mcY$, respectively. Then there exist subspaces $K_A'\subset K_A$ and $K_B'\subset K_B$ and a maximally entangled state $\ket{aux'}\in K_A'\otimes K_B'$ such that
    \begin{align*}
    \norm{(V_A\otimes V_B)\ket
    \psi-\ket{\wtd{\psi}}\otimes\ket{aux'}}&\leq 3\epsilon,\\
        \left(\expect{x\sim \nu_A}\sum_{a}  \norm{(V_A\otimes V_B)(A^x_a\otimes \Id_{H_B})\ket{\psi} - \big((\wtd{A}^x_a\otimes \Id_{\tilde{H}_B})\ket
        {\wtd{\psi}}\big)\otimes\ket{aux'}}^2    \right)^{1/2}&\leq 3\epsilon,\\
        \left(\expect{y\sim \nu_B}\sum_{b}  \norm{(V_A\otimes V_B)(\Id_{H_A}\otimes B^y_b)\ket{\psi} - \big((\Id_{\tilde{H}_A}\otimes\wtd{B}^y_b)\ket
        {\wtd{\psi}}\big)\otimes\ket{aux'}}^2    \right)^{1/2}&\leq 3\epsilon.
    \end{align*}
\end{lemma}
\begin{proof}
    Let 
    \begin{equation*}
        (V_A\otimes V_B)\kp = \sum_{i}\lambda_i \ket{u_i}\ket{v_i} \text{ and } \kpt\otimes \ket{aux} = \sum_{i}\kappa_i \ket{\tilde{u}_i}\ket{\tilde{v}_i}
    \end{equation*}
    be the Schmidt decompositions of $(V_A\otimes V_B)\kp$ and $\kpt\otimes \ket{aux}$, where $\bflambda=(\lambda_i)_i$ and $\bfkappa=(\kappa_i)_i$ are decreasing sequences. Because $\tilde{\mcS}$ is a local $(\epsilon,\nu)$-dilation of $\mcS$, we know that
    \begin{equation*}
        \norm{\bflambda-\bfkappa}_2\leq \norm{(V_A\otimes V_B)\kp-\kpt\otimes \ket{aux}}\leq \epsilon
    \end{equation*}
    by Lemma \ref{lemma:Schmidt-bound}.

    Since $\kp$ is maximally entangled, we know that there exists an $m\in \N$ such that $\lambda_i=\frac{1}{\sqrt{m}}$ for $1\leq i\leq m$ and $\lambda_i=0$ for $i>m$. On the other side we know that the multiplicity of each value in $\bfkappa$ is divisible by $n=\dim(\tilde{H}_A)$. We want to show that there exists a sequence $\bfkappa'$ of Schmidt coefficients which takes exactly one non-zero value, has multiplicities divisible by $n$ and is close to $\bfkappa$. First suppose that $m\geq n$. In this case, $\bflambda$ would be a good candidate, but the multiplicity of $\lambda_i$ does not have to be divisible by $n$. To remedy this, we define
    \begin{equation*}
        \kappa_i''=\begin{cases}
            \frac{1}{\sqrt{m}} & i\leq \lfloor\frac{m}{n}\rfloor n \text{ or }\left(\lfloor\frac{m}{n}\rfloor n < i \leq \lfloor(\frac{m}{n}+1)\rfloor n \text{ and } \kappa_i\geq \frac{1}{2\sqrt{m}}\right)\\
            0 & i > \lfloor(\frac{m}{n}+1)\rfloor n \text{ or }\left(\lfloor\frac{m}{n}\rfloor n < i \leq \lfloor(\frac{m}{n}+1)\rfloor n\text{ and } \kappa_i< \frac{1}{2\sqrt{m}}\right)
        \end{cases}
    \end{equation*}
    and observe that $\norm{\bfkappa''-\bfkappa}_2\leq \norm{\bflambda-\bfkappa}_2$. Unfortunately, $\norm{\bfkappa''}_2\neq 1$ in general. However, by the triangle inequality, $1-\epsilon\leq \norm{\bfkappa''}_2\leq 1+\epsilon$, so defining $\bfkappa'$ as the normalisation of $\bfkappa''$ gives us the desired sequence with $\norm{\bfkappa'-\bfkappa}_2\leq 2\epsilon$. For the case where $m<n$, we immediately define
    \begin{equation*}
        \kappa'_i=\begin{cases}
            \frac{1}{\sqrt{n}} & i \leq n\\
            0 & i > n
        \end{cases}
    \end{equation*}
    which is normalised and satisfies $\norm{\bfkappa'-\bfkappa}_2\leq\norm{\bflambda-\bfkappa}_2\leq \epsilon$. 

    Let
    \begin{equation*}
        \ket{aux}=\sum_i \mu_i \ket{w_i}\ket{w_i'}
    \end{equation*}
    be the Schmidt decomposition of $\ket{aux}$. Having obtained the sequence $\bfkappa'$ with the desired properties, we can now define
    \begin{equation*}
        \ket{aux'}=\sum_i \sqrt{n}\kappa_{ni}'\ket{w_i}\ket{w_i'}.
    \end{equation*}
    The state $\kpt\otimes\ket{aux'}$ is then given by
    \begin{equation*}
        \kpt\otimes \ket{aux'} = \sum_{i}\kappa_i' \ket{\tilde{u}_i}\ket{\tilde{v}_i},
    \end{equation*}
    and therefore we have $\norm{\kpt\otimes \ket{aux'} - \kpt\otimes \ket{aux}}_2\leq 2\epsilon$. By the triangle inequality, we find that
    \begin{equation*}
        \norm{(V_A\otimes V_B)\ket
    \psi-\ket{\wtd{\psi}}\otimes\ket{aux'}}\leq 3\epsilon.
    \end{equation*}

    Next, observe that
    \begin{equation*}
       \sum_{a}  \norm{((\tA_a^x\otimes \Id_{\tilde{H}_B})\kpt)\otimes(\ket{aux'}-\ket{aux})}^2  =\sum_{a}  \bra{\tilde{\psi}}(\tA_a^x\otimes \Id_{\tilde{H}_B})^2\kpt\norm{\ket{aux'}-\ket{aux}}^2.
    \end{equation*}
    Since $\tA$ is a POVM, we know that
    \begin{equation*}
        \sum_{a}  \bra{\tilde{\psi}}(\tA_a^x\otimes \Id_{\tilde{H}_B})^2\kpt\leq \sum_{a}  \bra{\tilde{\psi}}(\tA_a^x\otimes \Id_{\tilde{H}_B})\kpt=1,
    \end{equation*}
    so we find that
    \begin{equation*}
        \sum_{a}  \norm{((\tA_a^x\otimes \Id_{\tilde{H}_B})\kpt)\otimes(\ket{aux'}-\ket{aux})}^2\leq 4\epsilon^2.
    \end{equation*}
    We have
    \begin{equation*}
        \left(\expect{x\sim \nu_A}\sum_{a}  \norm{(V_A\otimes V_B)(A^x_a\otimes \Id_{H_B})\ket{\psi} - \big((\wtd{A}^x_a\otimes \Id_{\tilde{H}_B})\ket
        {\wtd{\psi}}\big)\otimes\ket{aux}}^2    \right)^{1/2}\leq \epsilon
    \end{equation*}
    since $\tilde{\mcS}$ is a local $(\epsilon,\nu)$-dilation of $\mcS$. Combining these equations using the triangle inequality for the vector space 
    \begin{equation*}
        \bigoplus_{x,a}\tilde{H}_A\otimes\tilde{H}_B\otimes K_A\otimes K_B
    \end{equation*}
    with norm given by
    \begin{equation*}
        \norm{v}=\left(\expect{x\sim \nu_A}\sum_{a}\norm{v(x,a)}^2\right)^{\frac{1}{2}},
    \end{equation*}
    we obtain
    \begin{equation*}
        \left(\expect{x\sim \nu_A}\sum_{a}  \norm{(V_A\otimes V_B)(A^x_a\otimes \Id_{H_B})\ket{\psi} - \big((\wtd{A}^x_a\otimes \Id_{\tilde{H}_B})\ket
        {\wtd{\psi}}\big)\otimes\ket{aux'}}^2    \right)^{1/2}\leq 3\epsilon.
    \end{equation*}
    The other inequality can be proved analogously.
\end{proof}

\begin{theorem}[Polar decomposition, {\cite[Theorem VIII.3.11]{Con10}}]
    Let $A\in B(H)$ for some Hilbert space $H$. There is a partial isometry $W$ with $\ker(W)=\ker(A)$ and $\mathrm{ran}(W)=\overline{\mathrm{ran}(A)}$ such that $A=W|A|$, where $|A|:=(A^*A)^{1/2}$ is defined by the functional calculus. Here $(W,|A|)$ is called the \emph{polar decomposition} of $A$.
\end{theorem}
\begin{remark}
    If $A:H\rightarrow K$ is a bounded linear map between Hilbert spaces, one can also take the polar decomposition by viewing $A$ as element of $B(H\oplus K)$. 
\end{remark}

\begin{lemma}\label{lem:isom-est}
    Let $V:H\rightarrow H'$ be an isometry, let $P\in B(H')$ be a projection and let $X=W|X|$ be the polar decomposition for $X=PV$. Then
    \begin{equation*}
        (V-W)^*(V-W)\leq 2(V-PV)^*(V-PV).
    \end{equation*}
\end{lemma}
\begin{proof}
    We start by expanding the left hand side, which gives us
    \begin{align*}
        (V-W)^*(V-W)&=V^*V-V^*W-W^*V+W^*W\\
        &=V^*V-V^*PW-W^*PV+W^*W\\
        &=V^*V-|X|W^*W-W^*W|X|+W^*W\\
        &=V^*V-|X|-|X|+W^*W
    \end{align*}
    Because $\norm{|X|}_{\infty}\leq 1$ and $W$ is a partial isometry, we have the inequalities
    \begin{equation*}
        -|X^*|\leq -|X^*|^2\text{ and } W^*W\leq \Id_H=V^*V.
    \end{equation*}
    Consequently, we have
    \begin{align*}
        (V-W)^*(V-W)&\leq V^*V-|X|^2-|X|^2+V^*V\\
        &=V^*V-|X|W^*W|X|-|X|W^*W|X|+V^*V\\
        &=2(V^*V-V^*PV)\\
        &=2(V-PV)^*(V-PV).\qedhere
    \end{align*}
\end{proof}

\begin{lemma}\label{lem:partial-iso}
     Let $\kp\in H_A\otimes H_B, \kpt\in \tilde{H}_A\otimes \tilde{H}_B$ be states such that $\kp$ is maximally entangled, let $\nu_A$ be a probability measure on $\mcX$ and let $A,\tA$ be POVMs on $H_A$ and $\tilde{H}_A$, respectively. Let $V_A:H_A\rightarrow \tilde{H}_A\otimes K_A$ and $V_B:H_B\rightarrow \tilde{H}_B\otimes K_B$ be isometries and $\ket{aux}\in K_A'\otimes K_B' \subset K_A\otimes K_B$ such that
     \begin{align*}
          \norm{(V_A\otimes V_B)\ket
    \psi-\ket{\wtd{\psi}}\otimes\ket{aux}}&\leq \epsilon,\\
        \left(\expect{x\sim \nu_A}\sum_{a}  \norm{(V_A\otimes V_B)(A^x_a\otimes \Id_{H_B})\ket{\psi} - \big((\wtd{A}^x_a\otimes \Id_{\tilde{H}_B})\ket
        {\wtd{\psi}}\big)\otimes\ket{aux}}^2    \right)^{1/2}&\leq \epsilon.
     \end{align*}
     Then there exist partial isometries $V_A':H_A\rightarrow \tilde{H}_A\otimes K_A'$ and $V_B':H_B\rightarrow \tilde{H}_B\otimes K_B'$ such that
     \begin{align}
          \norm{(V_A'\otimes V_B')\ket
    \psi-\ket{\wtd{\psi}}\otimes\ket{aux}}&\leq (1+2\sqrt{2})\epsilon,\label{eq:est-1-lem-app-part-iso-constr}\\
        \left(\expect{x\sim \nu_A}\sum_{a}  \norm{(V_A'A^x_a\otimes V_B')\ket{\psi} - \big((\wtd{A}^x_a\otimes \Id_{\tilde{H}_B})\ket
        {\wtd{\psi}}\big)\otimes\ket{aux}}^2    \right)^{\frac{1}{2}}&\leq (1+4\sqrt{2})\epsilon,\label{eq:est-2-lem-app-part-iso-constr}\\
        1-\norm{(V_A'\otimes \Id_{H_B})\ket{\psi}}^2 &\leq 4\epsilon^2,\label{eq:est-4-lem-app-part-iso-constr}\\
        1-\norm{\big((V_A')^*\otimes \Id_{\wtd{H}_B\otimes K_B}  \big)\ket{\wtd{\psi}}\otimes\ket{aux}}^2 &\leq \epsilon^2.\label{eq:est-5-lem-app-part-iso-constr}
     \end{align}
Moreover, if $\kpt\otimes\ket{aux}$ is maximally entangled on $(\tH_A\otimes K_A')\otimes(\tH_B\otimes K_B')$, then
\begin{equation}
    \norm{(\Id_{H_A}\otimes V_B')\ket{\psi}-\big((V_A')^*\otimes \Id_{\wtd{H}_B\otimes K_B}  \big)\ket{\wtd{\psi}}\otimes\ket{aux}} \leq 7\epsilon.\label{eq:est-3-lem-app-part-iso-constr}
\end{equation}
\end{lemma}
\begin{proof}
    Let $n=\dim(H_A)$. Let $P_A$ and $P_B$ be the projections onto $\tilde{H}_A\otimes K_A'$ and $\tilde{H}_B\otimes K_B'$, respectively. Because $P_A$ and $P_B$ are contractions, we have
    \begin{align}
        \norm{(V_A\otimes V_B)\kp-(P_AV_A\otimes V_B)\kp}&\leq 2\epsilon,\label{eq:proj-close-lem-app-part-iso-constr}\\
        \norm{(V_A\otimes V_B)\kp-(V_A\otimes P_BV_B)\kp}&\leq 2\epsilon\text{ and}\\
        \norm{(V_A\otimes V_B)\kp-(P_AV_A\otimes P_BV_B)\kp}&\leq 2\epsilon
    \end{align}
    by the triangle inequality, using the estimate $\norm{(V_A\otimes V_B)\ket
    \psi-\ket{\wtd{\psi}}\otimes\ket{aux}}\leq \epsilon$. Now let $X_A=P_AV_A$, $X_B=P_BV_B$ and let $X_A=V_A'|X_A|$ and $X_B=V_B'|X_B|$ be the polar decompositions of their adjoints. By Lemma \ref{lem:isom-est}, we have that 
    \begin{equation*}
        \norm{(V_A\otimes V_B)\kp-(V_A'\otimes V_B')\kp}^2\leq 2 \norm{(V_A\otimes V_B)\kp-(P_AV_A\otimes P_BV_B)\kp}^2.
    \end{equation*}
    Using the triangle inequality, this implies that
    \begin{equation*}
        \norm{(V_A'\otimes V_B')\kp-\ket{\wtd{\psi}}\otimes\ket{aux'}}\leq (1+2\sqrt{2})\epsilon,
    \end{equation*}
    proving \Cref{eq:est-1-lem-app-part-iso-constr}. Next, let $Q$ be the projection onto the orthogonal complement of the kernel of $X_A$. We see that
    \begin{equation*}
        \norm{(V_A'\otimes \Id_{H_B})\kp}^2=\bra{\psi}(Q\otimes \Id_{H_B})\kp\geq \bra{\psi}(X_A^*X_A\otimes \Id_{H_B})\kp=\norm{(P_AV_A\otimes V_B)\kp}^2.
    \end{equation*}
    Since $V_A$ and $V_B$ are isometries, the Pythagorean theorem combined with \Cref{eq:proj-close-lem-app-part-iso-constr} tells us that
    \begin{equation*}
        \norm{(V_A'\otimes \Id_{H_B})\kp}^2\geq \norm{(P_AV_A\otimes V_B)\kp}^2\geq 1-4\epsilon^2,
    \end{equation*}
    proving \Cref{eq:est-4-lem-app-part-iso-constr}. Similarly, we compute 
    \begin{align*}
        \norm{\big((V_A')^*\otimes \Id_{\wtd{H}_B\otimes K_B}  \big)\ket{\wtd{\psi}}\otimes\ket{aux}}&\geq \norm{\big(|X_A|(V_A')^*\otimes \Id_{\wtd{H}_B\otimes K_B}  \big)\ket{\wtd{\psi}}\otimes\ket{aux}}\\
        &=\norm{\big(V_A^*P_A\otimes \Id_{\wtd{H}_B\otimes K_B}  \big)\ket{\wtd{\psi}}\otimes\ket{aux}}\\
        &=\norm{\big(V_AV_A^*\otimes \Id_{\tH_B\otimes K_B}\big)\kpt\otimes \ket{aux}}.
    \end{align*}
    Note that $\Id_{\tH_A\otimes K_A}-V_AV_A^*$ is a contraction and $(\Id_{\tH_A\otimes K_A}-V_AV_A^*)V_A=0$, so we have
    \begin{equation} \label{eq:small-kpt-proj-lem-app-part-iso-constr}
        \norm{\left((\Id_{\tH_A\otimes K_A}-V_AV_A^*)\otimes \Id_{\tH_B\otimes K_B}\right)\ket{\wtd{\psi}}\otimes\ket{aux}}\leq \norm{(V_A\otimes V_B)\ket
    \psi-\ket{\wtd{\psi}}\otimes\ket{aux}}\leq \epsilon.
    \end{equation}
    These things combine to yield equation \Cref{eq:est-5-lem-app-part-iso-constr}, i.e. 
    \begin{equation*}
        1-\norm{\big((V_A')^*\otimes \Id_{\wtd{H}_B\otimes K_B}  \big)\ket{\wtd{\psi}}\otimes\ket{aux}}^2\leq \epsilon^2.
    \end{equation*}
    
    For the other estimates of this lemma, we need to use that $\kp$ is maximally entangled. First, note that for any partial isometry $w_B:H_B\rightarrow \tilde{H}_B\otimes K_B'$ we have
     \begin{equation*}
         \norm{\left((V_A-V_A')\otimes w_B\right)(A^x_a\otimes \Id_{H_{B}})\ket{\psi}}\leq \norm{((V_A-V_A')\otimes \Id_{H_B})(A^x_a\otimes \Id_{H_{B}})\ket{\psi}},
     \end{equation*}
     since $\Id_{\tilde{H}_A\otimes K_A}\otimes w_B$ is a contraction. Consequently,
    \begin{align*}
        \sum_{a}\norm{((V_A-V_A')&\otimes V_B')(A^x_a\otimes \Id_{H_B})\ket{\psi}}^2\\
        &\leq \sum_a\bra{\psi}A_a^x(V_A-V_A')^*(V_A-V_A')A_a^x\otimes \Id_{H_B}\kp\\
        &=\sum_a\Tr_{(\tilde{H}_A\otimes K_A)\otimes H_B}\big((V_A-V_A')A_a^x\otimes \Id_{H_B}\kp\bra{\psi}A_a^x(V_A-V_A')^*\otimes \Id_{H_B}\big)\\
        &=\frac{1}{n}\sum_a\Tr_{\tilde{H}_A\otimes K_A}\big((V_A-V_A')A_a^x\Id_{H_A}A_a^x(V_A-V_A')^*\big)\\
        &\leq\frac{1}{n}\Tr_{\tilde{H}_A\otimes K_A}\big((V_A-V_A')(V_A-V_A')^*\big)\\
        &=\frac{1}{n}\Tr_{H_A}\big((V_A-V_A')^*(V_A-V_A')\big)\\
        &=\norm{\big((V_A-V_A')\otimes \Id_{H_B}\big)\ket{\psi}}^2,
    \end{align*}
    where we used in the third step that $\kp$ is a maximally entangled state. An analogous calculation yields 
    \begin{equation*}
        \norm{(V_A\otimes V_B)\kp-(P_AV_A\otimes V_B)\kp}^2=\frac{1}{n}\Tr_{H_A}\left((V_A-P_AV_A)^*(V_A-P_AV_A)\right),
    \end{equation*}
    so by Lemma \ref{lem:isom-est}, we have
    \begin{equation*}
        \norm{\big((V_A-V_A')\otimes \Id_{H_B}\big)\ket{\psi}}\leq 2\sqrt{2}\epsilon,
    \end{equation*}
    and
    \begin{equation*}
        \sum_{a}\norm{((V_A-V_A')\otimes V_B')(A^x_a\otimes \Id_{H_B})\ket{\psi}}^2\leq 8\epsilon^2.
    \end{equation*}
    
    Analogously, we get
    \begin{align}
        \norm{\big(\Id_{H_A}\otimes (V_B-V_B')\big)\ket{\psi}}\leq 2\sqrt{2}\epsilon \text{ and }\label{eq:part-iso-close-to-iso-on-me-state}\\
        \sum_{a}\norm{(V_A\otimes (V_B-V_B'))(A^x_a\otimes \Id_{H_B})\ket{\psi}}^2\leq 8\epsilon^2.\nonumber
    \end{align}
    This in turn implies that
    \begin{equation*}
        \left(\expect{x\sim \nu_A}\sum_{a}\norm{(V_A\otimes V_B-V_A'\otimes V_B')(A^x_a\otimes \Id_{H_B})\ket{\psi}}^2\right)^{\frac{1}{2}}\leq 4\sqrt{2}\epsilon
    \end{equation*}
    by the triangle inequality, and another application of it yields \Cref{eq:est-2-lem-app-part-iso-constr}. 

    Now assume that $\kpt\otimes \ket{aux}$ is maximally entangled. Let $m=\dim(\tH_A\otimes K_A')$. Then for any $X: H_A\rightarrow \tH_A\otimes K_A$ we have
    \begin{align*}
        \norm{(X\otimes \Id_{H_B})\kp}^2&=\Tr_{\tH_A\otimes K_A\otimes H_B}\left((X\otimes \Id_{H_B})\kp\bra{\psi}(X\otimes \Id_{H_B})^*\right)\\
        &=\frac{1}{n}\Tr_{\tH_A\otimes K_A}\left(XX^*\right)\\
        &\geq\frac{m}{n}\Tr_{\tH_A\otimes K_A\otimes \tH_B\otimes K_B}\left((XX^*\otimes \Id_{\tH_B\otimes K_B})(\kpt\otimes\ket{aux})(\bra{\tilde{\psi}}\otimes\bra{aux})\right)\\
        &=\frac{m}{n}\norm{(X^*\otimes \Id_{\tH_B\otimes K_B})\kpt\otimes \ket{aux}}^2.
    \end{align*}
    To use the above result, we will need an estimate relating $m$ and $n$. Let $\bflambda$ and $\bfmu$ be the sequences of Schmidt coefficients of $\kp$ and $\kpt\otimes\ket{aux}$, respectively. By Lemma \ref{lemma:Schmidt-bound}, observing that local isometries preserve the Schmidt coefficients, and the assumptions in Lemma \ref{lem:partial-iso}, we know that
    \begin{equation*}
        \epsilon^2\geq \norm{\bflambda-\bfmu}_2^2.
    \end{equation*}
    Since both $\kp$ and $\kpt\otimes\ket{aux}$ are maximally entangled on Hilbert spaces of dimension $n$ and $m$, respectively, we find that (see also the proof of Lemma  \ref{lem:contained-ME-strat-dim-est})
    \begin{equation*}
        \norm{\bflambda-\bfmu}_2^2\geq \max(\frac{m-n}{m},\frac{n-m}{n}),
    \end{equation*}
    so $(1-\epsilon^2)m\leq n\leq (1-\epsilon^2)^{-1}m$.
    
    Consequently,
\begin{equation}\label{eq:adj-part-iso-close-to-adj-iso-on-me-state}
      \norm{\big((V_A^*-(V_A')^*)\otimes \Id_{\wtd{H}_B\otimes K_B}  \big)\ket{\wtd{\psi}}\otimes\ket{aux}}\leq \frac{1}{1-\epsilon^2}\norm{\big((V_A-V_A')\otimes \Id_{H_B}\big)\ket{\psi}}\leq \frac{2\sqrt{2}\epsilon}{1-\epsilon^2}. 
\end{equation}
So for the ``Moreover" part, we have
    \begin{align*}
         \norm{(\Id_{H_A}\otimes V_B')\ket{\psi}&-\big((V_A')^*\otimes \Id_{\wtd{H}_B\otimes K_B}  \big)\ket{\wtd{\psi}}\otimes\ket{aux}}\\
         &\leq \norm{\big((V_A^*-(V_A')^*)\otimes \Id_{\wtd{H}_B\otimes K_B}  \big)\ket{\wtd{\psi}}\otimes\ket{aux}}\\
         &+\norm{(V_A^*\otimes \Id_{\wtd{H}_B\otimes K_B})\big((V_A\otimes V_B)\ket{\psi}-\ket{\wtd{\psi}}\otimes\ket{aux}\big)}\\
         &+\norm{\big(\Id_{H_A}\otimes (V_B-V_B')\big)\ket{\psi}}\\
         &\leq (1+2\sqrt{2}+\frac{2\sqrt{2}}{1-\epsilon^2})\epsilon,
    \end{align*}
using that $V_A^*$ is a contraction, together with \Cref{eq:part-iso-close-to-iso-on-me-state} and \Cref{eq:adj-part-iso-close-to-adj-iso-on-me-state}. Since the above estimate becomes trivial if $\epsilon>0.29$, we conclude that
\begin{equation*}
    \norm{(\Id_{H_A}\otimes V_B')\ket{\psi}-\big((V_A')^*\otimes \Id_{\wtd{H}_B\otimes K_B}  \big)\ket{\wtd{\psi}}\otimes\ket{aux}}\leq 7\epsilon,
\end{equation*}
proving \Cref{eq:est-3-lem-app-part-iso-constr}. 
\end{proof}
Now we are ready to prove part (a) of \Cref{lem:PME-robust-self-testing-vNA}.
\begin{proof}[Proof of \Cref{lem:PME-robust-self-testing-vNA}, part (a)]
Let $\ket{\wtd{\psi}}\in\wtd{H}\otimes\wtd{H}$ and $\ket{\psi}\in H\otimes H$ be maximally entangled states such that $\tau^M(x)=\bra{\wtd{\psi}}x\otimes \Id_H\ket{\wtd{\psi}}$ for all $x\in M\subset  B(\wtd{H})$ and $\tau^N(x)=\bra{\psi}x\otimes \Id_{\wtd{H}}\ket{\psi}$ for all $x\in N\subset  B(H)$. Since $\wtd{\mcS}$ is an $(\epsilon,\nu)$-local dilation of $\mcS$, by \Cref{lem:me-aux}, there are isometries $V_A:H\arr\wtd{H}\otimes K_A$ and $V_B:H\arr \wtd{H}\otimes K_B$ and a maximally entangled state $\ket{aux}\in K_A'\otimes K_B'\subset K_A\otimes K_B$ such that
\begin{align*}
          \norm{(V_A\otimes V_B)\ket{\psi}-\ket{\wtd{\psi}}\otimes\ket{aux}}&\leq 3\epsilon,\\
        \left(\expect{x\sim \nu_A }\sum_{a}  \norm{(V_A\otimes V_B)(A^x_a\otimes \Id)\ket{\psi} - \big((\wtd{A}^x_a\otimes \Id)\ket
        {\wtd{\psi}}\big)\otimes\ket{aux}}^2    \right)^{1/2}&\leq 3\epsilon.
\end{align*}
Then by \Cref{lem:partial-iso}, there exist partial isometries $V_A':H\rightarrow \tilde{H}\otimes K_A'$ and $V_B':H\rightarrow \tilde{H}\otimes K_B'$ such that
     \begin{align*}
          \norm{(V_A'\otimes V_B')\ket{\psi}-\ket{\wtd{\psi}}\otimes\ket{aux}}&\leq (3+6\sqrt{2})\epsilon,\\
        \left(\expect{x\sim \nu_A }\sum_{a}  \norm{(V_A'\otimes V_B')(A^x_a\otimes \Id_{H})\ket{\psi} - \big((\wtd{A}^x_a\otimes \Id_{\tilde{H}})\ket
        {\wtd{\psi}}\big)\otimes\ket{aux}}^2    \right)^{1/2}&\leq (3+12\sqrt{2})\epsilon,\\
         \norm{(\Id_{H}\otimes V_B')\ket{\psi}-\big((V_A')^*\otimes \Id_{\wtd{H}\otimes K_B}  \big)\ket{\wtd{\psi}}\otimes\ket{aux}}&\leq 21\epsilon,\\
         1-\norm{(V_A'\otimes \Id_{H})\ket{\psi}}^2 &\leq 36\epsilon^2,\\
    1-\norm{\big((V_A')^*\otimes \Id_{\wtd{H}\otimes K_B'}  \big)\ket{\wtd{\psi}}\otimes\ket{aux}}^2 &\leq 9\epsilon^2.
     \end{align*}
Let $M_0:= B(K_A')\subset B(\ell^2(\N))$ and $P\in (M\otimes M_0)^{\infty}$ be a projection such that $N\cong P(M\otimes M_0)^{\infty}P$. After identifying $N$ and $P(M\otimes M_0)^{\infty}P$, we define $W=(V_A')^*\in P(M\otimes M_0)^{\infty}I_{M\otimes M_0}$.
 
Since $\ket{\wtd{\psi}}\otimes\ket{aux}$ is an maximally entangled state in $\wtd{H}\otimes K_A'$, and $\ket{\psi}$ is a maximally entangled state in $H\otimes H$, we have
\begin{align*}
    \tau^N(P-WW^*)=1-\norm{(V_A'\otimes \Id_{H})\ket{\psi}}^2 &\leq 36\epsilon^2,\\
    \tau^{M\otimes M_0}(I_{M\otimes M_0}-W^*W)=1-\norm{\big((V_A')^*\otimes \Id_{\wtd{H}\otimes K_B'}  \big)\ket{\wtd{\psi}}\otimes\ket{aux}}^2 &\leq 9\epsilon^2,
\end{align*}
and
\begin{align*}
    &\left(\expect{x\sim \nu_A }\sum_{a}\norm{\wtd{A}^x_a\otimes I_{M_0}-W^*A^x_aW}^2_{\tau^{M\otimes M_0}} \right)^{1/2}\\=&\left(\expect{x\sim \nu_A }\sum_{a}  \norm{\big(V_A'A^x_a(V_A')^*\otimes \Id_{\wtd{H}\otimes K_A'}\big)\ket{\wtd{\psi}}\otimes\ket{aux} - \big((\wtd{A}^x_a\otimes \Id_{\tilde{H}})\ket
        {\wtd{\psi}}\big)\otimes\ket{aux}}^2    \right)^{1/2} \\
        \leq&\left(\expect{x\sim \nu_A }\sum_{a}  \norm{\big(V_A'A^x_a\otimes V_B'\big)\ket{\psi} - \big((\wtd{A}^x_a\otimes \Id_{\tilde{H}})\ket
        {\wtd{\psi}}\big)\otimes\ket{aux}}^2    \right)^{1/2} +21\epsilon\\
        \leq& (24+12\sqrt{2})\epsilon,
\end{align*}
where the first inequality uses that $\sum_a A^x_a(V_A')^*V_A'A^x_a\leq \Id$. We conclude that $\wtd{\mcS}$ is a local $\big((24+12\sqrt{2})^2\epsilon^2,\nu_A \big)$-vNA-dilation of $\mcS$. Rounding $(24+12\sqrt{2})^2$ to 1700 gives the desired results.
\end{proof}

For the proof of the second part of Lemma \ref{lem:PME-robust-self-testing-vNA} we first need the following lemma.

\begin{lemma}\label{lem:ME-strat-dim-est}
    Let $M$ and $N$ be von Neumann algebras with tracial states $\tau^M$ and $\tau^N$. Suppose that there is a projection $P\in M^{\infty}$ with finite trace such that $N\cong PM^{\infty}P$ and $\tau^N=(\tau^{\infty}(P))^{-1}\tau^{\infty}$ after identification of $N$ and $PM^{\infty}P$. If there exist a $\delta<1$ and a partial isometry $w\in PM^{\infty}I_M$ such that
	\begin{equation*}
		\tau^M(I_M-w^*w)\leq \delta\text{ and } \tau^N(P-ww^*)\leq \delta,
	\end{equation*}
	then $1-\delta\leq \tau^{\infty}(P)\leq (1-\delta)^{-1}$.
\end{lemma}
\begin{proof}
    We show that $(1-\delta)\leq \tau^{\infty}(P)$. The other inequality is analogous. We compute
    \begin{align*}
        \delta&\geq \tau^M(I_M-w^*w)\\
        &=1-\tau^{\infty}(w^*w)\\
        &=1-\tau^{\infty}(P)\tau^N(ww^*)\\
        &=1-\tau^{\infty}(P)+\tau^{\infty}(P)\tau^N(P-ww^*)\\
        &\geq 1-\tau^{\infty}(P),
    \end{align*}
    so $(1-\delta)\leq \tau^{\infty}(P)$.
\end{proof}

\begin{proof}[Proof of \Cref{lem:PME-robust-self-testing-vNA}, part (b)]
Let $\kpt$ and $\kp$ be GNS states for $(M,\tau^M)$ and $(N,\tau^N)$, respectively. Since $\wtd{\mcS}$ is a local $(\epsilon,\nu_A )$-vNA-dilation of $\mcS$, there exist a finite dimensional von Neumann algebra $M_0$ with tracial state $\tau^{M_0}$, a projection $P\in (M\otimes M_0)^{\infty}$ of finite trace such that $N\cong P(M\otimes M_0)^{\infty}P$ and $\tau^N=\tau^{\infty}/\tau^{\infty}(P)$, and a partial isometry $W\in P(M\otimes M_0)^{\infty}I_{M\otimes M_0}$ such that 
\begin{equation}\label{eq:lem-PME-robust-self-testing-vNA-b-input-1}
    \expect{x\sim \nu_A }\sum_a\norm{\wtd{A}^x_a\otimes I_{M_0}-W^* A^x_a W}^2_{\tau^{M\otimes M_0}}\leq \epsilon
\end{equation}
and
\begin{equation}\label{eq:lem-PME-robust-self-testing-vNA-b-input-2}
    \tau^N(P-WW^*)\leq \epsilon, \tau^{M\otimes M_0}(I_{M\otimes M_0}-W^*W)\leq \epsilon.
\end{equation}
Our first step is to turn $W^*$ into an isometry. Let $P_1=P-WW^*$ and let $\check{M}_1\subset (\Id-I_{M\otimes M_0})(M\otimes M_0)^{\infty}(\Id-I_{M\otimes M_0})$ be a finite dimensional von Neumann algebra of the form $M\otimes B(H_1)$ such that $\dim P_1NP_1\leq \dim \check{M}_1$. Let $M_1=B(H_1)$. We can now choose a partial isometry $W_1\in P_1(M\otimes M_0)^{\infty}I_{\check{M}_1}$ such that $W_1W_1^*=P_1$. Define $V=W^*+W_1^*$, which now satisfies $P=V^*V$. Observe that
\begin{align*}
    \tau^{M\otimes M_0}(x)&=\tau^{\infty}((I_{M\otimes M_0}+I_{\check{M}_1})I_{M\otimes M_0}xI_{M\otimes M_0}(I_{M\otimes M_0}+I_{\check{M}_1})) \text{ for } x\in M\otimes M_0 \text{ and}\\
    \tau^N(x)&=\frac{1}{\tau^{\infty}(P)}\tau^{\infty}((I_{M\otimes M_0}+I_{\check{M}_1})VxV^*(I_{M\otimes M_0}+I_{\check{M}_1})) \text{ for } x\in P(M\otimes M_0)^{\infty}P.
\end{align*}
If we now perform the GNS construction using the state 
\begin{equation*}
    x\mapsto \tau^{\infty}(I_{\check{M}_1}xI_{\check{M}_1})
\end{equation*}
on $\check{M}_1$, we obtain a cyclic vector $\ket{\phi}\in (\tH\otimes H_1)\otimes (\tH\otimes H_1)$ such that 
\begin{align*}
    \tau^{\infty}(I_{\check{M}_1}xI_{\check{M}_1})=\bra{\phi}(x\otimes \Id_{\tH\otimes H_1})\ket{\phi}=\bra{\phi}( \Id_{\tH\otimes H_1}\otimes x^T)\ket{\phi}.
\end{align*} 
Note that $\ket{\phi}$ is \emph{not} a unit vector, but we will still use bra-ket notation for convenience. Let $\ket{\psi_0}\in H_0\otimes H_0$ be a GNS state for $(M_0,\tau^{M_0})$. All in all, this means that
\begin{equation*}
    \tau^{\infty}(x)=(\bra{\tilde{\psi}}\otimes \bra{\psi_0}+\bra{\phi})(x\otimes \Id_{\tH\otimes (H_0\oplus H_1)})(\kpt\otimes\ket{\psi_0}+\ket{\phi})
\end{equation*}
for all $x\in M\otimes (M_0\oplus M_1)$. Since
\begin{equation*}
    x\mapsto \tau^{\infty}(P)\bra{\psi}(V^*\otimes \overline{V}^*)\left(x\otimes \Id_{\tH\otimes(H_0\otimes H_1)}\right)(V\otimes \overline{V})\kp
\end{equation*}
and 
\begin{equation*}
    x\mapsto (\bra{\tilde{\psi}}\otimes \bra{\psi_0}+\bra{\phi})\left(VV^*xVV^*\otimes \Id_{\tH\otimes (H_0\oplus H_1)}\right)(\kpt\otimes\ket{\psi_0}+\ket{\phi})
\end{equation*}
implement the same positive linear functional on $M\otimes (M_0\oplus M_1)$ and both
\begin{equation*}
    (V\otimes \overline{V})\kp,\left(VV^*\otimes \Id_{\tH\otimes (H_0\oplus H_1)}\right)(\kpt\otimes\ket{\psi_0}+\ket{\phi})\in (\tH\otimes (H_0\oplus H_1))\otimes (\tH\otimes (H_0\oplus H_1)),
\end{equation*}
we know by the uniqueness of the GNS construction that there exists a unitary $U\in B(\tH\otimes(H_0\oplus H_1))$ such that
\begin{equation*}
    \left(\Id_{\tH\otimes (H_0\oplus H_1)}\otimes U\right)(V\otimes \overline{V})\kp=\left(VV^*\otimes \Id_{\tH\otimes (H_0\oplus H_1)}\right)(\kpt\otimes\ket{\psi_0}+\ket{\phi}).
\end{equation*}
We now aim to show that $\left(VV^*\otimes \Id_{\tH\otimes (H_0\oplus H_1)}\right)(\kpt\otimes\ket{\psi_0}+\ket{\phi})$ is close to $\kpt\otimes\ket{\psi_0}$. Note that
\begin{equation*}
    VV^*=(W^*+W_1^*)(W+W_1)=W^*W+W_1^*W_1.
\end{equation*}
Furthermore, we have
\begin{equation*}
    \left(W\otimes \Id_{\tH\otimes (H_0\oplus H_1)}\right)\ket{\phi}=0=\left(W_1\otimes \Id_{\tH\otimes (H_0\oplus H_1)}\right)\kpt\otimes \ket{\psi_0},
\end{equation*}
so 
\begin{align*}
    \left(VV^*\otimes \Id_{\tH\otimes (H_0\oplus H_1)}\right)(\kpt\otimes\ket{\psi_0}+\ket{\phi})=&\left(W^*W\otimes \Id_{\tH\otimes (H_0\oplus H_1)}\right)\kpt\otimes \ket{\psi_0}\\
    &+\left(W_1^*W_1\otimes \Id_{\tH\otimes (H_0\oplus H_1)}\right)\ket{\phi}.
\end{align*}
We then compute
\begin{align*}
    \bra{\phi}\left(W_1^*W_1\otimes \Id_{\tH\otimes (H_0\oplus H_1)}\right)^2\ket{\phi}&=\tau^{\infty}(W_1^*W_1)\\
    &=\tau^{\infty}(P)\tau^N(W_1W_1^*)\\
    &=\tau^{\infty}(P)\tau^N(P-WW^*)\leq \tau^{\infty}(P)\epsilon
\end{align*}
and 
\begin{align*}
    \bra{\tilde{\psi}}\otimes \bra{\psi_0}\left((\Id_{\tH\otimes (H_0\oplus H_1)}-W^*W)\otimes \Id_{\tH\otimes (H_0\oplus H_1)}\right)^2\kpt\otimes \ket{\psi_0}&=\tau^{\infty}((I_{M\otimes M_0}-W^*W)^2)\\
    &\leq \epsilon.
\end{align*}
Consequently, 
\begin{equation*}
    \norm{\left(\Id_{\tH\otimes (H_0\oplus H_1)}\otimes U\right)(V\otimes \overline{V})\kp-\kpt\otimes\ket{\psi_0}}\leq \sqrt{\epsilon}+\sqrt{\tau^{\infty}(P)\epsilon}\leq 3\sqrt{\epsilon},
\end{equation*}
where we use Lemma \ref{lem:ME-strat-dim-est} for the last inequality. Define $V_A=V$ and $V_B=U\overline{V}$, so
\begin{equation*}
    \norm{(V_A\otimes V_B)\kp-\kpt\otimes\ket{\psi_0}}\leq 3\sqrt{\epsilon}.
\end{equation*}

By the introduction of $\kpt$ and $\ket{\psi_0}$,  \Cref{eq:lem-PME-robust-self-testing-vNA-b-input-1} becomes
\begin{equation*}
    \left(\expect{x\sim \nu_A }\sum_a\norm{(\wtd{A}^x_a\otimes \Id_{H_0}\otimes \Id_{\tH\otimes H_0})\kpt\otimes \ket{\psi_0}-(W^* A^x_a W\otimes \Id_{\tH\otimes H_0})\kpt\otimes \ket{\psi_0}}^2\right)^{1/2}\leq \sqrt{\epsilon}.
\end{equation*}
Since 
\begin{equation*}
    \sum_a (W^*A_a^xW)^2\leq 1,
\end{equation*}
we know that
\begin{equation*}
    \left(\sum_a\norm{\left(W^* A^x_a W\otimes \Id_{\tH\otimes (H_0\oplus H_1)}\right)(\kpt\otimes \ket{\psi_0}-(V_A\otimes V_B)\kp)}^2\right)^{1/2}\leq 3\sqrt{\epsilon}.
\end{equation*}
The last step is to find an estimate for
\begin{align*}
    \sum_a\norm{\left((W^* A^x_a WV_A-V_AA_a^x)\otimes V_B\right)\kp}^2&\\
    =\sum_a&\tau^N\Big((W^* A^x_a WV_A-V_AA_a^x)^*(W^* A^x_a WV_A-V_AA_a^x)\Big).
\end{align*}
We compute
\begin{align*}
   \norm{((W^* A^x_a WV_A-V_AA_a^x)\otimes V_B)\kp}^2=&\tau^N\Big(V_A^*W^*A_a^xWW^*A_a^xWV_A+A_a^xV_A^*V_AA_a^x\\
    &-A_a^xV_A^*W^*A_a^xWV_A-V_A^*W^*A_a^xWV_AA_a^x\Big)\\
    =&\tau^N\Big(WW^*A_a^xWW^*A_a^x+(A_a^x)^2\\
    &-A_a^xWW^*A_a^xWW^*-WW^*A_a^xWW^*A_a^x\Big)\\
    =&\tau^N\left((A_a^x)^2-WW^*A_a^xWW^*A_a^x\right)\\
    \leq& \tau^N\left((P-WW^*)(A_a^x)^2\right)\\
    &+\tau^N\left((P-WW^*)A_a^xWW^*A_a^x\right)\\
    \leq& 2\tau^N\left((P-WW^*)(A_a^x)^2\right).
\end{align*}
Consequently, 
\begin{align*}
    \sum_a\norm{(W^* A^x_a WV_A-V_AA_a^x)\otimes V_B)\kp}^2&\leq 2\sum_a\tau^N\left((P-WW^*)(A_a^x)^2\right)\\
    &\leq 2\tau^N(P-WW^*)\leq 2\epsilon.
\end{align*}
By the triangle inequality, we have shown that
\begin{equation*}
    \left(\expect{x\sim \nu_A }\sum_a\norm{\left(\wtd{A}^x_a\otimes \Id_{H_0}\otimes \Id_{\tH\otimes H_0}\right)\kpt\otimes \ket{\psi_0}-(V_AA_a^x\otimes V_B)\kp}^2\right)^{1/2}\leq (4+\sqrt{2})\sqrt{\epsilon}.
\end{equation*}
Analogously, we obtain
\begin{equation*}
    \left(\expect{x\sim \nu_A }\sum_a\norm{\Id_{\tH\otimes H_0}\otimes((\wtd{A}^x_a)^T\otimes \Id_{H_0})\kpt\otimes \ket{\psi_0}-(V_A\otimes V_B(A_a^x)^T)\kp}^2\right)^{1/2}\leq (4+\sqrt{2})\sqrt{\epsilon},
\end{equation*}
proving the lemma.
\end{proof}

\paragraph{Author statements.} Every author contributed equally to the paper. No LLMs or other AI models were used in this project.

\printbibliography
\end{document}